\documentclass[aos]{imsart}
\usepackage{varioref} 
\usepackage{xr-hyper} 
\usepackage{hyperref}     
\usepackage{booktabs}
\RequirePackage[OT1]{fontenc}
\usepackage{amsthm,epsfig,amsmath,amssymb,euscript,hyperref}
\usepackage{natbib} 
\usepackage{array}
\usepackage{nccbbb}
\usepackage{verbatim} 
\usepackage[lined,boxed,commentsnumbered, ruled,resetcount,algosection]{algorithm2e} 
\usepackage{graphicx}
\usepackage{bm}
\usepackage{bbm} 
\usepackage{dsfont} 
\usepackage{enumitem}
\usepackage{float}
\usepackage{rotating}
\usepackage{placeins} 
\usepackage{float} 
\usepackage[usenames,dvipsnames]{xcolor}
\usepackage{mathabx} 
\usepackage{xcolor}
\usepackage{marginnote}
\usepackage{mathbbol}  
\RequirePackage[colorlinks,citecolor=blue,urlcolor=blue]{hyperref}
\usepackage{pdflscape}
\usepackage[24hr]{datetime}
\usepackage{tikz}
\usetikzlibrary{matrix}
\usepackage[cal=boondox]{mathalfa}
\usepackage{accents} 

\usepackage{mathtools}
\makeatletter
\DeclareRobustCommand\widecheck[1]{{\mathpalette\@widecheck{#1}}}
\def\@widecheck#1#2{%
    \setbox\z@\hbox{\m@th$#1#2$}%
    \setbox\tw@\hbox{\m@th$#1%
       \widehat{%
          \vrule\@width\z@\@height\ht\z@
          \vrule\@height\z@\@width\wd\z@}$}%
    \dp\tw@-\ht\z@
    \@tempdima\ht\z@ \advance\@tempdima2\ht\tw@ \divide\@tempdima\thr@@
    \setbox\tw@\hbox{%
       \raise\@tempdima\hbox{\scalebox{1}[-1]{\lower\@tempdima\box
\tw@}}}%
    {\ooalign{\box\tw@ \cr \box\z@}}}
\makeatother

\newcolumntype{L}[1]{>{\raggedright\let\newline\\\arraybackslash\hspace{0pt}}m{#1}}
\newcolumntype{C}[1]{>{\centering\let\newline\\\arraybackslash\hspace{0pt}}m{#1}}
\newcolumntype{R}[1]{>{\raggedleft\let\newline\\\arraybackslash\hspace{0pt}}m{#1}}
\newcolumntype{H}{>{\setbox0=\hbox\bgroup}c<{\egroup}@{}}   

\startlocaldefs
\numberwithin{equation}{section}   
\theoremstyle{plain}
\newtheorem{theorem}{Theorem}[section]
\newtheorem{lemma}[theorem]{Lemma}
\newtheorem{corollary}[theorem]{Corollary}
\newtheorem{proposition}[theorem]{Proposition}

\newtheorem{condition}{Condition}
\newtheorem{definition}{Definition}
\theoremstyle{definition}
\newtheorem{example}{Example}[section]
\newtheorem{remark}{Remark}[section]
\endlocaldefs
\newcommand{\ignore}[1]{}

\definecolor{BUred}{rgb}{0.8, 0.0, 0.0}
\DeclareMathOperator{\pr}{\mathsf{P}}

\DeclareMathOperator{\E}{\mathsf{E}}
\DeclareMathOperator{\Var}{\mathsf{Var}}
\DeclareMathOperator{\Bias}{\mathsf{Bias}}
\DeclareMathOperator{\MSE}{\mathsf{MSE}}
\DeclareMathOperator{\Cov}{\mathsf{Cov}}
\DeclareMathOperator{\Corr}{\mathsf{Corr}}

\def\T{{ \mathrm{\scriptscriptstyle T} }}  
\DeclareMathOperator{\tr}{tr}
\DeclareMathOperator{\C}{c}
\DeclareMathOperator{\obs}{obs}
\DeclareMathOperator{\mis}{mis}
\DeclareMathOperator{\com}{com}

\DeclareMathOperator{\diag}{diag}

\DeclareMathOperator{\vech}{vech}
\newcommand{\Bern}{\textnormal{Bern}}

\newcommand{\Normal}{\textnormal{N}}

\newcommand{\Ga}{\textnormal{Ga}}

\newcommand{\Po}{\textnormal{Po}}

\DeclareMathOperator*{\sumPrime}{\sum{}^{\prime}}
\DeclareMathOperator{\expit}{expit}

\DeclareMathOperator{\lrt}{L}
\DeclareMathOperator{\wt}{W}
\DeclareMathOperator{\rt}{R}
\DeclareMathOperator{\Jack}{Jack}
\DeclareMathOperator{\Full}{Full}
\DeclareMathOperator{\Pair}{Pair}

\newcommand{\df}{q}
\newcommand{\inP}{  \overset{ \mathrm{pr}  }{\rightarrow} }
\newcommand{\inD}{\Rightarrow}
\newcommand{\simIID}{   \overset{ \text{iid}   }{\sim} }

\makeatletter
\newcommand{\BIG}{\bBigg@{4}}
\newcommand{\BIGG}{\bBigg@{5}}
\makeatother

\begin{document}
\externaldocument[]{SMI_aos_supp}
\begin{frontmatter}
\title{General and Feasible Tests with Multiply-Imputed Datasets} 
\runtitle{Tests with Imputed Datasets}
%----------------------------------------------------------------------------------------------------
\begin{aug}
\author[A]{\fnms{Kin Wai} \snm{Chan}\ead[label=e1]{kinwaichan@cuhk.edu.hk}}
\address[A]{Department of Statistics, The Chinese University of Hong Kong, \printead{e1}}
\end{aug}
%----------------------------------------------------------------------------------------------------
\begin{abstract}
Multiple imputation (MI) is a technique especially designed for handling missing data
in public-use datasets. 
It allows analysts to perform incomplete-data inference straightforwardly
by using several already imputed datasets released by the dataset owners.
However, the existing MI tests 
require either 
a restrictive assumption on the missing-data mechanism, known as equal odds of missing information (EOMI), 
or an infinite number of imputations. 
Some of them also require analysts to have access to 
restrictive or non-standard computer subroutines. 
Besides, the existing MI testing procedures cover only
Wald's tests and likelihood ratio tests but not Rao's score tests, 
therefore, these MI testing procedures are not general enough. 
In addition, the MI Wald's tests and MI likelihood ratio tests are not procedurally identical, 
so analysts need to resort to distinct algorithms for implementation. 
In this paper, we propose a general MI procedure, called stacked multiple imputation (SMI),
for performing Wald's tests, likelihood ratio tests and Rao's score tests
by a unified algorithm. 
SMI requires neither EOMI nor an infinite number of imputations. 
It is particularly feasible for analysts as 
they just need to use a complete-data testing device  
for performing the corresponding incomplete-data test. 
\end{abstract}
%----------------------------------------------------------------------------------------------------
%----------------------------------------------------------------------------------------------------
\end{frontmatter}
%----------------------------------------------------------------------------------------------------

\section{New thoughts on the old results}\label{sec:new_and_old}
\subsection{Introduction}
Missing data are usually encountered in real-data analysis,
both in observational and experimental studies.
Statistical inference of incomplete datasets is harder than that of complete datasets. 
Multiple imputation (MI), proposed by \citet*{rubin1978}, is one of the most popular ways of 
handling missing data. 
This method requires specifying an imputation model for filling in 
the missing data multiple times so that 
standard complete-data procedures can be straightforwardly applied to each of 
the imputed datasets; 
see Sections \ref{sec:bkgd} and \ref{sec:existingWork} for a review.  
Also see \citet{rubin1987}, \cite{CarpenterKenward2013}, and \cite{KimShao2013} for an introduction.
Although a fairly complete theory for performing MI tests is available, 
all existing results suffer from at least one of following three problems:
(i) reliance on strong statistical assumptions, 
(ii) requirement of infeasible computer subroutines, and 
(iii) lack of unified combining rules for various types of tests. 
We will discuss these problems thoroughly in Section \ref{sec:strengthWeak}.

The major goal of this paper is to derive a handy MI test 
that resolves the aforementioned problems.
This paper is structured as follows. 
In the remaining part of this section, 
the existing MI tests are reviewed. 
We also discuss their pros and cons.
In Section \ref{sec:plan}, we motivate the proposal test statistics and 
present the plan to achieve our proposed test. 
In section \ref{sec:method}, 
our proposed methodology and principle are discussed.
In Section \ref{sec:estimation}, 
a novel theory for 
estimating the odds of missing information is presented.  
In Section \ref{sec:testing}, 
a new test based on multiply-imputed datasets 
is derived.  
In Section \ref{sec:applications}, applications and simulation experiments are illustrated. 
In Section \ref{sec:conclusion}, we conclude the paper and discuss possible future work.
Proofs, auxiliary results, additional simulation results and 
an R-package \texttt{stackedMI} 
are included as supplementary materials.

\subsection{Background and problem setup}\label{sec:bkgd}
Let $X$ be a dataset in the form of a $n\times p$ matrix consisting of $n$ rows of independent units. 
Assume that $X$ is generated from a probability density $f(X\mid \psi)$, where 
$\psi \in \Psi$ is the model parameter.
The parameter of interest is 
a sub-vector of $\psi$ denoted by $\theta \in \Theta \subseteq \mathbb{R}^k$.
We are interested in testing $H_0: \theta=\theta_0$ against $H_1:\theta \neq \theta_0$
for some fixed $\theta_0 \in \Theta$.
If the complete dataset $X_{\com} := X$ is available, 
we may perform the Wald's test, likelihood ratio (LR) test, and Rao's score (RS) test.
The test statistics are 
\begin{eqnarray}
	\mathcal{d}_{\wt}(X)
		&:=& d_{\wt}(\widehat{\theta},\widehat{V}) 
		:= (\widehat{\theta}-\theta_0)^{\T}\widehat{V}^{-1}(\widehat{\theta}-\theta_0), \label{eqt:com_WT}\\
	\mathcal{d}_{\lrt}(X) 
		&:=& d_{\lrt}(\widehat{\psi}, \widehat{\psi}_0\mid X) 
		:= 2\log\{f(X\mid \widehat{\psi})/f(X\mid\widehat{\psi}_0)\},  \label{eqt:com_LRT} \\
	\mathcal{d}_{\rt}(X)
		&:=& d_{\rt}(\widehat{\psi}_0 \mid X) 
		:= u(\widehat{\psi}_0)^{\T} \{ I(\widehat{\psi}_0) \}^{-1} u(\widehat{\psi}_0 ), \label{eqt:com_RT}
\end{eqnarray}
respectively, 
where $\widehat{\theta} :=\widehat{\theta}(X)$ is the maximum likelihood estimator (MLE) of $\theta$; 
$\widehat{V} :=\widehat{V}(X)$ is a variance estimator of $\widehat{\theta}$; 
$\widehat{\psi} :=\widehat{\psi}(X)$ and $\widehat{\psi}_0:=\widehat{\psi}_0(X)$ are the unrestricted and $H_0$-restricted MLEs of $\psi$; 
$u(\psi) := u(\psi\mid X) := \partial \log f(X\mid \psi)/ \partial \psi$ is the score function;
and $I(\psi) := I(\psi\mid X) := - \partial^2 \log f(X\mid \psi)/ \partial \psi\partial \psi^{\T}$ 
is the Fisher's information.
See \citet{TSH_lehmannRomano} for more details.   
Throughout the paper, we use $\aleph\in\{\wt,\lrt,\rt\}$ 
to abbreviate the name of the test in various subscripts. 
In each of (\ref{eqt:com_WT})--(\ref{eqt:com_RT}), 
the mapping $X\mapsto\mathcal{d}_{\aleph}(X)$ is a function of the dataset $X$ only. 
We call such function $X\mapsto \mathcal{d}_{\aleph}(X)$ 
a standard testing \emph{device} (i.e., a testing \emph{subroutine} or a testing \emph{procedure}) in computer software.
The device $\mathcal{d}_{\aleph}(\cdot)$ is the only requirement for complete-data testing.  
Under standard regularity conditions (see, e.g., Section 4.4 of \cite{serfling1980}) and under $H_0$, 
we have, for any $\aleph\in\{\wt,\lrt, \rt\}$, that 
\begin{eqnarray}\label{eqt:comTestchisq}
	\mathcal{d}_{\aleph}(X)/k \;\inD\; \chi^2_k/k, 
	\qquad \text{as} \qquad
	n\rightarrow\infty,
\end{eqnarray}
where ``$\inD$'' denotes weak convergence; see Chapter 2 of \cite{vanDerVaart2000}.

If a part of $X_{\com} = \{X_{\obs}, X_{\mis}\}$ is missing such that 
only $X_{\obs}$ is available, 
testing $H_0$ is more involved. 
One widely used method is \emph{multiple imputation} (MI), which 
is a two-stage procedure.
\begin{itemize}
	\item The first stage involves an \emph{imputer} to handle the missing data.
			The imputer draws $X_{\mis}^{1}, \ldots, X_{\mis}^m$
			from the conditional distribution $[X_{\mis}\mid X_{\obs}]$
			so that the missing part $X_{\mis}$ can be filled in by $X_{\mis}^{1}, \ldots, X_{\mis}^m$
			to form $m$ completed datasets $X^{\ell} := \{X_{\obs},X_{\mis}^{\ell}\}$ ($\ell=1, \ldots,m$).
			Note that MI assumes the missing mechanism is ignorable \citep{rubin1976}.
			See Remark \ref{rem:imputer} for more discussion of this stage. 

	\item The second stage involves possibly many \emph{analysts} 
			to perform inference of their own interests. 
			Each of them receives the same completed datasets $X^{1}, \ldots, X^m$ from the imputer.
			Then, s/he can repeatedly apply some standard \emph{complete-data procedures} 
			to $X^{1}, \ldots, X^m$, and obtain $m$ preliminary results.
			The final result is obtained by appropriately combining them; 
			see Section  \ref{sec:existingWork} for a review.
\end{itemize}
MI is an attractive method because it naturally divides imputation and analysis tasks 
into two separate stages 
so that analysts do not need to be trained for handling incomplete datasets.
Indeed, it has been a very popular method in various fields; see, for example, 
\cite{tu1993JASA}, \cite{rubin1996multiple}, \cite{rubin1987}, \cite{Schafer1999}, \cite{KingHonakerJosephScheve2001}, \cite{PeughEnders2004}, \cite{KenwardCarpenter2007}, \cite{HarelZhou2007}, \cite{HortonKleinman2007}, \cite{RoseFraser2008}, \cite{HolanTothFerreiraKarr2010}, \cite{KimYang2017}, and \cite{YuChaneetal2021}.

\begin{remark}\label{rem:imputer}
MI is originally designed for handling public-use datasets. 
Hence, the imputers in stage 1 are in general different from the analysts in stage 2; 
see, for example, \cite{ParkerSchenker2007}.
Consequently, analysts cannot produce arbitrarily many imputed datasets as they wish. 
For example, 
only multiply-imputed datasets 
are released in National Health Interview Survey (NHIS) conducted in the United States; 
see \cite{SchenkerOther06}. 
In particular, only five imputed datasets are released in 2018 NHIS; see
\url{https://www.cdc.gov/nchs/nhis/nhis_2018_data_release.htm}.
Moreover, because the imputers usually belong to the organizations, 
for example, a census bureau of the government,  
who collect the data, 
they usually know better how the data are missing, 
and have auxiliary variables to impute the missing data.
So, the ignorability could be a reasonable assumption.  
See \citet{rubin1996multiple} for more comprehensive discussions of these matters.   
\end{remark}

\subsection{MI combining rules and reference null distribution}\label{sec:existingWork}
Applying the functions $\widehat{\theta}(\cdot)$, $\widehat{V}(\cdot)$, $\widehat{\psi}(\cdot)$ and $\widehat{\psi}_0(\cdot)$ to each of the imputed datasets $X^1, \ldots, X^m$, 
an analyst obtains 
$\widehat{\theta}^{\ell}:= \widehat{\theta}(X^{\ell})$, 
$\widehat{V}^{\ell}:=\widehat{V}(X^{\ell})$, 
$\widehat{\psi}^{\ell}:=\widehat{\psi}(X^{\ell})$ and  
$\widehat{\psi}_0^{\ell}:=\widehat{\psi}_0(X^{\ell})$ 
for $\ell=1,\ldots,m$. Define 
\begin{align}
	\widetilde{d}'_{\wt} &:= \frac{1}{m}\sum_{\ell=1}^m d_{\wt}(\widehat{\theta}^{\ell}, \bar{V}), &
	\widetilde{d}''_{\wt} &:= d_{\wt}(\bar{\theta}, \bar{V}), & \label{eqt:dp_dpp_w_l} \\
	\widetilde{d}'_{\lrt} &:= \frac{1}{m}{\sum_{\ell=1}^m d_{\lrt}(\widehat{\psi}^{\ell}, \widehat{\psi}^{\ell}_0 \mid X^{\ell})}, &
	\widetilde{d}''_{\lrt} &:= \frac{1}{m}{\sum_{\ell=1}^m d_{\lrt}(\bar{\psi}, \bar{\psi}_0 \mid X^{\ell})},  &\label{eqt:dp_dpp_w_l_LER}
\end{align}
where $\bar{\theta} := \sum_{\ell=1}^m \widehat{\theta}^{\ell}/m$,
and $\bar{V}$, $\bar{\psi}$, $\bar{\psi}_0$ are similarly defined. 
The MI Wald's statistic \citep{li91JASA} and MI LR statistic \citep{mengRubin92} are
$\widetilde{D}_{\wt}$ and $\widetilde{D}_{\lrt}$, respectively, where, 
for $\aleph\in\{\wt, \lrt\}$,
\begin{eqnarray}\label{eqt:Da}
	\qquad
	\widetilde{D}_{\aleph} := \frac{\widetilde{d}''_{\aleph}}{k\left\{ 1+\left(1+\frac{1}{m}\right)\widetilde{\mu}_{r,\aleph}\right\}} \qquad \text{and} \qquad
	\widetilde{\mu}_{r,\aleph} := \frac{\widetilde{d}'_{\aleph}-\widetilde{d}''_{\aleph}}{k(m-1)/m}.
\end{eqnarray} 
The factor $\{ 1+(1+1/m)\widetilde{\mu}_{r,\aleph}\}$ in $\widetilde{D}_{\aleph}$ is used to deflate 
$\widetilde{d}''_{\aleph}$ in order to adjust for the loss of information due to missingness. 
The LR test statistic $\widetilde{D}_{\lrt}$ may be negative. \cite*{chanMeng2017_MILRT} recently
proposed a corrected version.

If $m$ is fixed and $n\rightarrow\infty$,
the limiting distribution of $\widetilde{D}_{\aleph}$ is notoriously complicated  
because of the dependence on the unknown matrix  
$F := I_{\mis}I_{\com}^{-1}$ in a tangled way, 
where 
\begin{equation*}
	I_{\com} := \E\left\{ - \frac{\partial^2\log f(X\mid \psi)}{\partial \theta \partial \theta^{\T}} \right\} 
	\qquad \text{and} \qquad
	I_{\mis} := \E\left\{ - \frac{\partial^2\log f(X\mid X_{\obs},\psi)}{\partial \theta \partial \theta^{\T}} \right\}  \label{eqt:FI}
\end{equation*}
are the complete-data and missing-data Fisher's information of $\theta$, respectively.  
We also denote $I_{\obs} := I_{\com}-I_{\mis}$ as the observed-data Fisher's information of $\theta$. 
Let the eigenvalues of $F$ be $f_1 \geq \cdots \geq f_k$, and 
denote $r_j := f_j/(1-f_j)$ for each $j$. 
The values $f_j$ and $r_j$ are known as the \emph{fraction of missing information} (FMI) and 
the \emph{odds of missing information} (OMI), respectively.  
There is no loss in information if $r_1 = \cdots = r_k = 0$, 
or, equivalently, $f_1 = \cdots = f_k = 0$. 
Under regularity conditions (RCs) and $H_0$,  
\begin{eqnarray}\label{eqt:representation_D_limit}
	\widetilde{D}_{\aleph} \inD
		\mathbb{D} := \frac{\frac{1}{k}\sum_{j=1}^k \left\{ 1+ (1+\frac{1}{m})r_{j} \right\} G_j}{ 1+\frac{1}{k}\sum_{j=1}^k (1+\frac{1}{m})r_{j} H_j}, 
\end{eqnarray}
as $n\rightarrow\infty$,
where $G_1,\ldots, G_k \sim \chi^2_1$ and $H_1, \ldots, H_k \sim \chi^2_{m-1}/(m-1)$ are independent. 
The RCs and derivation of (\ref{eqt:representation_D_limit}) are presented in 
Proposition \ref{prop:limit_of_Dhat}. 
Since the distribution $\mathbb{D}$  
depends on the nuisance parameters $r_{1}, \ldots, r_{k}$ in a complicated way, 
it is not immediately possible to use $\mathbb{D}$ as a reference null distribution.  
In order to mitigate this fundamental difficulty, 
it has been a common practice to assume some structure on $r_{1}, \ldots, r_{k}$ 
(see Condition \ref{ass:EOMI} below),
and/or resort to asymptotic ($m\rightarrow\infty$) approximation. 
We present these existing strategies one by one, and give a summary in Table \ref{tab:approxNullDist}.

\begin{condition}[Equal odds of missing information (EOMI)]\label{ass:EOMI}
There is $\mu_r$ such that $r_{1} =\cdots=r_{k}=\mu_r$. 
\end{condition}

\begin{table}
\setlength{\tabcolsep}{3pt}
\renewcommand{\arraystretch}{0.85}
\centering
\footnotesize
\begin{tabular}{Hlllll}
\toprule
&\multicolumn{2}{c}{\bf Assumptions}  & \multicolumn{3}{c}{\bf Asymptotic null distribution of $\widetilde{D}$ or $\widehat{D}$}  \\
\cmidrule(r){2-3}\cmidrule(r){4-6}
\bf Case  & Condition \ref{ass:EOMI} & $m\rightarrow\infty$  &  \bf Exact  & \bf Approximated   & \bf Name \\
\midrule
(1) & Required& Required   & $\chi^2_k/k$ & $\chi^2_k/k$ & T1\\
\cmidrule(r){1-6}
(2) &Required & Not required & $\mathbb{D}_0$ in (\ref{eqt:approx_limit_D}) & $\approx F(k, \widehat{\df})$ \citep{li91JASA} & T2\\
\cmidrule(r){1-6} 
(3) &Not required & Required &   $\mathbb{D}_{\infty}$ in (\ref{eqt:limitingDm_mLarge}) &  $\approx \widehat{c}_0 + \widehat{c}_1 \chi^2_k/k$  \citep{mengPhDthesis}  & T3\\
\cmidrule(r){1-6}
(4) &Not required & Not required  &  $\mathbb{D}$ in (\ref{eqt:representation_D_limit}) & $\approx \widehat{\mathbb{D}}$ (Algorithm \ref{algo:SMI_exact}) & T4 (proposal)\\ 
\bottomrule
\end{tabular}
\caption{\footnotesize
Summary of the asymptotic reference null distributions. 
The reference distribution refers to 
the common weak limit ($n\rightarrow\infty$ and $m>1$) of
the existing test statistic $\widetilde{D} \in \{ \widetilde{D}_{\wt}, \widetilde{D}_{\lrt} \}$ 
defined in (\ref{eqt:Da}) and 
the proposed test statistics 
$\widehat{D}\in \{ \widehat{D}_{\wt}, \widehat{D}_{\lrt}, \widehat{D}_{\rt}\}$ to be defined in (\ref{eqt:asy_version_D_proposal}).} 
\label{tab:approxNullDist}
\end{table}

Condition \ref{ass:EOMI} is equivalent to $f_1 = \cdots = f_k$, 
known as equal fraction of missing information (EFMI).
Although it is a strong assumption and is almost always violated in real problems, 
it is widely used in the literature because of simplicity;
see \cite{rubin1987}.
Under Condition \ref{ass:EOMI},
$\mathbb{D}$ 
can be represented as 
\begin{eqnarray}\label{eqt:approx_limit_D}
	\mathbb{D}_0 := \frac{\left\{1+(1+\frac{1}{m})\mu_r \right\} G}{ 1+(1+\frac{1}{m})\mu_r H},  
\end{eqnarray}
where 
$G\sim {\chi^2_k}/{k}$ and $H\sim{\chi^2_{k(m-1)}}/\{k(m-1)\}$ are independent. 
In this case, $\mathbb{D}_0$ depends only on one unknown $\mu_r$, 
hence, approximating (\ref{eqt:approx_limit_D}) is easier. 
The first approximation of $\mathbb{D}_0$ was provided by \cite{rubin1987}.
Then \cite{li91JASA} refined it to $F(k, \widehat{\df})$, where
$\widehat{\df}$ is an estimate of 
\begin{equation*}\label{eqt:df1}
	\df := \left\{ \begin{array}{ll}
		4+(K_m-4)\left[1+(1-2/K_m)\{ (1+\frac{1}{m})\mu_r \}^{-1}\right]^2 , & \text{if $K_m>4$}; \\
		(m-1)\left[1+ \{ (1+\frac{1}{m})\mu_r \}^{-1} \right]^2(k+1)/2 ,& \text{otherwise},
	\end{array}\right. \nonumber \\
\end{equation*}
where $K_m := k(m-1) $.
In practice, $\widehat{\df}$ is constructed by plugging in 
an estimate of $\mu_r$ into $\df$. 
\citet{li91JASA} and \citet{mengRubin92} proposed to estimate $\mu_r$ by
$\widetilde{\mu}_{r,\wt}$ and $\widetilde{\mu}_{r,\lrt}$
if Wald's test and LR test are used, respectively.
Many software routines have implemented this approximation, for example, 
\cite{vanBuurenGroothuisOudshoorn2011}.
Besides, some approximations of $\mathbb{D}_0$ designed for small $n$
can be found in \cite{rubinSchenker1986}, \cite{BarnardRubin1999} and \cite{Reiter07}.

If Condition \ref{ass:EOMI} does not hold but $m\rightarrow\infty$, 
then $\mathbb{D}$ is simplified to $\mathbb{D}_{\infty}$,  
which can be represented by 
\begin{align}\label{eqt:limitingDm_mLarge}
	\mathbb{D}_{\infty} := \frac{1}{k(1+ \mu_r)}\sum_{j=1}^k (1+r_{j}) G_j , 
	\qquad \text{where} \qquad 
	\mu_r := \frac{1}{k}\sum_{j=1}^k r_{j}.
\end{align}
\cite{mengPhDthesis} and \cite{li91JASA} proposed approximating $\mathbb{D}_{\infty}$ by 
$\widehat{c}_0 + \widehat{c}_1\chi^2_k/k$ via matching their first two moments. 
The coefficients $\widehat{c}_0$ and $\widehat{c}_1$ are estimates of 
${c}_1 := \left\{1+\sqrt{\sigma^2_r/(1+\mu_r^2)} \right\}^{1/2}$
and 
${c}_0 := 1-{c}_1$,
where $\sigma^2_r := \sum_{j=1}^k (r_j - \mu_r)^2/k$.
It is remarked that \cite{mengPhDthesis} is an improvement over \cite{li91JASA}.
\cite{mengPhDthesis} proposed estimating $\mu_r$ and $\sigma^2_r$
by 
$\widetilde{\widetilde{\mu}}_{r,\wt} := \tr\{\widehat{B}\bar{V}^{-1}\}/k$ and
$\widetilde{\widetilde{\sigma}}_{r,\wt}^2 := \tr\{(\widehat{B}\bar{V}^{-1})^2\}/k - \widetilde{\widetilde{\mu}}_{r,\wt}^2\{1+ k/(m-1)\}$, 
respectively, where
$\widehat{B} :=\sum_{\ell=1}^m (\widehat{\theta}^{\ell}- \bar{\theta})(\widehat{\theta}^{\ell}- \bar{\theta})^{\T}/(m-1)$.

If both Condition \ref{ass:EOMI} and $m\rightarrow\infty$ are satisfied, 
then $\mathbb{D}\sim\chi^2_k/k$, that is, the standard reference distribution 
in (\ref{eqt:comTestchisq}). 
This very rough approximation was discussed in 
Section 3.2 of \cite{li91JASA}.

\subsection{Theoretical, implementational and universal concerns}\label{sec:strengthWeak}
Although fairly complete theories for performing MI tests are available, 
they are not fully satisfactory because of the following three concerns.

The first concern is a theoretical consideration.
As we discussed in Section \ref{sec:existingWork}, 
all existing MI tests (e.g., T1, T2 and T3 in Table \ref{tab:approxNullDist}) require 
Condition \ref{ass:EOMI} and/or $m\rightarrow \infty$.
In many real applications, they are very restrictive and can hardly be satisfied. 
Roughly speaking, Condition \ref{ass:EOMI} means 
that all parameters in $\theta = (\theta_1, \ldots, \theta_k)^{\T}$ 
are equally impacted by the missing data. 
It is violated in many applications, for example, 
linear regression (see Section \ref{sec:eg_reg}), 
testing variance-covariance matrix (see Section \ref{sec:eg_mvn} of the supplementary note), etc. 
The assumption of $m\rightarrow\infty$ does not match 
the current practice as well. 
For public-use datasets, 
the dataset owners may refuse to release a large number of imputed datasets to the public due to 
storage problem, processing inconvenience, or privacy concern, 
therefore, $m \leq 30$, or even $m\leq 10$, is typically used; 
see Remark \ref{rem:imputer} for an example. 
Moreover, unlike typical simulation study, 
the analysts cannot arbitrarily generate a large number imputed datasets.
Hence, the value of $m$ is typically not very large.

The second concern is about implementation.
MI tests are most useful if users only need to apply their 
intended complete-data testing device $X \mapsto \mathcal{d}(X)$, that is, 
either (\ref{eqt:com_WT}), (\ref{eqt:com_LRT}) or (\ref{eqt:com_RT}),
repeatedly to the imputed datasets 
$X^1, \ldots, X^m$.
Unfortunately, this ideal minimal requirement is not sufficient for most existing MI tests 
(e.g., T2, and T3 in Table \ref{tab:approxNullDist}). 
Instead of merely requiring the device $\mathcal{d}(\cdot)$, they may need 
(i) the variance-covariance matrix estimator $\widehat{V}(X)$, or
(ii) a non-standard likelihood function. 
The device (i) is required for computing, for example,  
$\widetilde{\mu}_{r,\wt}$ \citep{li91JASA}, 
and $\widetilde{\widetilde{\mu}}_{r,\wt}$ and $\widetilde{\widetilde{\sigma}}^2_{r,\wt}$ \citep{mengPhDthesis}.
It is infeasible because, in some problems, 
$\widehat{V}(X)$ is typically 
unavailable in standard computer subroutines. 
For example, 
to perform a G-test, that is, the LR test for goodness of fit in contingency tables
(see Section 6.4.3 of \citet{shan99}),
one may use the R function \texttt{GTest} in the package \texttt{DescTools}. 
But this function 
does not provide $\widehat{V}(X)$ in the output.
In some problems, computing $\widehat{V}(X)$ is highly non-trivial, for example, 
testing variance-covariance matrices; see Section \ref{sec:eg_mvn} of the supplementary note.  
The device (ii) usually requires analysts' effort to build,
so, it can be challenging or simply troublesome for them. 
For example, $(X,\psi_0,\psi_1)\mapsto d_{\lrt}(\psi_1, \psi_0\mid X)$,
instead of the standard $X\mapsto \mathcal{d}_{\lrt}(X)$,
is required in $\widetilde{\mu}_{r,\lrt}$ \citep{mengRubin92}. 
Arguably, most (if not all) LR test statistic subroutines are not built in this way. 
Although the recent work by \cite{chanMeng2017_MILRT} 
only requires $X\mapsto \mathcal{d}_{\lrt}(X)$ for performing MI tests, 
it assumes equal OMI and is restricted to LR tests.
In order to handle unequal OMI, we need a substantially more sophisticated principle and technique, 
which are completely novel and have never been discussed in the literature.
 
The third concern is about universality. 
Computing $\widetilde{D}_{\wt}$ and $\widetilde{D}_{\lrt}$ require different algorithms
as the functional forms of
$(\widetilde{d}'_{\wt}, \widetilde{d}''_{\wt})$ 
in (\ref{eqt:dp_dpp_w_l})
and $(\widetilde{d}'_{\lrt}, \widetilde{d}''_{\lrt})$ 
in (\ref{eqt:dp_dpp_w_l_LER}) are different.
It is inconvenient for users. 
In addition, 
the existing MI procedures only cover
Wald's test and LR test but not RS test. 
Since many tests are RS tests in nature
(see, e.g., \cite{beraBilias2001}), 
it reveals a gap between MI testing theory and practical usage.  
As discussed in \citet{rao2005},
there are many reasons to use RS test. 
For example,
RS test does not require fitting the full models, 
which are non-identifiable or computationally intensive in some problems; 
see Section \ref{sec:logisticEG} of the supplementary note for an example. 
So, it is desirable to have a unified MI procedure for all tests.

This paper addresses these three problems.  
A general, unified and feasible MI test without requiring 
Condition \ref{ass:EOMI} or $m\rightarrow \infty$ is proposed. 
It only requires the analysts to have a standard complete-data test device $X\mapsto \mathcal{d}(X)$, 
where $\mathcal{d}(\cdot)$ can be either $\mathcal{d}_{\wt}(\cdot)$, $\mathcal{d}_{\lrt}(\cdot)$ or $\mathcal{d}_{\rt}(\cdot)$
defined in (\ref{eqt:com_WT})--(\ref{eqt:com_RT}).

\section{Motivation and plan of proposal}\label{sec:plan}
\subsection{Motivation}\label{sec:motivation}
Complete-data Wald's and LR test statistics are asymptotically equivalent under $H_0$; 
see Section 4.4 of \cite{serfling1980}.
\cite{mengRubin92} showed 
that this asymptotic equivalence continues to hold for the MI statistics defined  
in (\ref{eqt:dp_dpp_w_l}) and (\ref{eqt:dp_dpp_w_l_LER}), that is, 
$\widetilde{d}'_{\wt} \sim \widetilde{d}'_{\lrt}$
and 
$\widetilde{d}''_{\wt} \sim \widetilde{d}''_{\lrt}$,
where $A_n \sim B_n$ means that $A_n-B_n\inP 0$, and  
``$\inP$'' denotes convergence in probability. 
So, we may asymptotically represent $\widetilde{d}'_{\wt}$ and $\widetilde{d}'_{\lrt}$ as $\widetilde{d}'$, 
and represent $\widetilde{d}''_{\wt}$ and $\widetilde{d}''_{\lrt}$ as $\widetilde{d}''$. 
Then $\widetilde{D}_{\wt}$ and $\widetilde{D}_{\lrt}$ equal to $\widetilde{D}$ asymptotically, where 
\begin{equation}
	\widetilde{D} := \frac{\widetilde{d}''}{k\left\{ 1+ (1+\frac{1}{m})\widetilde{\mu}_r \right\} }  
	\qquad\text{and}\qquad
	\widetilde{\mu}_r := \frac{\widetilde{d}'-\widetilde{d}''}{k(m-1)/m} . \label{eqt:asy_version_D}  
\end{equation}
From (\ref{eqt:asy_version_D}), we know that the MI test statistic $\widetilde{D}$ 
depends on $X^1, \ldots, X^m$ only through $\widetilde{d}'$ and 
$\widetilde{d}''$. 
So, all information contained in $X^1, \ldots, X^m$ 
is summarized by $\widetilde{d}'$ and $\widetilde{d}''$. 
In general,
the two-number summary $(\widetilde{d}',\widetilde{d}'')$ is not enough for estimating 
$k$ unknown parameters $r_{1}, \ldots, r_{k}$. 
Hence, in order to estimate all individual $r_1, \ldots, r_k$, 
it is necessary to derive a more general class of MI statistics 
other than $\widetilde{d}'$ and $\widetilde{d}''$.

Besides, we would like to have a MI testing procedure
that can be completed solely by the device $\mathcal{d}(\cdot)$.
In order to achieve this goal, 
we begin with representing the statistics $\widetilde{d}'$ and $\widetilde{d}''$ 
in terms of $\mathcal{d}(\cdot)$.
According to \cite{XXMeng2017}, we can asymptotically represent $\widetilde{d}'$ and $\widetilde{d}''$ as
\begin{eqnarray}\label{eqt:dp_dpp}
	\widetilde{d}' \sim \frac{1}{m}\sum_{\ell=1}^m \mathcal{d}(X^{\ell})
	\qquad \text{and} \qquad
	\widetilde{d}'' \sim \frac{1}{m}\mathcal{d}(X^{\{1:m\}}), 
\end{eqnarray}
where
$X^{\{1:m\}} := [(X^{1})^{\T}, \ldots, (X^m)^{\T} ]^{\T}$ is a stacked dataset.
Using (\ref{eqt:dp_dpp}),
we may interpret $\widetilde{d}'$ and $\widetilde{d}''$
as summary statistics via stacking \emph{one} and \emph{all} imputed datasets, respectively. 
Stacking different numbers of imputed datasets produces distinct inferential tools.
For example, 
$\widetilde{d}''$, the numerator of the MI test statistic $\widetilde{D}$ in (\ref{eqt:asy_version_D}), 
measures the amount of evidence against $H_0$; whereas
$\widetilde{d}'-\widetilde{d}''$, 
which is proportional to the estimator $\widetilde{\mu}_r$ in (\ref{eqt:asy_version_D}), 
measures the amount of information loss due to missing data. 
Hence, it motivates us to derive new MI statistics by stacking imputed datasets in various ways.

\subsection{Overview and plan of proposal}
Following the motivations in Section \ref{sec:motivation}, we propose a MI test statistic that admits the form 
\begin{eqnarray}
	\widehat{D} := \frac{\widehat{d}^{\{1:m\}}}{k\left\{ 1 + (1+\frac{1}{m})\widehat{\mu}_r \right\} }  ,
	\label{eqt:asy_version_D_proposal}  
\end{eqnarray}
where $\widehat{d}^{\{1:m\}}$ and $\widehat{\mu}_r$  
are some statistics to be derived so that 
(i) they can be computed solely by using the device $\mathcal{d}(\cdot)$, and  
(ii) the asymptotic null distribution of $\widehat{D}$ is $\mathbb{D}$ 
(see (\ref{eqt:representation_D_limit}))
without any ad-hoc approximation. 
The goals (i) and (ii) are completed in Sections \ref{sec:method}
and \ref{sec:testing}, respectively. 
Since the distribution of $\mathbb{D}$ depends on $r_1, \ldots, r_k$, 
we propose to approximate them by their estimators $\widehat{r}_j$'s, 
which are derived in Section \ref{sec:estimation} by using various stacked statistics.
Algorithm \ref{algo:SMI_exact} computes our proposed test statistic $\widehat{D}$ 
and the corresponding $p$-value.
We will explain the steps of Algorithm \ref{algo:SMI_exact} in the subsequent sections.

\begin{algorithm}[t]
\caption{Asymptotically correct MI test for $H_0$}\label{algo:SMI_exact}
\SetAlgoVlined
\DontPrintSemicolon
\SetNlSty{texttt}{[}{]}
\small
\textbf{Input}: {\;
(i) $X\mapsto\mathcal{d}(X)$ -- any complete-data testing device in (\ref{eqt:com_WT})--(\ref{eqt:com_RT}); \;
(ii) $X^1, \ldots, X^m$ -- $m$ properly imputed datasets; and\; 
(iii) $k$ -- the dimension of $\Theta$.}\\
\Begin{
Stack $X^1, \ldots,  X^m$ row-by-row to form $X^{\{1:m\}}$.\;
Compute $\widehat{d}^{\{1:m\}} \gets \mathcal{d}(X^{\{1:m\}})/m$.\;
\For{$\ell \in \{1, \ldots, m\}$}{
	Stack $X^1, \ldots, X^{\ell-1}, X^{\ell+1}, \ldots, X^m$ row-by-row to form $X^{\{-\ell\}}$.\;
	Compute $\widehat{d}^{\{-\ell\}} \gets \mathcal{d}(X^{\{-\ell\}})/(m-1)$.\;
	Compute $\widehat{d}^{\{\ell\}} \gets \mathcal{d}(X^{\ell})$.\; 
	Compute $\widehat{T}_{\ell} \gets (m-1)\widehat{d}^{\{-\ell\}} + \widehat{d}^{\{\ell\}} - m\widehat{d}^{\{1:m\}}$.\;
}
Compute $\widehat{t}_{j} \gets \sum_{\ell=1}^m \widehat{T}^{j}_{\ell}/m$
for each $j=1, \ldots, k$.\;
Compute $\widehat{r}_{1:k} \gets M_1^{-1}(M_2^{-1}(\widehat{t}_{1:k}))$ according to 
Proposition \ref{prop:inverse}.\;
Compute $\widehat{D}$ according to (\ref{eqt:asy_version_D_proposal}).\;
Draw $G_j^{(\iota)}\sim \chi^2_1$ and $H_j^{(\iota)}\sim \chi^2_{m-1}/(m-1)$ independently for $\iota=1, \ldots, N$ and $j=1, \ldots, k$. \;
Compute $\mathbb{\widehat{D}}^{(\iota)}$, $\iota=1, \ldots, N$,  according to (\ref{eqt:DmHat_est}). Set $N=10^4$ by default.\;
Compute $\widehat{p} \gets \sum_{\iota=1}^N \mathbb{1}(\widehat{\mathbb{D}}^{(\iota)} \geq \widehat{D})/N$.\;
}	
\textbf{return}: 
{$\widehat{p}$ -- the $p$-value for testing $H_0$ against $H_1$.}
\end{algorithm}

\section{Methodology and Principle}\label{sec:method}
\subsection{Stacking Principle}\label{sec:stackPrinciple}
In this section, we introduce a new class of MI statistics
by stacking $X^1, \ldots, X^m$ in various ways.  
As we shall see in Theorem \ref{thm:asy_Distof_bar_d_S} below, 
stacking them differently extracts different information from $X^1, \ldots, X^m$. 
We refer this phenomenon to a \emph{stacking principle}. 
Define the stacked dataset $X^S$ by 
stacking $\{ X^{\ell} : \ell\in S\}$ row-by-row for some non-empty $S \subseteq \{1,\ldots,m\}$.
For example, $X^S = [(X^{\ell_1})^{\T},\ldots, (X^{\ell_s})^{\T}]^{\T}$ is an $(ns)\times p$ matrix if $S=\{\ell_1,\ldots, \ell_s\}$. 
Define 
\begin{eqnarray}\label{eqt:def_bardS}
	\widehat{d}^S 
		:= \frac{1}{|S|} \mathcal{d}(X^S),
\end{eqnarray}
where $|S|$ is the cardinality of $S$, 
and $\mathcal{d}(\cdot)$ is any testing device in (\ref{eqt:com_WT})--(\ref{eqt:com_RT}). 
In particular, we have
$\widehat{d}^{\{\ell\}} = \mathcal{d}(X^{\ell})$,
$\widehat{d}^{\{-\ell\}}= \mathcal{d}(X^{\{-\ell\}})/(m-1)$, and
$\widehat{d}^{\{1:m\}} = \mathcal{d}(X^{\{1:m\}})/m$,
where $\{1:m\} := \{1,\ldots,m\}$ and 
$\{-\ell\} := \{1,\ldots,m\}\setminus \{\ell\}$. 
We denote $\widehat{d}^S$ by $\widehat{d}^S_{\wt}$, $\widehat{d}^S_{\lrt}$ and $\widehat{d}^S_{\rt}$ 
to emphasize that 
$\mathcal{d}_{\wt}$, $\mathcal{d}_{\lrt}$ and $\mathcal{d}_{\rt}$ are used, respectively.

It is worth mentioning that ``stacking'' is a \emph{universal} operation. 
Since it is problem-independent, 
the analysts can apply this operation to all kinds of testing problems universally. 
This nature is similar to some well-known procedures, for example, 
bootstrapping \citep{efro1979}, jackknife resampling \citep{Quenouille56}, subsampling \citep{politis1999}, etc.  
All these procedures are model-free and fully non-parametric.

The usefulness of 
$\widehat{d}^S$ can be seen from its asymptotic distribution.  
To derive it, we need some RCs.

\begin{condition}\label{cond-Normality}
The observed-data MLE $\widehat{\theta}_{\obs}$ of $\theta$ 
satisfies $T^{-1/2}(\widehat{\theta}_{\obs}-\theta^{\star}) \inD \Normal_k(0_k,I_k)$, 
where $T := I_{\obs}^{-1}$ is well-defined, 
$\theta^{\star}$ is the true value of $\theta$;
$0_k$ is a $k$-vector of zeros; and
$I_k$ is a $k\times k$ identity matrix.
\end{condition}

\begin{condition}\label{cond-properImp} 
The imputed statistics 
$(\widehat{\theta}^{1},\widehat{V}^{1}), \ldots, (\widehat{\theta}^{m},\widehat{V}^{m})$ 
are conditionally independent given $X_{\obs}$. 
Moreover, for each $\ell=1, \ldots,m$, they satisfy
$\{ B^{-1/2}(\widehat{\theta}^{\ell}-\widehat{\theta}_{\obs}) \mid X_{\obs}\}\inD \Normal_k(0_k,I_k)$ and $\{ T^{-1}(\widehat{V}^{\ell}-V) \mid X_{\obs} \} \inP O_k$, 
where $B := I_{\obs}^{-1}-I_{\com}^{-1}$ and $V := I_{\com}^{-1}$ are well-defined, and 
$O_k$ is a $k\times k$ matrix of zeros. 
\end{condition}

Condition \ref{cond-Normality} is satisfied under the usual RCs that guarantee 
asymptotic normality of MLEs; 
see, for example, \cite{wangRobins1998} and \cite{KimShao2013}.
Condition \ref{cond-properImp} is satisfied if a proper imputation model \citep{rubin1987} is used. 
The posterior predictive distribution $f(X_{\mis}\mid X_{\obs})$
is an example of proper imputation models, 
which have been widely used and adopted in MI; see 
Sections 2.4--2.7 of \citet{rubin1996multiple} for a comprehensive discussion of this assumption.
We emphasize that the analysts do not need to compute or know
the imputed statistics $(\widehat{\theta}^{\ell},\widehat{V}^{\ell})$ by themselves. 
Condition \ref{cond-properImp} guarantees that the imputer does his/her imputation job correctly. 
It is remarked that if the analyst's and imputer's models are uncongenial \citep{meng1994}, 
then the standard Rubin's MI procedure may not be valid. 
Generalizing MI procedures to the uncongenial case is not completely solved yet. 
Interested readers are referred to a recent discussion article \citep{XXMeng2017}
for a simple remedy when $m\rightarrow\infty$.
Extending our proposed method to the uncongenial case is left for further study.

\begin{definition}\label{cond-localH1}
Let $\theta^{\star}$ be the true value of $\theta$, and 
$\theta_0$ be the null value of $\theta$ specified in $H_0$.
The difference $\theta^{\star}-\theta_0$ satisfies that 
$\sqrt{n}A(\theta^{\star}-\theta_0) \rightarrow \delta \equiv (\delta_1,\ldots,\delta_k)^{\T}$  
for some $\delta$ and invertible matrix $A$.
\end{definition}
Definition \ref{cond-localH1} defines a sequence of local alternative hypotheses; 
see, for example,  \cite{vanDerVaart2000}. 
Note that it is required for 
proving asymptotic equivalence of 
the test statistics 
$\mathcal{d}_{\wt}(X)$, 
$\mathcal{d}_{\lrt}(X)$ and 
$\mathcal{d}_{\rt}(X)$; see, for example, \cite{serfling1980} and \cite{TSH_lehmannRomano}.  
This general setting allows us to prove the validity of MI estimators, 
even when $H_0$ is not true. 
Moreover, practical testing problems are more challenging under a local alternative hypothesis 
than under an obviously wrong fixed alternative hypothesis 
because, under a fixed alternative hypothesis,   
all three test statistics obviously diverge to infinity, and have power one asymptotically.
Hence, our setting is sufficient for most practical applications.
The following theorem states the asymptotic distribution of $\widehat{d}^{S}$.

\begin{theorem}\label{thm:asy_Distof_bar_d_S}
Assume Conditions \ref{cond-Normality}--\ref{cond-properImp}. 
Let $W_j, Z_{j1}, \ldots, Z_{jm}$, $j=1,\ldots,k$, be independent $\Normal(0,1)$ random variables, and 
$\bar{Z}_{j(S)} = \sum_{\ell\in S}Z_{j\ell}/|S|$ be the average of $\{Z_{j\ell} : \ell\in S\}$.
Then, for any non-empty multiset $S$ from $\{1,\ldots,m\}$, 
\begin{eqnarray}\label{eqt:representation_of_d}
	\widehat{d}^{S}
		\inD \mathbb{d}^S 
		:= \sum_{j=1}^k \left\{ \delta_j + (1+r_{j})^{1/2}W_j + r_{j}^{1/2} \bar{Z}_{j(S)} \right\}^2 , 
\end{eqnarray}
where the convergence is true jointly for all $S$.
\end{theorem}

Theorem \ref{thm:asy_Distof_bar_d_S} is true for any multiset $S$, for example, $S=\{1,1,2,3\}$. 
In this case $\bar{Z}_{j(S)} = (Z_{j1}+Z_{j1}+Z_{j2}+Z_{j3})/4$.
We emphasize that Theorem \ref{thm:asy_Distof_bar_d_S} 
requires neither 
Condition \ref{ass:EOMI} nor $m\rightarrow\infty$.
So, it is in line with the practical situation. 
The convergence (\ref{eqt:representation_of_d}) is with respect to the regime $n\rightarrow\infty$,
hence, it is simply a usual large-sample asymptotic result. 
More importantly, Theorem \ref{thm:asy_Distof_bar_d_S} sheds light on 
performing hypothesis tests because of two reasons. 
First, the limiting distribution $\mathbb{d}^S$ depends on $\delta$.
When $H_0$ is true, that is, $\delta_1 = \cdots = \delta_k = 0$, 
the statistic $\widehat{d}^{S}$ converges weakly to a non-degenerated distribution. 
When $\delta_j\rightarrow \infty$ for some $j$, 
$\widehat{d}^{S}$ diverges to infinity. 
So, the statistic $\widehat{d}^S$ can be used to test $H_0$ 
for every non-empty $S$.
Second, 
$\mathbb{d}^S$ depends on $r_{1},\ldots, r_{k}$, 
but the dependence on $r_{1},\ldots, r_{k}$ varies among $S$. 
Consequently, pooling information from different $\widehat{d}^S$ may help to estimate $r_{1},\ldots, r_{k}$.

However,  
$\widehat{d}^S$ is not immediately useful because of two reasons.  
First, $\mathbb{d}^S$ depends on 
$r_{1},\ldots, r_{k}$ in a complicated way.  
In particular, the $j$th summand on the right-hand side of (\ref{eqt:representation_of_d})
depends on $r_{j}$ \emph{non-linearly}. 
Thus, it is not clear how to use $\widehat{d}^S$ to estimate $r_{1}, \ldots, r_{k}$.
Second, 
for any nonempty $S_1$ and $S_2$, the statistics 
$\widehat{d}^{S_1}$ and $\widehat{d}^{S_2}$ are asymptotically dependent
through \emph{both} $\{W_j\}$ and $\{Z_{j\ell}\}$. 
The random variables $W_1, \ldots, W_k$ always appear in the 
limiting distributions of $\widehat{d}^{S_1}$ and $\widehat{d}^{S_2}$,
no matter how $S_1$ and $S_2$ are chosen.
So, $\widehat{d}^{S_1}$ and $\widehat{d}^{S_2}$ are too correlated to be useful 
if $S_1$ and $S_2$ are blindly selected.

\subsection{Stacked multiple imputation -- a new class of MI procedures}\label{sec:SMI_stat}
In this section, we propose a novel methodology for properly using $\widehat{d}^S$. 
According to the discussion in Section \ref{sec:stackPrinciple},
we know that $\widehat{d}^S$ is too complex to be useful
because of two sources of dependence: 
(i) the non-linear dependence of the limiting distribution of $\widehat{d}^S$ on $r_{1}, \ldots, r_{k}$, and 
(ii) the strong probabilistic dependence among different $\widehat{d}^S$ through the random variables $W_1,\ldots, W_k$.
In this section, we propose a method to get rid of all these two unwanted sources of dependence.

For any non-empty sets $S_1,S_2\subseteq \{1,\ldots, m\}$ such that $S_1\neq S_2$, define 
\begin{equation*}\label{eqt:J}
	\widehat{T}_{S_1, S_2} 
		= \frac{|S_1|+|S_2|}{|S_1|+|S_2|-2|S_1\cap S_2|} 
			\left\{ |S_1|\widehat{d}^{S_1} + |S_2|\widehat{d}^{S_2} - (|S_1|+|S_2|)\widehat{d}^{S_1\oplus S_2} \right\}, \nonumber \\
\end{equation*}
where $S_1 \oplus S_2$ is the multiset addition, for example, $\{1,3\}\oplus\{1,2\} = \{1,1,2,3\}$. 
We call $\widehat{T}_{S_1, S_2}$ a \emph{stacked multiple imputation} (SMI) statistic.  
Note that
the SMI statistic $\widehat{T}_{S_1, S_2}$ can be computed 
solely by stacking the imputed datasets and applying the complete-data testing device $\mathcal{d}(\cdot)$.  
Besides, as we shall see in Proposition \ref{prop:limit_of_J} below, 
$\widehat{T}_{S_1, S_2}$ is free of the two aforementioned sources of unwanted dependence.  
Hence, the SMI statistic $\widehat{T}_{S_1, S_2}$  has nice computational and theoretical properties. 
Consequently, it
is qualified to be a building block for all MI procedures proposed in this paper. 
Let 
\begin{eqnarray}\label{eqt:R_Rstar}
	R_{\tau} := \sum_{j=1}^{k} r_j^{\tau}, 
	\qquad \tau = 1,\ldots, k.
\end{eqnarray}
The asymptotic distribution of $\widehat{T}_{S_1, S_2}$ and its properties are shown below.

\begin{proposition}\label{prop:limit_of_J}
Assume Conditions \ref{cond-Normality}--\ref{cond-properImp}.  
Let $S_1,S_2\subseteq \{1,\ldots, m\}$ be any non-empty and non-identical sets.
Define $W_j, Z_{j1}, \ldots, Z_{jm}$,  $j=1,\ldots,k$ as in Theorem \ref{thm:asy_Distof_bar_d_S}.
\begin{enumerate}
	\item $\widehat{T}_{S_1, S_2}\inD \mathbb{T}_{S_1,S_2}$, where $\mathbb{T}_{S_1,S_2}$ is represented 
			as
			\begin{align}\label{eqt:limit_of_J}
				\mathbb{T}_{S_1,S_2}
					:= 
					\frac{|S_1|\times|S_2|}{|S_1|+|S_2|-2|S_1\cap S_2|}
					\sum_{j=1}^k r_j \left\{ \bar{Z}_{j(S_1)} - \bar{Z}_{j(S_2)} \right\}^2. 
			\end{align}
	\item $\mathbb{T}_{S_1,S_2}$ has the same marginal distribution as $\mathbb{T} := \sum_{i=1}^k r_i U_i$,
where $U_1,\ldots, U_k\sim \chi^2_1$ independently.
	\item Let $t_{\tau} := \E( \mathbb{T}^{\tau} )$ for $\tau\in\{1,\ldots, k\}$, 
			and $t_0:=1$.
			Then $t_1, \ldots, t_k$ can be found iteratively as follows: 
			\begin{eqnarray}\label{eqt:def_q}
				t_1 = R_1 
				\qquad \text{and} \qquad
				t_{\tau} = \sum_{j=1}^{\tau} \frac{(\tau-1)!}{(\tau-j)!} 2^{j-1} R_{j} t_{\tau-j}
			\end{eqnarray}
			for $\tau = 2,\ldots, k$. 
			In particular, 
			$\E(\mathbb{T}) = R_1$ and $\Var(\mathbb{T}) = 2R_2$. 
\end{enumerate}
\end{proposition}

According to Proposition \ref{prop:limit_of_J} (1),  
the limiting distribution $\mathbb{T}_{S_1,S_2}$ 
depends on $r_{1}, \ldots, r_{k}$ linearly, and 
is independent on the random variables $W_1, \ldots, W_k$ that appear in (\ref{eqt:representation_of_d}).
Hence, the SMI statistic $\widehat{T}_{S_1, S_2}$ ``filters'' out the two unwanted dependence structures.
Consequently, $\widehat{T}_{S_1, S_2}$ is easier to work with.
Proposition \ref{prop:limit_of_J} (2)
states that 
$\widehat{T}_{S_1, S_2}$'s are asymptotically 
identically distributed over different $S_1, S_2$.
Proposition \ref{prop:limit_of_J} (3) states that the $\tau$th moments of $\mathbb{T}$
depends on $r_1, \ldots r_k$ only through $R_1, \ldots, R_{\tau}$ for each $\tau =1, \ldots, k$.

Proposition \ref{prop:limit_of_J} implies that 
the sample $\tau$th moments of $\widehat{T}_{S_1, S_2}$ over a set of pairs of $(S_1, S_2)$ 
is a natural estimator of $t_{\tau}$. 
Let $\Lambda \subseteq \mathcal{L} := \{ (S_1,S_2) \subseteq \{1,\ldots,m\}^2: S_1,S_2\neq\emptyset ; S_1\neq S_2\}$
so that $\Lambda$ is a set of appropriately chosen pairs of $(S_1, S_2)$. 
For $\tau\in\{1, \ldots, k\}$, define an estimator of $t_{\tau}$ by 
\begin{eqnarray}\label{eqt:tHat}
	\widehat{t}_{\tau} := \widehat{t}_{\tau}(\Lambda) := \frac{1}{|\Lambda|}\sum_{(S_1,S_2)\in\Lambda} \widehat{T}_{S_1, S_2}^{\tau}, 
\end{eqnarray}
where the short notation $\widehat{t}_{\tau}$ is used when 
$\Lambda$ is clear in the context. 
Some examples of $\Lambda$ are given below. 

\begin{example}\label{eg:Lambda}  
The following selection rules of $\Lambda$ are suggested. Let
\begin{align}\label{eqt:LambdaEg}
	\begin{array}{c}
	\displaystyle \Lambda_{\Jack} := \Big\{ \Big(\{\ell\}, \{-\ell\}\Big) \Big\}_{1\leq \ell\leq m},\quad
	\displaystyle \Lambda_{\Full} := \Big\{ \Big(\{\ell\}, \{1:m\}\Big) \Big\}_{1\leq \ell\leq m}, \\[2ex]
	\displaystyle \Lambda_{\Pair} := \Big\{ \Big(\{\ell\}, \{\ell'\}\Big) \Big\}_{1\leq \ell<\ell'\leq m}
	\end{array}
\end{align}
be the \emph{Jackknife}, \emph{full} and \emph{pair} selection rules for $\Lambda$.
Note that $|\Lambda_{\Jack}| = |\Lambda_{\Full}| = m$, whereas $|\Lambda_{\Pair}|=m(m-1)/2$.
Putting $\Lambda = \Lambda_{\Jack}, \Lambda_{\Full},\Lambda_{\Pair}$ into (\ref{eqt:tHat}), 
we obtain the following three estimators of $t_{\tau}$:
\begin{align*}\label{eqt:tHatEg}
	\begin{array}{c}
	\displaystyle\widehat{t}_{\tau}(\Lambda_{\Jack}) = \frac{1}{m}\sum_{\ell=1}^m \widehat{T}_{\{\ell\},\{-\ell\}}^{\tau} , \qquad
	\displaystyle\widehat{t}_{\tau}(\Lambda_{\Full}) = \frac{1}{m}\sum_{\ell=1}^m \widehat{T}_{\{\ell\},\{1:m\}}^{\tau},\\
	\displaystyle\widehat{t}_{\tau}(\Lambda_{\Pair}) = \frac{2}{m(m-1)}\sum_{\ell'=2}^m\sum_{\ell=1}^{\ell'-1} \widehat{T}_{\{\ell\},\{\ell'\}}^{\tau},
	\end{array}
\end{align*}
respectively. 
Since $\widehat{t}_{\tau}(\Lambda_{\Pair})$ requires stacking at most two datasets,
it should be used when computing $\mathcal{d}(X)$ is difficult for a large dataset $X$. 
Although the device $\mathcal{d}(\cdot)$ has to be implemented $3m(m-1)/2 = O(m^2)$ times, 
the computations can be parallelized easily.
On the other hand, computing $\widehat{t}_{\tau}(\Lambda_{\Jack})$ and $\widehat{t}_{\tau}(\Lambda_{\Full})$ requires implementing the device $\mathcal{d}(\cdot)$ $3m=O(m)$ times only. It is preferable when $m$ is large.
\end{example}

The message behind Example \ref{eg:Lambda} is that 
stacked MI is flexible enough to allow users to choose the most computationally viable 
statistics according to their problems. 
Although testing on stacked datasets is more computationally intensive and 
requires more computing memory, 
these computational requirements are usually affordable by standard laptop computers nowadays. 
Indeed, the increase in \emph{computing cost} is used to exchange for 
a decrease in \emph{human time cost} of deriving or searching non-standard computing devices
required in the exiting tests T2 and T3 stated in Table \ref{tab:approxNullDist}.

\subsection{Asymptotic Properties}
In this section, we study the asymptotic properties of $\widehat{t}_{\tau}(\Lambda)$ 
in the $\mathcal{L}^2$ sense. 
From now on,
we denote the indicator function by $\mathbb{1}(\cdot)$.
The following condition is required for developing $\mathcal{L}^2$ convergence results.

\begin{condition}\label{cond-UI}
For any non-empty and non-identical $S_1,S_2\subseteq \{1,\ldots, m\}$, denote $\widehat{T}_{S_1, S_2}(n) = \widehat{T}_{S_1, S_2}$ as a sequence indexed by the sample size $n$. 
The sequence $\{ \widehat{T}_{S_1, S_2}^{2\tau}(n) : n\in\mathbb{N} \}$ is assumed uniformly integrable, that is, 
\[
	\lim_{C\rightarrow\infty} \sup_{n\in\mathbb{N}} \E\left[ \widehat{T}_{S_1, S_2}^{2\tau}(n) \mathbb{1}\{|\widehat{T}_{S_1, S_2}(n)|>C\} \right]
	=0.
\]	
\end{condition}

The following theorem derives 
the limits of $\E\{\widehat{t}_{\tau}(\Lambda)\}$ and $\Var\{\widehat{t}_{\tau}(\Lambda)\}$ 
as $n\rightarrow\infty$.

\begin{theorem}\label{thm:limit_EJ}  
Assume Conditions \ref{cond-Normality}--\ref{cond-UI}. 
Let $m>1$, $\Lambda\subseteq\mathcal{L}$ and $\tau\in\mathbb{N}$. 
Then the following hold.  
(1) $\E\{\widehat{t}_{\tau}(\Lambda)\}\rightarrow t_{\tau}$ as $n\rightarrow\infty$. 
(2)  For $(S_1,S_2),(S_3,S_4)\in \Lambda$, define 
			\begin{eqnarray*}
				\rho(S_1,S_2,S_3,S_4) &:=& (s_1+s_2-2s_{12})^{-1/2}(s_3+s_4-2s_{34})^{-1/2}\\
					&&\qquad \times\left(\frac{s_{13}}{s_1s_3}-\frac{s_{14}}{s_1s_4}-\frac{s_{23}}{s_2s_3}+\frac{s_{24}}{s_2s_4}\right), 
			\end{eqnarray*}
			where $s_a=|S_a|$ and $s_{ab}=|S_a\cap S_b|$ for $a,b\in\{1,2,3,4\}$.
			Then 
			\begin{align}\label{eqt:limVarJbar}
				\Var\{\widehat{t}_{\tau}(\Lambda)\} \rightarrow \frac{1}{|\Lambda|^2}\sum_{(S_1,S_2)\in\Lambda}\sum_{(S_3,S_4)\in\Lambda} \rho^2(S_1,S_2,S_3,S_4) C_{\tau}(S_1,S_2,S_3,S_4), 
			\end{align}
			where $C_{\tau}(S_1,S_2,S_3,S_4)$ satisfies 
			$
				\sup_{(S_1,S_2),(S_3,S_4)\in\Lambda} |C_{\tau}(S_1,S_2,S_3,S_4)| \leq C_{\tau}
			$
			for some finite $C_{\tau}$ which depends on $\tau,r_1,\ldots, r_k$. 
			In particular, $C_{1}(S_1,S_2,S_3,S_4) = 2R_2$.

\end{theorem}

From Theorem \ref{thm:limit_EJ}, the choice of $\Lambda$ only affects the limit of $\Var\{\widehat{t}_{\tau}(\Lambda)\}$ but not the limit of $\E\{\widehat{t}_{\tau}(\Lambda)\}$. 
The asymptotic variance of $\widehat{t}_{\tau}(\Lambda)$ is affected by 
two factors: (i) $1/|\Lambda|^2$ and (ii) the double summation in (\ref{eqt:limVarJbar}). 
If $|\Lambda|$ is too small, the first factor $1/|\Lambda|$ is large.  
On the other hand, if $|\Lambda|$ is too large, 
then 
$\{ \widehat{T}_{S_1, S_2} : (S_1,S_2)\in\Lambda\}$ may be highly correlated
in the sense that  
most of the $\rho(S_1,S_2,S_3,S_4)$'s in (\ref{eqt:limVarJbar}) are large.
Hence, $\Lambda$ needs to be chosen carefully 
so that the estimator $\widehat{t}_{\tau}(\Lambda)$ has a small asymptotic variance. 
Theorem \ref{thm:limit_EJ} is easy to use.
Once the set $\Lambda$ is fixed, one can easily compute 
$\rho(S_1,S_2,S_3,S_4)$. 
By simple counting, the order of magnitude of the asymptotic variance of  
$\widehat{t}_{\tau}(\Lambda)$ can also be found.
In particular, we show
that $\widehat{t}_{\tau}(\Lambda_{\Jack})$, $\widehat{t}_{\tau}(\Lambda_{\Full})$ and $\widehat{t}_{\tau}(\Lambda_{\Pair})$
in Example \ref{eg:Lambda}
are all good estimators of $t_{\tau}$.

\begin{corollary}\label{coro:example_var}
Define $\Lambda_{\Jack}$, $\Lambda_{\Full}$ and $\Lambda_{\Pair}$ according to (\ref{eqt:LambdaEg}).
Assume Conditions \ref{cond-Normality}--\ref{cond-UI}. 
For any $\mathcal{S}\in\{\Jack, \Full, \Pair\}$ and any $\tau\in\{1,\ldots, k\}$,  
we have
$\Var\{\widehat{t}_{\tau}(\Lambda_{\mathcal{S}})\} \rightarrow V_{\mathcal{S},\tau}(m)$
as $n\rightarrow\infty$,
where $V_{\mathcal{S},\tau}(m)=O(1/m)$ as $m\rightarrow\infty$.
In particular, $V_{\Jack, 1}(m) = V_{\Full, 1}(m) =V_{\Pair, 1}(m) = 2R_2/(m-1)$.
\end{corollary}

Corollary \ref{coro:example_var} shows that, asymptotically,   
the precision of $\widehat{t}_{\tau}(\Lambda_{\mathcal{S}})$ 
increases with $m$ for any selection rule $\mathcal{S}\in\{\Jack, \Full, \Pair\}$.
Together with Theorem \ref{thm:limit_EJ} (1), 
the mean squared error (MSE) of $\widehat{t}_{\tau}(\Lambda_{\mathcal{S}})$, that is, 
$\MSE\{\widehat{t}_{\tau}(\Lambda_{\mathcal{S}})\} := \E\{\widehat{t}_{\tau}(\Lambda_{\mathcal{S}}) - t_{\tau}\}^2$ 
decreases in the order of $O(1/m)$ as $m$ increases.

Unless otherwise stated, we use $\Lambda = \Lambda_{\Jack}$ by default. 
Then the estimator $\widehat{t}_{\tau}(\Lambda_{\Jack})$ is as simple as 
\[
	\widehat{t}_{\tau} := \frac{1}{m}\sum_{\ell=1}^m \widehat{T}_{\ell}^{\tau}, 
	\qquad \text{where} \qquad 
	\widehat{T}_{\ell}  
		:= \widehat{d}^{\{\ell\}} + (m-1)\widehat{d}^{\{-\ell\}} - m\widehat{d}^{\{1:m\}} .
\]

\section{Estimation of OMI}\label{sec:estimation}
This section proposes estimators for all individual $r_1, \ldots, r_k$. 
For notational simplicity, 
we denote $r_{1:k} = (r_1, \ldots, r_k)^{\T}$. 
We also abbreviate other variables similarly, 
for example, $t_{1:k} = (t_1, \ldots, t_k)^{\T}$.

From Proposition \ref{prop:limit_of_J}, 
$t_{1:k}$ are defined via the following two-step mapping:
\begin{eqnarray}\label{eqt:mapping_2steps}
	r_{1:k}
	\;\overset{M_1}{\mapsto}\; R_{1:k}
	\;\overset{M_2}{\mapsto}\; t_{1:k},
\end{eqnarray}
where the maps $M_1$ and $M_2$ are defined according to (\ref{eqt:R_Rstar}) and (\ref{eqt:def_q}), 
respectively. 
From Theorem \ref{thm:limit_EJ}, 
$\widehat{t}_{1:k}$ are good estimators of $t_{1:k}$.
Our goal is to estimate $r_{1:k}$,
through ``reverse engineering'' the two-step transformation (\ref{eqt:mapping_2steps}). 
However, it is impossible unless $r_{1:k}$ is uniquely determined by $t_{1:k}$. 
The following proposition shows that  
(\ref{eqt:mapping_2steps}) is a one-to-one function, 
that is, the function inverse of $M_2(M_1(\cdot))$ always exists.

\begin{proposition}\label{prop:inverse}
Define the functions $M_1(\cdot)$ and $M_2(\cdot)$ according to (\ref{eqt:R_Rstar}) and (\ref{eqt:def_q}), 
respectively. 
\begin{enumerate}
	\item The inverse function $r_{1:k} = M_1^{-1}(R_{1:k})$ exists,
			that is, $r_{1:k}$ are uniquely determined by $R_{1:k}$. 
			Let
			\begin{eqnarray}
			A := 
			\left[ 
			\begin{array}{c|ccc}
				0_{k-1} & & I_{k-1}& \\ \hline
				{a}_1 & {a}_2 & \cdots & {a}_{k}  
			\end{array}
			\right],  \label{eqt:def_A}
		\end{eqnarray}
		where $a_k:=R_1$ and ${a}_{k-j+1} := ( {R}_j - \sum_{i=1}^{j-1} {R}_{i} {a}_{k-j+i+1} )/j$, 
		for $j=2,\ldots,k$.
		Then $r_j$ is the $j$th largest modulus of the eigenvalue of $A$ for $j=1, \ldots,k$.

	\item The inverse function $R_{1:k} = M_2^{-1}(t_{1:k})$ exists,
			that is, $R_{1:k}$ are uniquely determined by $t_{1:k}$ as follows: 
			$R_1 = t_1$ and 
			\begin{align}\label{eqt:estimator_Rstar}
				{R}_{\tau} = \frac{t_{\tau}}{(\tau-1)!2^{\tau-1}} - \sum_{j=1}^{\tau-1} \frac{t_{\tau-j}{R}_j}{(\tau-j)! 2^{\tau-j}}, \qquad \tau=2, \ldots,k.  
			\end{align}			
\end{enumerate}
\end{proposition}

Proposition \ref{prop:inverse} implies 
that the parameters $r_{1:k}$ are identifiable via 
the moment conditions $\E(\mathbb{T}^{\tau})=t_{\tau}$ ($\tau = 1,\ldots, k$)
in such a way that
$r_{1:k} = M_1^{-1}\left(M_2^{-1}(t_{1:k})\right)$.
Since the function $M_1^{-1}(M_2^{-1}(\cdot))$ is continuous and 
does not depend on any unknown,
we can estimate $R_{1:k}$ and  $r_{1:k}$ by
\begin{eqnarray}\label{eqt:initialEst_rR}
	\widehat{R}_{1:k}
	:= M_2^{-1}(\widehat{t}_{1:k})
	\qquad \text{and} \qquad
	\widehat{r}_{1:k}
	:= M_1^{-1} (\widehat{R}_{1:k}),
\end{eqnarray}
respectively. 
The corollary below states the large-$n$ asymptotic MSEs of $\widehat{r}_{j}$.

\begin{corollary}\label{coro:rHat}
Let $\Lambda$ be either $\Lambda_{\Jack}$, $\Lambda_{\Full}$ or $\Lambda_{\Pair}$. 
Under Conditions \ref{cond-Normality}--\ref{cond-UI}, as $n\rightarrow\infty$, 
$\MSE(\widehat{r}_j) := \E(\widehat{r}_j- r_j)^2 \rightarrow V(m)$, 
where $V(m)$ is a function of $m$ such that $V(m)\rightarrow 0$ as $m\rightarrow\infty$.
\end{corollary}

Corollary \ref{coro:rHat} guarantees that the estimators $\widehat{r}_{1:k}$ have 
small MSEs when $m,n$ are sufficiently large. 
The step-by-step procedure for computing $\widehat{r}_{1:k}$ 
with $\Lambda = \Lambda_{\Jack}$
is shown in Algorithm \ref{algo:SMI_exact}.
This algorithm is user-friendly for analysts 
as only the complete-data testing device $X\mapsto \mathcal{d}(X)$ is required. 
It is, indeed, the minimal requirement even for compete-data testing.

Besides, it is sometimes informative to summarize $r_{1}, \ldots, r_k$ through their mean and variance, that is,  
\begin{equation*}\label{eqt:moments_r}
	\mu_r := \frac{1}{k}\sum_{j=1}^k r_j  
	\qquad \text{and}  \qquad
	\sigma^2_r := \frac{1}{k}\sum_{j=1}^k(r_j - \mu_r)^2  .
\end{equation*}
These two values are required for approximating 
the limiting null distribution $\mathbb{D}$ in (\ref{eqt:representation_D_limit}) for testing $H_0$, 
for example, T2 and T3 in Table \ref{tab:approxNullDist}; 
see \citet{mengPhDthesis}, \cite{li91JASA}, and \cite{mengRubin92} for details. 
In Section \ref{sec:estimator_of_moments_r} of the supplementary note, 
we show that  
\begin{equation}\label{eqt:estimator_mur_sigmar}
	\widehat{\mu}_r := \frac{\widehat{t}_1}{k}
	\qquad \text{and} \qquad
	\widehat{\sigma}_{r}^2 
		:=  \frac{\left\{k(m-1)+2\right\} \widehat{t}_2
			-(m-1)(k+2)\widehat{t}_1^2}{2k^2(m-2)} 
\end{equation}
are asymptotically unbiased estimators of $\mu_r$ and $\sigma^2_r$, respectively.  
The precise statement and their properties are deferred to the supplementary note
due to space constraint. 
There are two immediate applications of (\ref{eqt:estimator_mur_sigmar}).
\begin{enumerate}
	\item The estimators $\widehat{\mu}_r$ and $\widehat{\sigma}_{r}^2$
			can be used to compute the approximated null distributions of T2 and T3 in Table \ref{tab:approxNullDist} 
			as the distributions depend only on $\mu_r$ and $\sigma^2_r$. 
			We emphasize that the original approximations T2 \citep{li91JASA} and T3 \citep{mengPhDthesis}
			work for Wald's test only. 
			Although \citet{mengRubin92} extended T2 to LR test, 
			T2 and T3 were still inapplicable to RS test.
			With the proposed estimators in  
			(\ref{eqt:estimator_mur_sigmar}), 
			the generalized T2 and T3 support Wald's, LR and RS tests.
	\item The statistic $\widehat{\sigma}_{r}^2/\widehat{\mu}_r^2$ estimates the squared 
			coefficient of variation of $\{r_j\}$.
			It can be used to construct a formal test 
			for the validity of EOMI; see
			Section \ref{sec:testEOMI} of the supplementary note for details.
\end{enumerate}

We conclude this subsection with an example.

\begin{example}\label{eg:rHat}
Let $r_1, \ldots, r_k$ be evenly spread in $[0.1, r_{\max}]$, that is,  
\begin{eqnarray}\label{eqt:delta_L25}
	r_j = 0.1+(r_{\max}-0.1)\frac{j-1}{k-1}, \qquad j=1, \ldots, k,
\end{eqnarray}
where $r_{\max} \in \{0.1, 0.2, \ldots, 0.9\}$ and $k\in\{2,4,6\}$.  
The coefficient of variation (CV) of $r_{1:k}$
increases with the value of $r_{\max}$. 
We evaluate $\widehat{r}_{1:k}$ under $H_0$ and $H_1$. 
Under $H_0$, $\delta_1=\cdots = \delta_k = 0$ in (\ref{eqt:representation_of_d}).
Under $H_1$, we set $\delta_1=\cdots = \delta_k = 1$ in particular.
In the simulation experiments, we assume that the sample size $n\rightarrow \infty$
but the number of imputation $m$ is fixed.
So, the experiments assess solely the performance of the MI procedure  
instead of the performance of the large-$n$ $\chi^2$-approximation (\ref{eqt:comTestchisq}). 
Applications and simulation experiments with finite $n$ are studied in Section \ref{sec:applications}.

\begin{figure}
\begin{center}
\includegraphics[width=\linewidth]{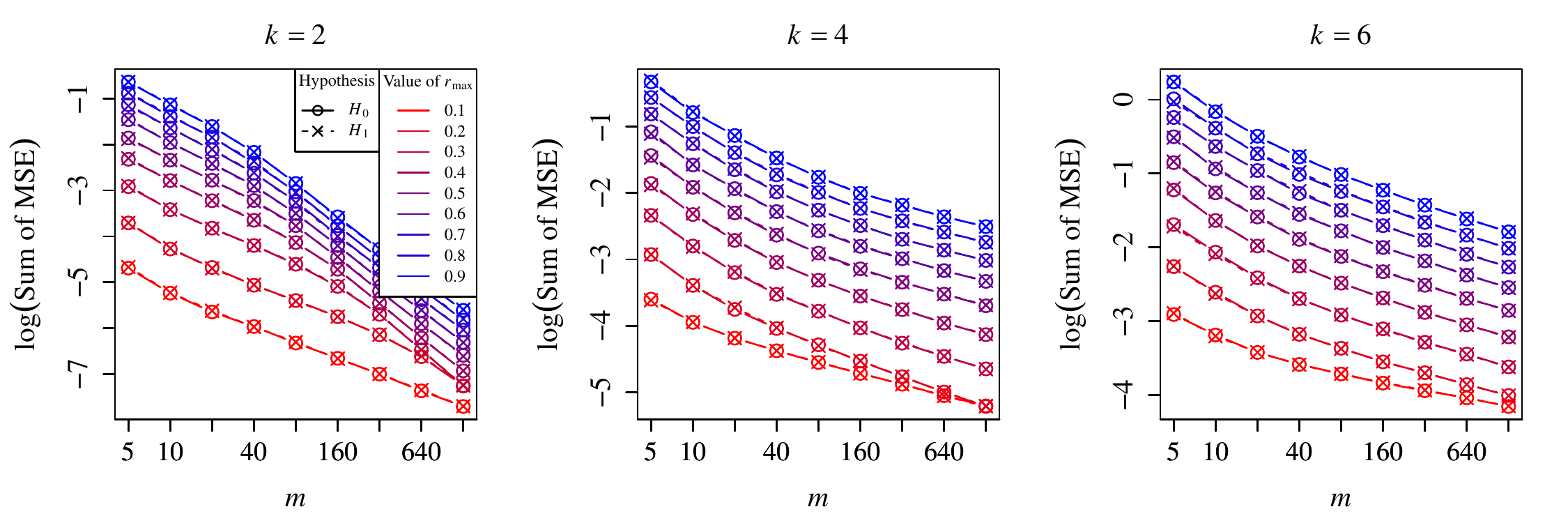}
\end{center}
\vspace{-0.5cm}
\caption{\footnotesize The log of sum of MSEs of $\widehat{r}_{1:k}$, i.e.,
$\log \sum_{j=1}^k \E(\widehat{r}_j - r_j)^2$; see Example \ref{eg:rHat}. } 
\label{fig:SMI_L025-030_ind_smallk}
\end{figure}

The sum of the MSEs of $\widehat{r}_{1:k}$, that is, 
$E := \sum_{j=1}^k \E(\widehat{r}_j - r_j)^2$
is shown in Figure \ref{fig:SMI_L025-030_ind_smallk}. 
The values of $E$ under $H_0$ and $H_1$ are nearly identical, implying  
that it is safe to use $\widehat{r}_{1:k}$,
no matter $H_0$ is true or not. 
Second, the value of $E$ decreases when $m$ increases. 
It verifies Corollary \ref{coro:rHat}. 
However, the performance of $\widehat{r}_{1:k}$ declines when $r_{\max}$ or $k$ increases. 
It is reasonable as the CV of $r_{1:k}$ increases with $r_{\max}$, 
and the number of estimands ($r_{1:k}$) increases with $k$.
In either case, the estimation problem is harder by nature.
\end{example}

\section{General multiple imputation procedures}\label{sec:testing}
\subsection{Hypothesis testing of model parameters}
We denote $\widehat{D}$ by $\widehat{D}_{\wt},\widehat{D}_{\lrt},\widehat{D}_{\rt}$ 
to emphasize that 
$\mathcal{d}=\mathcal{d}_{\wt},\mathcal{d}_{\lrt},\mathcal{d}_{\rt}$ is used, respectively.
The limiting distribution of $\widehat{D}$ is stated below.

\begin{proposition}\label{prop:limit_of_Dhat}
Assume Conditions \ref{cond-Normality}--\ref{cond-properImp}.  
Let $\widehat{D}\in \{\widehat{D}_{\wt}, \widehat{D}_{\lrt},\widehat{D}_{\rt}\}$ and $m>1$. 
Under $H_0$, we have, as $n\rightarrow\infty$, that 
(1) $\widehat{D}- \widetilde{D} \inP 0$, 
			where 
			$\widetilde{D}\in\{\widetilde{D}_{\wt},\widetilde{D}_{\lrt}\}$; and 
(2) $\widehat{D} \inD \mathbb{D}$, where 
			$\mathbb{D}$ is defined in  (\ref{eqt:representation_D_limit}).
\end{proposition}

Proposition \ref{prop:limit_of_Dhat} states that 
$\widehat{D}$ and $\widetilde{D}$ are asymptotically ($n\rightarrow\infty$) equivalent
for any $m$, and have the same 
limiting null distribution $\mathbb{D}$.
We emphasize again that computing $\widetilde{D}$ is not feasible
as it requires problem-specific devices other than $\mathcal{d}(\cdot)$.
However, computing our proposed $\widehat{D}$ requires only $\mathcal{d}(\cdot)$.

We propose approximating the limiting null distribution $\mathbb{D}$ by
substituting $r_{j}=\widehat{r}_{j}$ into (\ref{eqt:representation_D_limit}). 
Although it is not a named distribution, 
we can easily compute its quantile via Monte Carlo methods. 
Precisely, we first generate
$G_j^{(\iota)}\sim \chi^2_1$ and $H_j^{(\iota)}\sim \chi^2_{m-1}/(m-1)$ independently 
for $j=1, \ldots, k$ and $\iota=1, \ldots, N$. 
Upon conditioning on $\widehat{r}_{1}, \ldots, \widehat{r}_{k}$, 
we can generate $N$ random replicates of $\widehat{\mathbb{D}}$ as follows:
\begin{equation}\label{eqt:DmHat_est}
	\mathbb{\widehat{D}}^{(\iota)} := \frac{\frac{1}{k}\sum_{j=1}^k \{ 1+ (1+\frac{1}{m})\widehat{r}_{j}\} G_j^{(\iota)}}{ 1+\frac{1}{k}\sum_{j=1}^k (1+\frac{1}{m})\widehat{r}_{j} H_j^{(\iota)}}, 
	\qquad \iota=1, \ldots, N.
\end{equation}
The $100(1-\alpha_0)\%$ quantile of $\mathbb{D}$ can then be estimated by the 
$100(1-\alpha_0)\%$ sample quantile of $\{\widehat{\mathbb{D}}^{(1)},\ldots, \widehat{\mathbb{D}}^{(N)}\}$, 
where $\alpha_0 \in (0, 1)$. 
The sample quantile can be served as a critical value for testing $H_0:\theta=\theta_0$.
Similarly, 
the $p$-value can be found by 
$
	\widehat{p} = \sum_{\iota=1}^N \mathbb{1} \{\widehat{\mathbb{D}}^{(\iota)} \geq \widehat{D}  \}/N. 
$ 
The null hypothesis $H_0$ is rejected at size $\alpha_0$ if $\widehat{p} < \alpha_0$.
We emphasize that the proposed MI test
is asymptotically correct 
with or without Condition \ref{ass:EOMI}. 
Step-by-step procedure for computing 
$\widehat{p}$ is presented in Algorithm \ref{algo:SMI_exact}.

In Section \ref{sec:approxMethod} of the supplementary note,  
we present several alternative approximation schemes by 
projecting $\mathbb{D}$ to some distributions that depend only on
$\mu_r$ and $\sigma^2_r$ instead of $r_{1:k}$. 
This idea is similar to \cite{mengPhDthesis}.
Such approximations can be used if one only wants to estimate $\mu_r$ and $\sigma^2_r$.
Although these approximations are algorithmically simpler,  
the resulting MI tests control type-I error rates substantially worse than 
our proposal (\ref{eqt:DmHat_est}).
A quick simulation example is presented in Section \ref{sec:theoretical_test_size}
of the supplementary note for illustration.

\subsection{Discussion on applications}\label{sec:applicationsDiscussion}
Our proposed method uses $p$-value as a one-number summary for assessing 
variability of estimators in the presence of missing data. 
It can be applied not only to hypothesis testing but also 
other statistical procedures that require variability assessment
or use $p$-value as a part of the automatic procedures.   

Let $\widehat{p}(\theta_0)$ be the $p$-value returned by Algorithm \ref{algo:SMI_exact} 
for testing $H_0:\theta=\theta_0$.
The function $\widehat{p}(\cdot):\Theta \rightarrow [0,1]$ is called 
a $p$-value function \citep{Fraser2019}. 
It measures the degree of falsity of $H_0:\theta=\theta_0$.
Similar concepts include confidence curves \citep{Birnbaum1961}, 
confidence distributions \citep{XieSinghStrawderman2011,XieSingh2013}, 
significance functions \citep{Fraser1991},
plausibility functions \citep{Martin2015}, etc.
There are many automatic procedures that are built on the $p$-value function; 
see, for example, \citet{Martin2017}. 
We are not able to exhaust all applications.
Only four examples are presented here.

First, 
one obvious application is confidence regions (CR) construction. 
By the duality of hypothesis testing and CR (Section 5.4 of \cite{TSH_lehmannRomano}), 
a $100(1-\alpha_0)\%$ CR for $\theta$ is 
$\mathcal{C}=\{\theta_0 \in\Theta : \widehat{p}(\theta_0) \geq \alpha_0 \}$, 
which can be obtained by repeatedly using Algorithm \ref{algo:SMI_exact}.
Second, if researchers prefer not to fix the confidence level in advance, 
it is possible to report the $p$-value function as an estimator of $\theta$; 
see \cite{InfangerSchmidt2019} for some examples in medical studies. 
Third, 
$p$-value has been a commonly-used tool for combining evidence in meta-analysis 
\citep{HeardHeard2018}. 
For example, if the $p$-value for testing $H_0:\theta=\theta_0$ by the $g$th dataset is 
$\widehat{p}_g(\theta_0)$ for $g=1,\ldots G$, a possible combined $p$-value is 
$\sum_{g=1}^G \log \widehat{p}_g(\theta_0)$ \citep{Fisher1934}.
Fourth, $p$-values can be used for stepwise variable selection in generalized linear model 
(Section 4.6.1 of \citet{Agresti2015}).

In a nutshell, 
classical Rubin's rule use variance as a medium for assessing uncertainty, 
whereas our proposed method use $p$-value.
Our proposal is useful not only for the standard null hypothesis testing but also for other statistical procedures
that require variability assessment.

\section{Monte Carlo experiments and applications}\label{sec:applications}
Incomplete-data testing of linear regression coefficient 
and region estimation of a probability vector are presented in Sections \ref{sec:eg_reg} and \ref{sec:CR}, 
respectively.
Due to space constraints,
some additional simulation experiments and real-data examples are deferred to the supplementary note. 
They include 
(i) inference of variance-covariance matrix in Section \ref{sec:eg_mvn}, 
(ii) variable selection in generalized linear model in Section \ref{variable_selection_glm},
(iii) contingency tables in Sections \ref{sec:eg_contingency_table}, and 
(iv) logistic regression in Section \ref{sec:logisticEG}.
  
\subsection{Linear regression}\label{sec:eg_reg}
Let $y_i$ and $x_i=(x_{i1},\ldots, x_{ip})^{\T}$ be a univariate random response and $p$ deterministic covariates of the $i$th unit, respectively, where $i=1,\ldots,n$. 
Consider the linear regression model: $y_i = \beta_0 + x_i^{\T}\beta + \epsilon_i$
for each $i$, 
where $\epsilon_1, \ldots, \epsilon_n \simIID \Normal(0,\sigma^2)$. 
The full set of model parameters is $\psi = (\beta_0, \beta^{\T}, \sigma)^{\T}$; and  
the parameter of interest is $\theta = \beta$. 
We want to test $H_0: \theta=0_p$ against $H_1:\theta\neq 0_p$. 
Clearly, in this case, $k=p$. 

For each $i$, suppose $y_i$ and $x_{i1}$ are always observed, while $x_{i2}, \ldots, x_{ip}$ 
may be missing. 
Let $I_{ij}=0$ if $x_{ij}$ is missing, otherwise $I_{ij}=1$. 
Suppose further that $I_{ij}$ follows a logistic regression model:  
\begin{equation}\label{eqt:missingDataMech}
	\pr\left( I_{ij} = 1 \mid x_{i,j-1},I_{i,j-1} =a\right) 
		= \expit\left(  \gamma_0 + \gamma_1 x_{i,j-1} \right) \mathbb{1}(a=1)
\end{equation}
for $i=1,\ldots, n$ and $j=2, \ldots, p$, where $\expit(t) := 1/(1+e^{-t})$.
The missing data $X_{\mis} = \{x_{ij} : I_{ij}=0\}$ are then imputed $m\in\{10,30\}$ times  
by a Bayesian model; see 
Section \ref{sec:imputationModel_regression} of the 
supplementary note for details.

In the simulation experiment, 
all three devices $\mathcal{d}_{\wt}$, $\mathcal{d}_{\lrt}$ and $\mathcal{d}_{\rt}$
are studied; 
see Section \ref{sec:imputationModel_regression} for their formulas. 
We compute the MI test statistics $\widehat{D}_{\wt}$, $\widehat{D}_{\lrt}$ and $\widehat{D}_{\rt}$,
and refer them to 
four different approximated null distributions in Table \ref{tab:approxNullDist}, that is,  
T1 ($\chi^2_k/k$), 
T2 \citep{li91JASA}, 
T3 \citep{mengPhDthesis}, and 
T4 (the proposal in Algorithm \ref{algo:SMI_exact}). 
Note that the unknowns $\mu_r$ and $\sigma^2_r$ in T2 and T3 are estimated by (\ref{eqt:estimator_mur_sigmar})
with the devices $\mathcal{d}_{\wt}$, $\mathcal{d}_{\lrt}$ and $\mathcal{d}_{\rt}$
for $\widehat{D}_{\wt}$, $\widehat{D}_{\lrt}$ and $\widehat{D}_{\rt}$, respectively.

We consider $n=1000$, $p=5$, 
$x_i\simIID \Normal_p( 1_p , \Sigma_x)$ with 
$\left( \Sigma_{x} \right)_{ab}=2^{-|a-b|}$ for each $1\leq a,b\leq p$, 
$\sigma^2=1$ and $\beta_0 = 1$, where $1_p$ is a $p$-vector of ones.  
Two sets of $(\gamma_0,\gamma_1) \in\{ (1,0), (0,1) \}$ are considered. 
Note that the data are missing completely at random (MCAR) when $\gamma_1=0$, whereas 
the data are missing at random (MAR) when $\gamma_1=1$. 
Note that the fractions of missing covariates 
are about $(0,16\%, 29\%, 41\%, 50\%)$ and $(0, 24\%, 40\%,51\%, 61\%)$
in the MCAR and MAR cases, respectively.

\begin{figure}
\centering
\includegraphics[width=\linewidth]{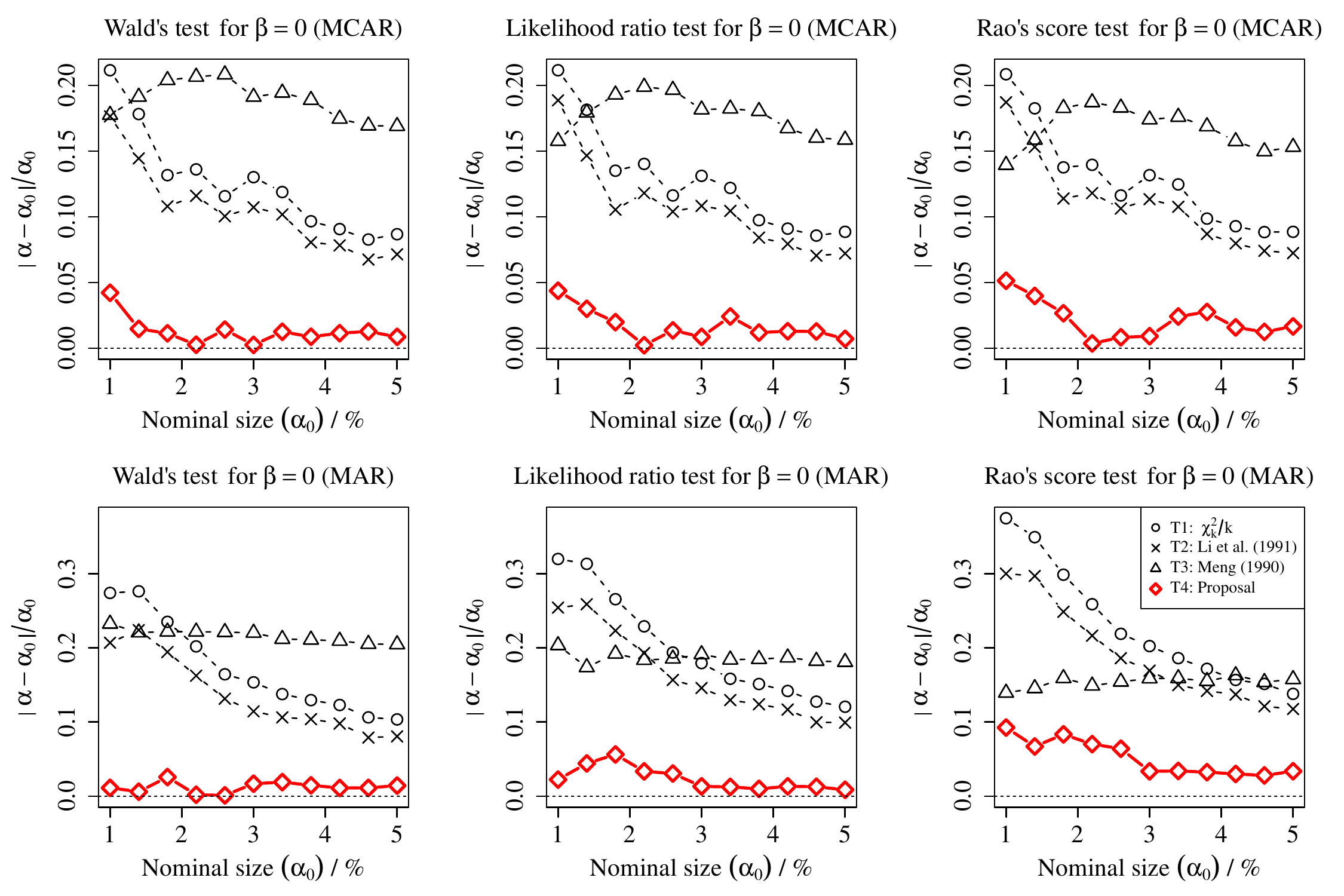}
\vspace{-0.5cm}
\caption{\footnotesize Size accuracy of the MI tests regarding regression coefficients in Section \ref{sec:eg_reg}.}
\label{fig:SMI_I34_H0}
\end{figure}

We record the sizes of the tests (denoted by $\alpha$)
at various nominal size $\alpha_0 \in [1\%, 5\%]$.
The results for $m=30$ are shown in Figure \ref{fig:SMI_I34_H0}. 
The proposed test T4 controls the size substantially 
more accurately than all other competitors in all cases. 
Although the widely used T2 performs well when Condition \ref{ass:EOMI} is true (see \cite{li91JASA}),
it only has a marginal improvement over T1 when Condition \ref{ass:EOMI} does not hold. 
Besides, our proposed test is a unified method for 
Wald's test, LR test, and RS test.
So, users only need to change the testing device 
(i.e., $\mathcal{d}_{\wt}, \mathcal{d}_{\lrt}, \mathcal{d}_{\rt}$) in Algorithm \ref{tab:approxNullDist}.
It is worth mentioning again that our proposal is the only MI procedure that 
can handle RS tests.

The results for $m=10$ are deferred to Figure \ref{fig:SMI_I33_H0_full}
of the supplementary note.  
The patterns are similar to Figure \ref{fig:SMI_I34_H0} except for T3. 
Note that T3 is not trustworthy unless $m$ is large. %as  
In this example, its performance converges to an increasingly bad state as $m$ increases.  
The powers of the MI tests are not directly comparable 
as their sizes are not equally accurate. 
Their size-adjusted powers are identical 
because they are based on the same test statistic $\widehat{D}$ (or its asymptotic equivalent). 
For reference, the power curves are shown in Section \ref{sec:imputationModel_regression}
of the supplementary note. 
Additional discussions about the effects of $m$ are deferred to 
Section \ref{sec:imputationModel_regression}.

\subsection{Confidence region and $p$-value function of a probability vector}\label{sec:CR}
Let $[y_i \mid x_i] \sim \Bern(\theta_1^{x_i}\theta_0^{1-x_i})$ independently for $i=1, \ldots, n$, 
where $x_1, \ldots, x_n \in \{0,1\}$ are fixed binary covariates, and $\theta_0, \theta_1\in (0,1)$
are unknown parameters.
The covariates $x_1, \ldots, x_n$ are always observed, but the responses $y_1, \ldots, y_n$ may be missing. 
If $y_i$ is observed, we denote $I_i=1$, otherwise $I_i=0$. 
Suppose that the missing mechanism is $[I_i\mid x_i] \sim \Bern( \pi_1^{x_i} \pi_0^{1-x_i})$, 
where $\pi_0,\pi_1\in(0,1)$ are unknown. 
The goal is to estimate $\theta = (\theta_0,\theta_1)^{\T}$ with the incomplete dataset 
$\{(x_i, y_i I_i) \}_{i=1}^n$.
In the simulation study, $n=100$ and around 40\% of the $x_i$'s are 1. 
The unknown true values are $\theta_0 = 0.15$, $\theta_1 = 0.75$, $\pi_0 = 0.9$ and $\pi_1 = 0.1$.
Note that $y_i$ is very likely to be missing when $x_i=1$. 
As discussed in Section \ref{sec:applicationsDiscussion}, 
we could construct a CR or $p$-value function as an estimator of $\theta$. 
We try both in this example.

Traditionally, one may construct the Wald's CR by Rubin's rule:
\[
	\mathcal{C}_{\wt} = \left\{\theta \in \mathbb{R}^2: (\bar{\theta}- \theta)^{\T}\left(\widehat{B}+\frac{m+1}{m}\bar{V}\right)^{-1}(\bar{\theta}- \theta) \leq c \right\}, 
\]
where 
the critical value $c$ can be found according to \cite{li91JASA}. 
It is well-known that the Wald's CR $\mathcal{C}_{\wt}$ must be an ellipse, 
which is restrictive in some problems. 
Moreover, $\mathcal{C}_{\wt}$ may not be a subset of the support $\Theta = (0,1)^2$.
Figure \ref{fig:SMI_T3_802} (a) visualizes these two problems. 
Alternatively, one may invert the LR test to construct a CR 
$\mathcal{C}_{\lrt}$ (say). 
It does not have the two aforementioned structural problems; 
see Figure \ref{fig:SMI_T3_802} (a) again.

We assess the performance of 
the likelihood-based $100(1-\alpha_0)\%$ CRs for $\theta$ by using T1--T4. 
We compute the actual non-coverage rates $\widehat{\alpha}$ for different methods, 
and compare them with the nominal value $\alpha_0$. 
The relative error $|\widehat{\alpha}-\alpha_0|/\alpha_0$ is reported for each method in Figure \ref{fig:SMI_T3_802} (b)
under various $\alpha_0 \in[0.01, 0.05]$.
The relative error of the Wald's CR is also computed for reference. 
Our proposed method has the lowest error uniformly. 
Figure \ref{fig:SMI_T3_802} (c) shows the $p$-value function $\widehat{p}(\theta)$ 
produced by our proposed T4 and Algorithm \ref{algo:SMI_exact}. 
It offers an alternative way for estimating $\theta$ with variability assessment
but without specifying the confidence level in advance.  
From the above example, our proposed method is particularly useful for 
handling non-normal data and parameters with bounded supports.

\begin{figure}
\centering
\includegraphics[width=1\linewidth]{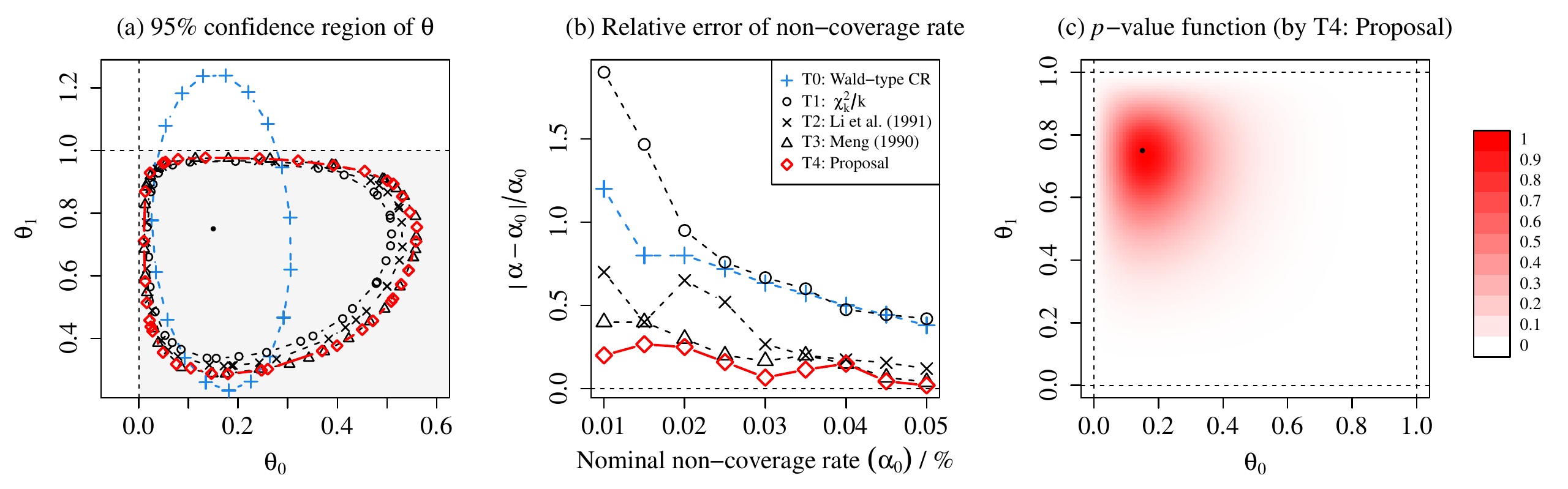}
\vspace{-0.5cm}
\caption{\footnotesize (a) One typical realization of 95\% CRs of $\theta$ using different methods, where the grey region is the support $(0,1)^2$ of $\theta$.
(b) Coverage accuracy of different CRs. (c) Heat-map of the $p$-value function by using the 
proposed T4. 
The solid black dots in plots (a) and (c) denote the true value of $\theta=(0.15, 0.75)^{\T}$. 
See Section \ref{sec:CR} for detailed descriptions.}
\label{fig:SMI_T3_802}
\end{figure}

\section{Conclusion, discussion and future work}\label{sec:conclusion}
The proposed test for handling multiply-imputed datasets is general as 
it does not require $m\rightarrow \infty$ or equal odds of missing information. 
So, it is particularly suitable for handling 
public-use datasets that have non-trivial missingness structures; 
see Remark \ref{rem:imputer} and Section \ref{sec:applications} for some examples. 
The test is feasible in the sense that 
only a standard complete-data testing device is needed 
for performing incomplete-data Wald's, likelihood ratio, and Rao's score tests; 
see Algorithm \ref{algo:SMI_exact}.  
Besides, the proposed method is also useful for 
general statistical procedures that use $p$-value as an assessment tool; 
see Section \ref{sec:applicationsDiscussion}. 
So, it has a wide range of applications.
Although the proposed test improves the existing counterparts, 
further studies are needed in the following directions. 

First, this paper assumes the sample size $n$ is large enough so that 
the standard large-sample $\chi^2$ approximation kicks in. 
Small-sample approximations (e.g., approaches similar to \cite{BarnardRubin1999} and \cite{Reiter07})
are likely to further improve the proposed test. 
Second, 
we need to perform more computationally expensive tests on several stacked datasets.
We increase the computing cost 
in order to minimize the human time cost needed to 
build non-standard computing functions required in tests T2 and T3; see Table \ref{tab:approxNullDist}. 
Given the computing power of the current computers, 
we believe that computing time cost is a lesser constraint than human time cost. 
However, it is still desirable to further reduce the computing cost.

%%%%%%%%%%%%%%%%%%%%%%%%%%%%%%%%%%%%%%%%%%%%%%
%% Support information, if any,             %%
%% should be provided in the                %%
%% Acknowledgements section.                %%
%%%%%%%%%%%%%%%%%%%%%%%%%%%%%%%%%%%%%%%%%%%%%%
\begin{acks}[Acknowledgments]
A part of the theoretical results in this article are partially developed from the author's Ph.D. thesis 
under the supervision of Xiao-Li Meng, who provided many insightful ideas that greatly contribute to this paper. 
The author would also like to thank the anonymous referees, an Associate
Editor and the Editor for their constructive comments that improved the
scope and presentation of the paper. 
\end{acks}

%%%%%%%%%%%%%%%%%%%%%%%%%%%%%%%%%%%%%%%%%%%%%%
%% Funding information, if any,             %%
%% should be provided in the                %%
%% funding section.                         %%
%%%%%%%%%%%%%%%%%%%%%%%%%%%%%%%%%%%%%%%%%%%%%%
\begin{funding}
The author acknowledges the financial support from the 
Early Career Scheme (24306919) provided by the University Grant Committee of Hong Kong.
\end{funding}

\smallskip
\begin{center}
{SUPPLEMENTARY MATERIAL}
\end{center}
The supplementary note includes  
proofs, supplementary results and additional examples.
An R-package \texttt{stackedMI} is also provided.

%%%--------------------------------------------------------------------------------------
%%%--------------------------------------------------------------------------------------
%%%
%%% Bibliography 
%%%
%%%--------------------------------------------------------------------------------------
%%%--------------------------------------------------------------------------------------
\bibliographystyle{rss}
\bibliography{myRef.bib}

\newpage
\begin{center}
\bfseries\MakeUppercase{Supplementary Note to ``General and Feasible Tests with Multiply-Imputed Datasets''}
\end{center}
\appendix

\section{Auxiliary and additional results}\label{sec:auxiliary_results}

\subsection{Unbiased Estimation of Mean and Variance of OMI}\label{sec:estimator_of_moments_r}
Recall that 
the mean and variance of $r_{1}, \ldots, r_k$ are defined as 
\begin{equation*}\label{eqt:moments_r}
	\mu_r := \frac{1}{k}\sum_{j=1}^k r_j = R_1/k
	\qquad \text{and}  \qquad
	\sigma^2_r := \frac{1}{k}\sum_{j=1}^k(r_j - \mu_r)^2 = R_2/k - \mu_r^2.
\end{equation*}
These two values are required for approximating 
the limiting null distribution $\mathbb{D}$ in (\ref{eqt:representation_D_limit}) for testing $H_0$, 
for example, T2 and T3 in Table \ref{tab:approxNullDist}. 
In this section, we derive their asymptotically unbiased estimators.

By Theorems \ref{prop:limit_of_J} and \ref{thm:limit_EJ}, 
$\widehat{t}_1 \inP R_1$ and $\widehat{t}_2 \inP 2R_2 + R_1^2$. 
So, $\mu_r$ and $\sigma^2_r$ can be trivially estimated by
\begin{eqnarray}\label{eqt:muHatr}
	\widehat{\mu}_r := \frac{\widehat{t}_1}{k}	
	\qquad \text{and} \qquad
	\widecheck{\sigma}^{2}_r := \frac{\widehat{t}_2}{2k} - \frac{(2+k)\widehat{t}_1^2}{2k^2},	
\end{eqnarray}
respectively;
see Remark \ref{rem:asy_equiv_of_murHat} for an asymptotically equivalent definition of 
$\widehat{\mu}_r$. 
It is not hard to see that the estimator $\widehat{\mu}_r$ is asymptotically unbiased for $\mu_r$
even for finite $m$. 
However, $\widecheck{\sigma}^{2}_r$ can be severely biased for $\sigma^2_r$ when $m$ is small; 
see Table \ref{table:SMI_L025-030_bias} of Example \ref{eg:varHat} presented below.  
It is desirable to have a de-biased estimator of $\sigma^2_r$.
Since we have the detailed moment properties of 
$\widehat{t}_{\tau} = \sum_{\ell=1}^m \widehat{T}_{\ell}^{\tau}/m$
in Proposition \ref{prop:limit_of_J}, 
we can correct the bias of $\widecheck{\sigma}_{r}^{2}$.
If $m\geq 3$, 
the bias-corrected estimator of $\sigma^2_r$ is given by 
\begin{eqnarray*}
	\widehat{\sigma}_{r}^2 
		:=  \frac{\left\{k(m-1)+2\right\} \widehat{t}_2
			-(m-1)(k+2)\widehat{t}_1^2}{2k^2(m-2)} . 
\end{eqnarray*}
The following corollary shows its asymptotic ($n\rightarrow\infty$) 
properties for $m\geq 3$.

\begin{corollary}\label{coro:unbiasedEstofVarlambda}  
Assume Conditions \ref{cond-Normality}--\ref{cond-UI}. 
Let $m\geq 3$ be fixed and $\Lambda=\Lambda_{\Jack}$.
As $n\rightarrow\infty$, we have 
(1) $\E( \widehat{\sigma}_{r}^2  ) \rightarrow \sigma^2_r$, and 
(2) 
\[
	\Var( \widehat{\sigma}_{r}^2  )
		\rightarrow {2\left( 6k^2R_4 -8k R_3R_1 + k^2R_2^2 + 4 R_1^2 R_2\right)}/(k^4m)
			+ I_m,
\]
where $I_m$ is a reminder term such that $I_m = o(1/m)$.
\end{corollary}

Corollary \ref{coro:unbiasedEstofVarlambda} states that 
$\widehat{\sigma}_{r}^2$ is asymptotically unbiased, and  
its variance decreases in the rate of $O(1/m)$. 
So, $\widehat{\sigma}_{r}^2$ is a good estimator of $\sigma^2_r$.
Since $\widecheck{\sigma}_{r}^{2}$ and $\widehat{\sigma}_{r}^2$
are not guaranteed to be non-negative, 
we may use 
$\widecheck{\sigma}_{r+}^2 := \max\{0,\widecheck{\sigma}_{r}^{2}\}$ and $\widehat{\sigma}_{r+}^2:=\max\{0,\widehat{\sigma}_{r}^2\}$ in practice.
Besides,  $\widecheck{\sigma}^{2}_r$ and $\widehat{\sigma}^{2}_r$
are asymptotically equivalent 
when $m,n\rightarrow \infty$.
So, the bias improvement of $\widehat{\sigma}^{2}_r$ is substantial only when $m$ is small. 
We also remark that the exact limit of $\Var( \widehat{\sigma}_{r}^2 )$ is complicated:
\begin{align*}
	&\lim_{n\rightarrow\infty}
		\Var\left(\widehat{\sigma}_{r}^2 \right)\\
		&=\frac{18(m-3)(k+2)^2}{k^4(m-1)^5m(m-2)^2}R_4  \nonumber\\
		&\quad
			+ \frac{2}{k^4(m-1)^3(m-2)} 
			\left[\begin{array}{c}
				2\left\{(m-1)^2(3m-7)k^2 - 4(m-1)k -4\right\}\\
				-8(m-1)^2(m-2)k\\
				(m-2)\left\{ (m-1)^2k^2 + 4(m-1)k+4m\right\}\\
				4(m-1)^2(m-2)
			\end{array}
			\right]^{\T}
			\left[\begin{array}{c}
				R_4\\
				R_3R_1\\
				R_2^2\\ 
				R_1^2R_2
			\end{array}
			\right].\nonumber
\end{align*}

An application of $\widehat{\mu}_r$ and $\widehat{\sigma}^2_r$
is to approximate $\mathbb{D}$ in (\ref{eqt:representation_D_limit}).
Recall that the reference null distributions of the existing 
MI tests T2 and T3 in Table \ref{tab:approxNullDist}
require estimating $\mu_r$ and/or $\sigma^2_r$. 
However, in the literature, they cannot be estimated solely 
by the complete-data testing device $\mathcal{d}(\cdot)$. 
Consequently, those MI tests cannot be carried out 
if users are only equipped with $\mathcal{d}(\cdot)$.
In practice, this situation is not uncommon.   
So, our proposed method provides a viable route to implement existing MI tests.

\begin{example}\label{eg:varHat}
We consider the same setting as in Example \ref{eg:rHat}.
First, we compare the biases of $\widecheck{\sigma}_{r}^2$ and $\widehat{\sigma}_{r}^2$, 
that is, $\Bias(\widecheck{\sigma}^2_r) := \E(\widecheck{\sigma}^2_r)-\sigma^2_r$ and 
$\Bias(\widehat{\sigma}^2_r) := \E(\widehat{\sigma}^2_r)-\sigma^2_r$, respectively.  
The results are shown in Table \ref{table:SMI_L025-030_bias}. 
Note that $\widecheck{\sigma}_{r}^2$ is negatively biased for $\sigma^2_r$,
but $\widehat{\sigma}_r^2$ only has a negligible bias.

\begin{table*}[t]
\centering
\setlength{\tabcolsep}{3pt}
\footnotesize
\begin{tabular}{llrrrrrrrrr}\toprule
	&& \multicolumn{9}{c}{Value of $r_{\max}$}\\
	\cline{3-11}
	$m$ & Estimator  & $0.1$ & $0.2$ & $0.3$ & $0.4$ & $0.5$ & $0.6$ & $0.7$ & $0.8$ & $0.9$  \\
	\cmidrule(r){1-11}
	$5$ & $\widecheck{\sigma}^2_r$ (Biased)	&	$ -0.37$	&	$ -0.90$	&	$ -1.71$	&	$ -2.81$	&	$ -4.24$	&	$ -5.91$	&	$ -7.96$	&	$-10.09$	&	$-12.51$	\\ 
	& $\widehat{\sigma}^2_r$ (Unbiased) &	$  0.01$	&	$ -0.01$	&	$  0.00$	&	$ -0.03$	&	$  0.01$	&	$  0.12$	&	$ -0.08$	&	$ -0.02$	&	$  0.27$	\\ 
	\hline
	$10$ & $\widecheck{\sigma}^2_r$ (Biased)&	$ -0.17$	&	$ -0.41$	&	$ -0.77$	&	$ -1.20$	&	$ -1.87$	&	$ -2.64$	&	$ -3.56$	&	$ -4.51$	&	$ -5.75$	\\ 
	 & $\widehat{\sigma}^2_r$	(Unbiased) &	$  0.00$	&	$ -0.02$	&	$ -0.01$	&	$  0.06$	&	$  0.00$	&	$ -0.04$	&	$ -0.08$	&	$  0.02$	&	$ -0.16$	\\ 
	\bottomrule
\end{tabular}
\caption{\footnotesize The biases of $\widecheck{\sigma}^2_r$ and $\widehat{\sigma}^2_r$
after multiplying by $100$ 
when $m\in\{5,10\}$ and $r_{\max}\in\{0.1, \ldots, 0.9\}$.}
\label{table:SMI_L025-030_bias}
\end{table*}

Next, we study  
the MSE of $\widehat{\sigma}_{r}^2$, 
whose value in the log-scale, 
that is, $\log\E(\widehat{\sigma}_{r}^2-\sigma_r^2)^2$, is shown in 
Figure \ref{fig:SMI_L025-030_moment_smallk}.
Clearly, the MSE decreases when $m$ increases.
Besides, 
the MSE is less impacted by an increase of $k$
compared with Figure \ref{fig:SMI_L025-030_ind_smallk}. 
It is because we just need to estimate one parameter $\sigma^2_r$ for any $k$.

\begin{figure}
\begin{center}
\includegraphics[width=\linewidth]{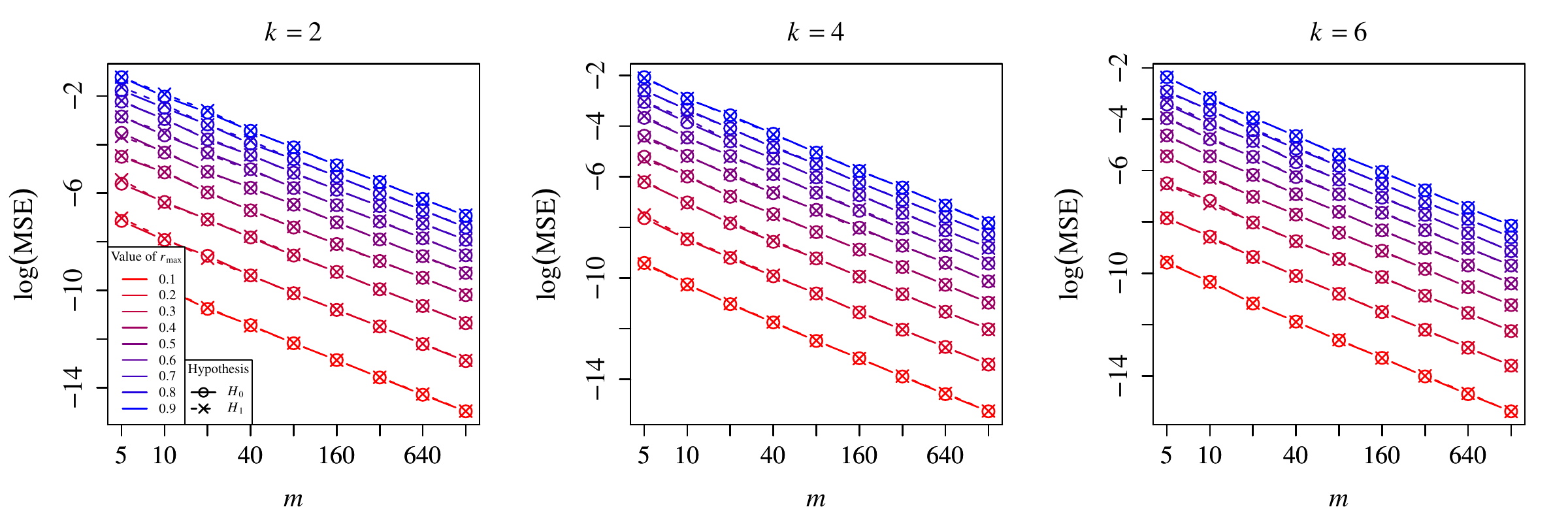}
\end{center}
\vspace{-0.7cm}
\caption{\footnotesize The log MSE of $\widehat{\sigma}_{r}^2$ in the log-scale, i.e.,
$\log\E(\widehat{\sigma}_{r}^2-\sigma_r^2)^2$. } 
\label{fig:SMI_L025-030_moment_smallk}
\end{figure}
\end{example}

\begin{remark}\label{rem:asy_equiv_of_murHat}   
Using Theorem \ref{thm:asy_Distof_bar_d_S},  
one can easily verify that $\widehat{\mu}_r$ in (\ref{eqt:muHatr})
is asymptotically equivalent to 
\begin{equation*}
	\overline{\mu}_r = \frac{\widehat{d}'-\widehat{d}''}{k(m-1)/m},
	\quad \text{where} \quad
	\widehat{d}' := \frac{1}{m} \sum_{\ell=1}^m \widehat{d}^{\{\ell\}},
	\qquad
	\widehat{d}'' := \widehat{d}^{\{1:m\}},
	\label{eqt:barmurDef}
\end{equation*}
for any $m$,
as $n\rightarrow\infty$.
So, one may use $\widehat{\mu}_r$ and $\overline{\mu}_r$ interchangeably 
without altering the asymptotic behaviors.  
Since computing $\widetilde{\mu}_{r,\wt}$ and $\widetilde{\mu}_{r,\lrt}$ defined in (\ref{eqt:Da})
require problem-specific device other than $\mathcal{d}(\cdot)$,
one may regard $\widehat{\mu}_r$ and $\overline{\mu}_r$ as computationally feasible alternatives 
to $\widetilde{\mu}_{r,\wt}$ and $\widetilde{\mu}_{r,\wt}$.
\end{remark}

\subsection{Hypothesis testing of the EOMI assumption}\label{sec:testEOMI}
This section provides a test for Condition \ref{ass:EOMI}.
Formally, we would like to test $H_0': \sigma^2_r = 0$ against $H_1': \sigma^2_r > 0$. 
Remark \ref{rem:EOMItest} discusses the purposes of testing $H_0'$.
Let the CV of $r_1, \ldots, r_k$ be
$c_r:=\sigma_r/\mu_r$, where $0/0 := 0$.
Since $c_r^2$ normalizes $\sigma^2_r$ by $\mu_r^2$, the quotient $c_r^2$ is scale-free.
For testing $H_0'$, 
we proposed the test statistic
\begin{eqnarray}\label{eqt:def_w}
	\widehat{Q} 
		:= \sqrt{m} {\widehat{\sigma}_r^2}/{\widehat{\mu}_r^2}
\end{eqnarray} 
if $\widehat{\mu}_r\neq 0$, and $\widehat{Q}=0$ if otherwise.
The limiting distribution of $\widehat{Q}$ under both $H_0'$ and $H_1'$ are stated below. 

\begin{proposition}\label{prop:wHat}
Assume Conditions \ref{cond-Normality}--\ref{cond-localH1}.
Define $Z_{j\ell}$'s according to Theorem \ref{thm:asy_Distof_bar_d_S}. 
We have    
\begin{equation}\label{eqt:conv_of_u}
	\widehat{Q} \inD   \mathbb{Q} := \sqrt{m}\left\{ \frac{\{k(m-1)+2\}m}{2(m-2)}\frac{\sum_{\ell=1}^m M_{\ell}^2}{\left(\sum_{\ell=1}^m M_{\ell}\right)^2} - \frac{(m-1)(k+2)}{2(m-2)} \right\},
\end{equation}
as $n\rightarrow\infty$ for any $m\geq 3$, 
where $M_{\ell} := \sum_{j=1}^k r_j (Z_{j\ell} - \bar{Z}_{j\bullet})^2$ and $\bar{Z}_{j\bullet} := m^{-1}\sum_{\ell=1}^m Z_{j\ell}$. 
If $H_0'$ is true, 
$\mathbb{Q}$ reduces to $\mathbb{Q}_0$
which is defined as in (\ref{eqt:conv_of_u}) except that 
$M_{\ell}$ is replaced by $M_{0\ell} := \sum_{i=1}^k  (Z_{i\ell} - \bar{Z}_{i\bullet})^2$.
\end{proposition}

We emphasize that proposition \ref{prop:wHat} is true even for any $m\geq 3$.
Under $H_0'$, $\widehat{Q}\inD\mathbb{Q}_0$, which is a pivotal distribution 
because of self normalization. 
It is desirable because the quantiles of $\mathbb{Q}_0$ 
can be tabulated. 
Some commonly used quantiles are shown in Table \ref{tab:qQ0}.
The hypothesis $H_0'$ is rejected at size $\alpha_0\in(0,1)$ if $\widehat{Q} > q_{1-\alpha_0}$, where 
$q_{1-\alpha_0}$ is the $1-\alpha_0$ quantile of $\mathbb{Q}_0$.
Alternatively, we can compute the $p$-value by simulation
as in Algorithm \ref{algo:SMI_exact}.
The power of the test is investigated in Example \ref{eg:power_of_test_of_EFMI} below.

\begin{table}
\centering
\footnotesize
\begin{tabular}{cccccccccc}
\toprule
$m\backslash k$	 & 2 & 3 & 4 & 5 & 6 & 7 & 8 & 9 & 10  \\
\cmidrule(r){1-10}
	$5$	&	$2.57$	&	$2.75$	&	$2.81$	&	$2.85$	&	$2.88$	&	$2.89$	&	$2.91$	&	$2.92$	&	$2.93$	\\[0.5ex]
	$10$	&	$3.12$	&	$3.05$	&	$3.00$	&	$2.96$	&	$2.94$	&	$2.92$	&	$2.90$	&	$2.89$	&	$2.88$	\\[0.5ex]
	$20$	&	$3.37$	&	$3.18$	&	$3.07$	&	$3.00$	&	$2.95$	&	$2.91$	&	$2.88$	&	$2.85$	&	$2.83$	\\[0.5ex]
	$30$	&	$3.45$	&	$3.22$	&	$3.09$	&	$3.00$	&	$2.94$	&	$2.89$	&	$2.85$	&	$2.83$	&	$2.80$	\\[0.5ex]
\bottomrule
\end{tabular}
\caption{\footnotesize The 95\% quantiles of $\mathbb{Q}_0$ for different $m$ and $k$. The quantiles are estimated by $10^8$ random replications.} 
\label{tab:qQ0}
\end{table}

\begin{example}\label{eg:power_of_test_of_EFMI}
This example investigates the power achieved by the test statistic $\widehat{Q}$ 
at different values of $m$ when $n\rightarrow\infty$. 
Suppose that 
$r_1 = \cdots = r_{k-1} = 1-C$ and $r_k = 1+(k-1)C$,
where $C \in [0,1]$ and $k \in \{5, 10\}$. 
Then, $\mu_r = 1$, $\sigma^2_r = (k-1)C^2$ and $c_r = C\sqrt{k-1}$. 
Note that $H_0'$ is true if and only if $C = 0$. 
We compute the power of the test at 5\% level.
The results are shown in Figure \ref{fig:SMI_C002}. 
Under $H_0'$, the test has an exact size for all $m$. It verifies Proposition \ref{prop:wHat}. 
Although the test is increasingly more powerful when $m$ or $c_r$ increases, 
it does not perform well for small $m$. 
Since the values of $r_1, \ldots, r_k$ are reflected by the variability across $m$ imputed datasets, 
the power achieved by the test statistic $\widehat{Q}$ increases with $m$.
For example, in the case $c_r \approx 1$ and $k=5$, 
a decent power ($>60\%$) can be obtained when $m\geq20$, 
but the power is quite low ($<30\%$) when $m\leq 10$.  

\begin{figure}
\begin{center}
\includegraphics[width=.8\textwidth]{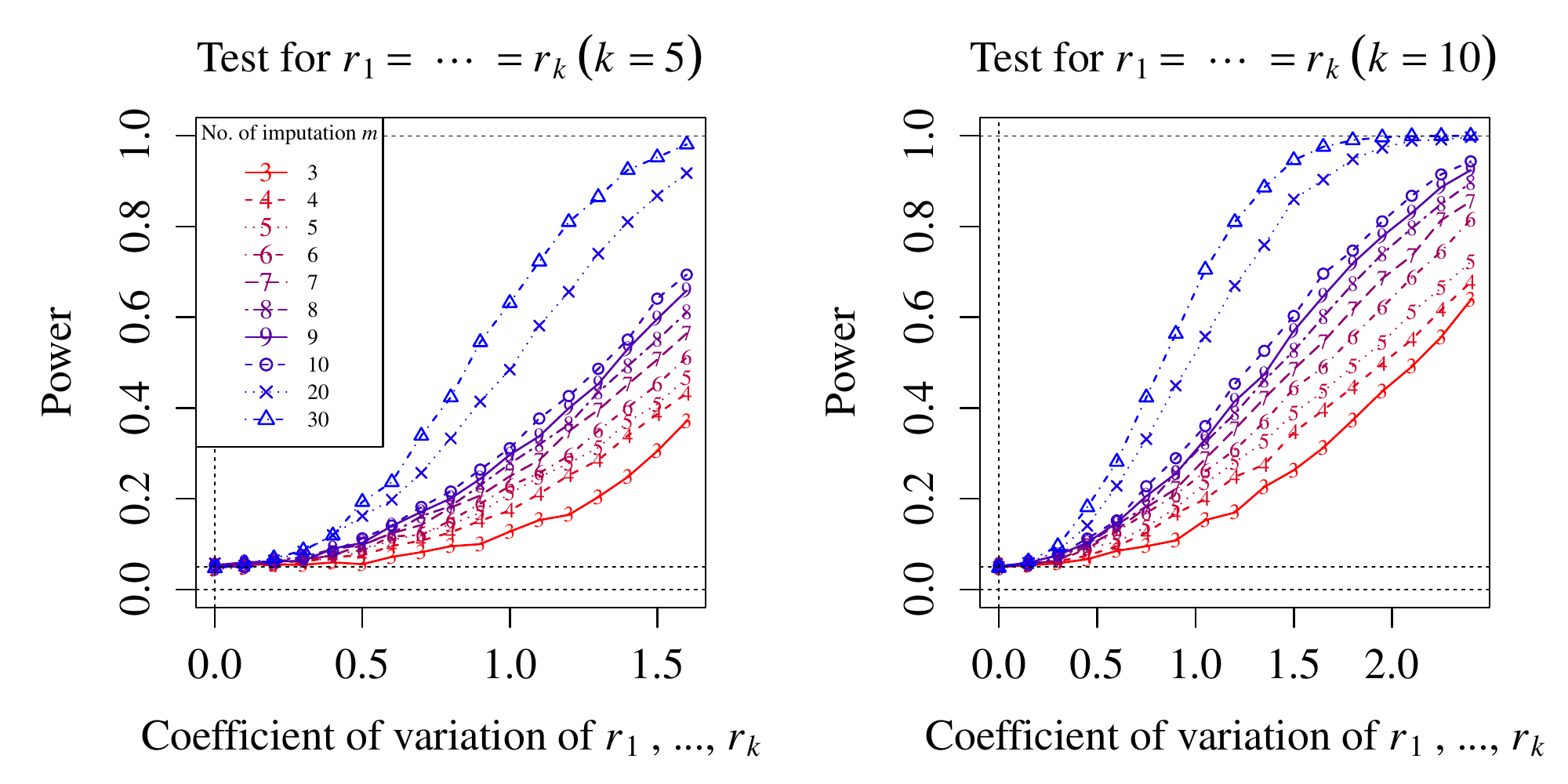}
\end{center}
\vspace{-0.7cm}
\caption{\footnotesize The power curves of the test (\ref{eqt:def_w}) 
against the CV of $r_{1:k}$ when $k\in\{5,10\}$. 
The nominal size is $\alpha_0=5\%$.} 
\label{fig:SMI_C002}
\end{figure}
\end{example}

\begin{remark}\label{rem:EOMItest}
The null $H_0'$, that is, EOMI, is almost nearly violated in practice, 
however, it is hard to be rejected when $m$ is small.
The message is that, statistically, it is not wrong to assume 
EOMI when there are only a small $m$ number of imputed datasets.
The test that assume EOMI still gives reasonable conclusion, but 
the price is that it does not control the type-I error very well as demonstrated in Section 
\ref{sec:theoretical_test_size}.
So, this section serves as a theoretical guard to the existing results
that assume EOMI
despite the fact that they could be improved.  
See Section \ref{sec:eg_contingency_table} for a numerical example. 
\end{remark}

\subsection{Other approximations of the reference null distribution $\mathbb{D}$}\label{sec:approxMethod}
In Section \ref{sec:testing},
an asymptotically exact MI test for  
$H_0:\theta =\theta_0$  
is presented. 
The proposed test statistic $\widehat{D}$ is referred to (\ref{eqt:DmHat_est}), 
which is an asymptotically \emph{exact} approximation of $\mathbb{D}$ .
Alternatively, we may approximate $\mathbb{D}$
by projecting it into a distribution depending only on $\mu_r$ and $\sigma^2_r$.
This approach is simpler, but the resulting approximations may be \emph{inexact} even asymptotically. 
In this section, we discuss two classes of such approximations:
(i) deterministic moments matching, and (ii) stochastic moments matching.  
There are two goals of this discussion. 
The first one is to provide some algorithmically simpler approximations 
to users who just want quick pilot results. 
The second one is to argue that straightforward modifications of existing results 
are not sufficient to produce good results.

In the literature, 
$\mathbb{D}$ is usually approximated by matching moments \emph{deterministically}. 
It means that $\mathbb{D}$ is approximated by a certain traceable 
parametric distribution, say $\mathbb{D}_{\text{proj}}$,
such that the parameters are determined by matching moments of $\mathbb{D}$ and $\mathbb{D}_{\text{proj}}$.
This class of methods include, for example, 
\cite{mengPhDthesis}, \cite{li91JASA} and \cite{chanMeng2017_MILRT}. 
However, 
all existing approximations either require Condition \ref{ass:EOMI} or $m\rightarrow\infty$.
Consequently, they may perform poorly when these assumptions fail.
We provide a trivial modification of the approximations provided in \cite{chanMeng2017_MILRT}
without resorting to Condition \ref{ass:EOMI} or $m\rightarrow\infty$. 
Following \cite{chanMeng2017_MILRT}, 
we approximate the numerator and the denominator of $\mathbb{D}$ in (\ref{eqt:representation_D_limit}) 
by two independent scaled $\chi^2$ random variables, 
$\alpha_1 \chi^2_{\beta_1}/\beta_1$ and $\alpha_2 \chi^2_{\beta_2}/\beta_2$, respectively. 
By matching the first two moments,
we have 
\[
	\alpha_1 = \alpha_2 = 1+\left(1+{1}/{m}\right)\mu_{r}, 
	\qquad
	\beta_1 = \frac{k}{ 1+\zeta_1^2 }, 
	\qquad
	\beta_2 = \frac{k(m-1)}{ \zeta_2^2 + \zeta_1^2} ,
\]
where 
\begin{equation*}\label{eqt:r_v_zeta_f}
	\zeta_1 = \frac{\left(1+\frac{1}{m}\right)\sigma_{r}}{1+\left(1+\frac{1}{m}\right)\mu_{r} } 
	\qquad\text{and}\qquad
	\zeta_2 = \frac{\left(1+\frac{1}{m}\right)\mu_{r}}{1+\left(1+\frac{1}{m}\right)\mu_{r}}.
\end{equation*}
Consequently, 
we obtain $F(\beta_1,\beta_2)$ as an approximation for $\mathbb{D}$. 
Since $\beta_1$ and $\beta_2$ depend on the unknowns $\mu_r$ and $\sigma^2_r$, 
one may replace $\mu_r$ and $\sigma^2_r$ by $\widehat{\mu}_r$ and $\widehat{\sigma}^2_r$ in practice;
see Section \ref{sec:estimator_of_moments_r} for more details.  
Unfortunately, this approximation can be terribly bad because the first degree of freedom ($\beta_1$)
can be severely affected by $\zeta_1$. 
If $\zeta_1$ is not accurately estimated, the resulting $F$-approximation can be seriously off from $\mathbb{D}$. 
Apart from $F$-approximation, there are other parametric distributions for approximating $\mathbb{D}$.
One example is given here. 
From (\ref{eqt:representation_D_limit}), 
we know that $\mathbb{D}$ admits the form $K_1/ (1+K_2)$, 
where $K_1$ and $K_2$ are two weighted sums of $\chi^2$ random variables. 
One may approximate $K_1$ and $K_2$ individually
by some existing methods, for example, \cite{BodenhamAdams2016}. 
However, this ad-hoc method requires two separate calibrations, 
which may result in an unsatisfactory overall approximation of $K_1/(1+K_2)$.   
Indeed, this type of approximations are only slightly better than the $F$-approximation mentioned above.

We propose another approximation that is tailor-made for $\mathbb{D}$.
Note that 
summarizing $r_1, \ldots, r_k$ deterministically through $\mu_r$ and $\sigma^2_r$ 
ignores the fact that the same set of $r_1, \ldots, r_k$
is used in both denominator and numerator of $\mathbb{D}$ in (\ref{eqt:representation_D_limit}). 
It motivates us to generate a set of random variables, 
say $\xi_1, \ldots, \xi_k$, to replace $r_{1}, \ldots, r_{k}$
in both denominator and numerator of $\mathbb{D}$ simultaneously. 
This method fully respects the functional dependence of $r_1, \ldots, r_k$ in $\mathbb{D}$, 
hence, it may lead to a better overall performance.
The random variables $\{\xi_j\}$ are chosen to match the mean and variance of $\{r_{j}\}$. 
So, we call it a \emph{stochastic} moment matching approach. 
In particular, we propose approximating $\mathbb{D}$ by
\begin{eqnarray}\label{eqt:Dcheck}
	\widecheck{\mathbb{D}} = \frac{\frac{1}{k}\sum_{j=1}^k \{1+(1+\frac{1}{m})\xi_{j}\} G_j}{ 1+\frac{1}{k}\sum_{j=1}^k (1+\frac{1}{m})\xi_{j} H_j}
\end{eqnarray}
and
$\xi_1, \ldots \xi_k \sim \text{Gamma}(\alpha, \beta)$, 
$G_1,\ldots, G_k \sim \chi^2_1$ and $H_1, \ldots, H_k \sim \chi^2_{m-1}/(m-1)$ are independent. 
The parameters $\alpha,\beta$ are determined 
such that the mean and variance of $\text{Gamma}(\alpha, \beta)$ are
${\alpha}/{\beta} = \mu_r$ and 
${\alpha}/{\beta^2} = \sigma^2_r$. 
Equivalently, the parameters $\alpha$ and $\beta$ are given by
$\alpha =  {{\mu}_r^2}/{{\sigma}^2_r}$
and
$\beta =  {{\mu}_r}/{{\sigma}^2_r}$,
which can be consistently estimated by 
$\widehat{\alpha} = {\widehat{\mu}_r^2}/{\widehat{\sigma}^2_r}$
and 
$\widehat{\beta} = {\widehat{\mu}_r}/{\widehat{\sigma}^2_r}$,
respectively.
In practice, we need to generate 
$\xi_j^{(\iota)} \sim \text{Gamma}(\widehat{\alpha}, \widehat{\beta})$
independently
for $\iota=1, \ldots, N$ and $j=1, \ldots,k$.
Then, upon conditioning on $\widehat{\mu}_r$ and $\widehat{\sigma}^2_r$, 
we can generate $N$ random replicates of $\widecheck{\mathbb{D}}$ as follows:
\begin{eqnarray}\label{eqt:Dcheck_sim}
	\widecheck{\mathbb{D}}^{(\iota)} = \frac{\frac{1}{k}\sum_{j=1}^k \left\{1+(1+\frac{1}{m})\xi_{j}^{(\iota)}\right\} G_j^{(\iota)}}{ 1+\frac{1}{k}\sum_{j=1}^k (1+\frac{1}{m})\xi_{j}^{(\iota)} H_j^{(\iota)}}, 
	\qquad \iota=1, \ldots, N.
\end{eqnarray}
Similar to Algorithm \ref{algo:SMI_exact}, critical value and $p$-value can be computed accordingly. 
Step-by-step procedure for computing 
$\widehat{p}$ is presented in Algorithm \ref{algo:SMI_approx}.
However, this approximation may not work well when $r_{1:k}$ are not approximately ``distributed'' 
as $\Ga(\alpha, \beta)$. 
Hence, it should be used with caution.

\begin{algorithm}[t]
\caption{Approximately correct MI test for $H_0$}\label{algo:SMI_approx}
\SetAlgoVlined
\DontPrintSemicolon
\SetNlSty{texttt}{[}{]}
\small
\textbf{Input}: {\;
(i) $X\mapsto\mathcal{d}(X)$ -- a function returning the complete-data test statistic; \;
(ii) $X^1, \ldots, X^m$ -- $m$ properly imputed datasets; and\;
(iii) $k$ -- dimension of $\Theta$.}\;
\Begin{
Obtain $\widehat{t}_1, \widehat{t}_2$, $\widehat{d}^{\{1:m\}}$ and $\widehat{d}^{\{1\}}, \ldots, \widehat{d}^{\{m\}}$ computed in Algorithm \ref{algo:SMI_exact}. \;
Compute $\widehat{D}$ according to (\ref{eqt:asy_version_D_proposal}).\;
Compute $\widehat{\mu}_r \gets \widehat{t}_1/k$ and 
	$\widehat{\sigma}_{r}^2 \gets [\left\{k(m-1)+2\right\} \widehat{t}_2
			-(m-1)(k+2)\widehat{t}_1^2]/\{2k^2(m-2)\}$.\;
Compute $\widehat{\alpha}\gets {\widehat{\mu}_r^2}/{\widehat{\sigma}^2_r}$ and 
		$\widehat{\beta}\gets {\widehat{\mu}_r}/{\widehat{\sigma}^2_r}$.\;
Draw $\xi_j^{(\iota)} \sim \text{Gamma}(\widehat{\alpha}, \widehat{\beta})$, 
	$G_j^{(\iota)}\sim \chi^2_1$ and $H_j^{(\iota)}\sim \chi^2_{m-1}/(m-1)$ independently 
	for $\iota=1, \ldots, N$ and $j=1, \ldots, k$. Set $N=10^4$ by default.\;
Compute $\widecheck{\mathbb{D}}^{(\iota)}$, $\iota=1, \ldots, N$,  according to (\ref{eqt:Dcheck_sim}).\;
Compute $\widehat{p} \gets \sum_{\iota=1}^N \mathbb{1}(\widecheck{\mathbb{D}}^{(\iota)} > \widehat{D})/N$.
}	
\textbf{return} $\widehat{p}$ -- the $p$-value for testing $H_0$ against $H_1$.
\end{algorithm}

In a nutshell, the most promising MI test is our major proposal T4. 
It is theoretically grounded, and does not require any strong assumption.

\section{Simulation details and extra examples}\label{sec:sim_details}

\subsection{Follow-up study of Example \ref{eg:rHat}}\label{sec:theoretical_test_size}
This subsection investigates the size accuracy of the MI test statistic $\widehat{D}$
with the reference null distributions stated in Table \ref{tab:approxNullDist}, that is,  
T1 ($\chi^2_k/k$), 
T2 \citep{li91JASA}, 
T3 \citep{mengPhDthesis}, and 
T4 (the proposal in Algorithm \ref{algo:SMI_exact}).
The unknowns $\mu_r$ and $\sigma^2_r$ in T2 and T3 are estimated by 
$\widehat{\mu}_r$ and $\max(\widehat{\sigma}_{r}^2,0)$, respectively; 
see (\ref{eqt:estimator_mur_sigmar}). 
In T4, 
$r_{1:k}$ are estimated by $\widehat{r}_{1:k}$.

Similar to Example \ref{eg:rHat}, 
the experiments are performed when $n\rightarrow \infty$ but $m \in\{10,20\}$ is fixed.
In particular, we set $k=4$ with
$r_1 = r_2 = 0$ and $r_3 = r_4 = r_{\max}$, 
where $r_{\max} \in \{0.3, 0.6, 0.9\}$. 
The MI tests are performed
under $H_0$ (i.e., $\delta_1 = \cdots = \delta_k=0$ in (\ref{eqt:representation_of_d}))
when different 
nominal sizes (i.e., type-I error rates) $\alpha_0\in[1\%, 5\%]$ are used. 
Figure \ref{fig:SMI_K22} shows the results.

Among T1, T2 and T4, 
our proposed T4 has the highest size accuracy in all cases. 
It means that T4 controls the type-I error rate very well at the nominal value $\alpha_0$.
The performance of T3 is trickier. 
When $m=10$, the size of T3 can be accurate in some cases (e.g., $r_{\max}=0.9$), 
but inaccurate in some cases (e.g., $r_{\max}=0.3$). 
When $m$ increases, its size accuracy deteriorates.
This counterintuitive phenomenon is due to its reliance on $m\rightarrow\infty$. 
The performance of T3 is not trustworthy unless $m$ is large enough. 
In a nutshell, 
the proposed MI test T4 outperforms all other MI tests when $m=20$.

\begin{figure}
\begin{center}
\includegraphics[width=\linewidth]{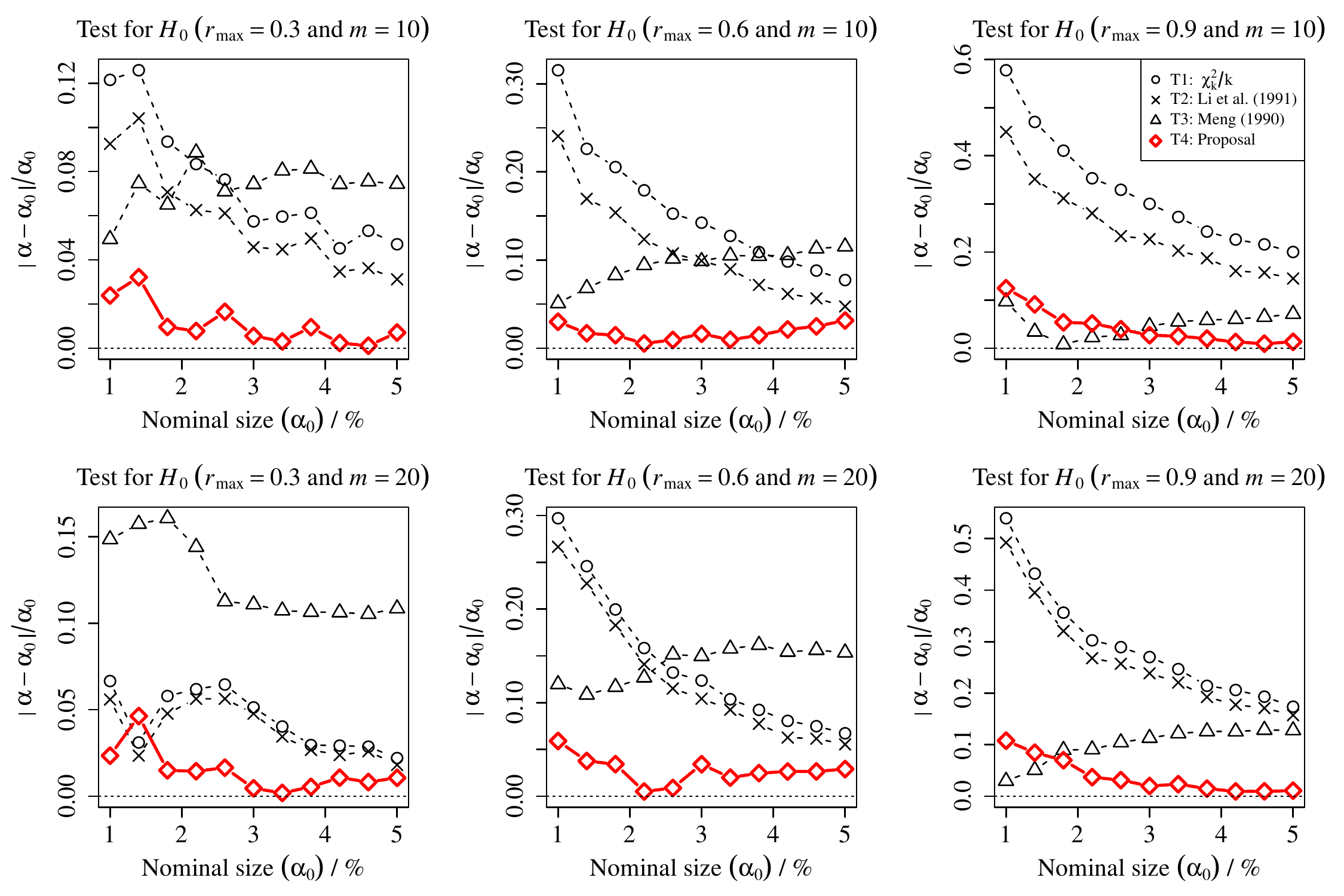}
\end{center}
\vspace{-0.5cm}
\caption{\footnotesize The relative size error $|\alpha-\alpha_0|/\alpha_0$ of MI tests, where $\alpha_0$ is the nominal size 
and $\alpha$ is the actual size. See Section \ref{sec:theoretical_test_size}.} 
\label{fig:SMI_K22}
\end{figure}

\subsection{Imputation model and additional results in Section \ref{sec:eg_reg}}\label{sec:imputationModel_regression}
Since the missing pattern is monotone, 
we can re-arrange the dataset such that $I_{ij}\geq I_{i'j}$ for all $i<i'$ and $j$ 
without loss of generality.
Let 
$\tilde{X}_j = (\tilde{x}_{1j}, \ldots, \tilde{x}_{n_jj})^{\T}$  and $n_j = \sum_{i=1}^n I_{ij}$, 
where
$\tilde{x}_{ij} =(1, x_{i1}, \ldots, x_{i,j-1})^{\T}$ and 
$w_j = (x_{1j}, \ldots, x_{n_jj})^{\T}$.
We assume the following Bayesian imputation model: 
\begin{eqnarray*}
	\left(x_{i1} \mid \phi_1 , v_1\right) &\sim& \Normal(\phi_1, v_1) ,\\
	\left(x_{ij} \mid x_{i1},\ldots, x_{i,j-1}, \phi_j, v_j\right) &\sim& \Normal(\phi_j^{\T} \tilde{x}_{ij}, v_j),
	\qquad j=2,\ldots,p ,\\
	f(\phi_1, \ldots, \phi_p, v_1, \ldots, v_p)
	&\propto& 1/(v_1\cdots v_p).
\end{eqnarray*} 
A Gibbs sampler is used to draw posterior samples from 
$[ v_1, \ldots, v_p, \phi_1, \ldots, \phi_p , X_{\mis}\mid X_{\obs}]$,
where $X_{\obs} = \{ x_{ij} : I_{ij}=1\}$ and $X_{\mis} = \{ x_{ij} : I_{ij}=0\}$.
Discarding the first $N_{\text{burn}}$ of the posterior samples as burn-in, 
we select $m$ subsequent generated samples of $(v_1, \ldots, v_p, \phi_1, \ldots, \phi_p)$
each separated by $N_{\text{thin}}$ iterations. 
Then the imputed missing data are generated upon conditioning on the selected parameter samples.

Let
$\mathbb{Y} := (y_1-\bar{y}, \ldots, y_n-\bar{y})^{\intercal} \in\mathbb{R}^{n\times 1}$, and
$\mathbb{X} := (x_1-\bar{x}, \ldots, x_n-\bar{x})^{\intercal}\in \mathbb{R}^{n\times p}$
be two matrices, 
where 
$\bar{y} := \sum_{i=1}^n y_i/n$ and 
$\bar{x} := \sum_{i=1}^n x_i/n$.
Also denote the (unavailable) complete dataset by $Z_{\com} := \{ y_i, x_i : 1\leq i\leq n\}$. 
The devices for computing the complete-data Wald's, LR, and RS test statistics are 
\begin{gather*}
	\mathcal{d}_{\wt}(Z_{\com}) 
		= \frac{1}{\widehat{\sigma}^2} 
			\mathbb{Y}^{\intercal}\mathbb{X}(\mathbb{X}^{\intercal}\mathbb{X})^{-1}
			\mathbb{X}^{\intercal}\mathbb{Y}, \qquad
	\mathcal{d}_{\lrt}(Z_{\com}) 
		= 2n\log\frac{\widehat{\sigma}_0}{\widehat{\sigma}}, \\
	\mathcal{d}_{\rt}(Z_{\com}) 
		= \frac{1}{\widehat{\sigma}_0^2} 
			\mathbb{Y}^{\intercal}\mathbb{X}(\mathbb{X}^{\intercal}\mathbb{X})^{-1}
			\mathbb{X}^{\intercal}\mathbb{Y},
\end{gather*}
respectively, 
where 
$\widehat{\sigma}^2_0 := \sum_{i=1}^n(y_i - \widehat{\beta}_0)^2/n$, 
$\widehat{\sigma}^2 := \sum_{i=1}^n(y_i - \widehat{\beta}_0- \widehat{\beta}x_i)^2/n$, 
$\widehat{\beta}_0 := \bar{y}$, and 
$\widehat{\beta} := (\mathbb{X}^{\intercal}\mathbb{X})^{-1}\mathbb{X}^{\intercal}\mathbb{Y}$.

When $m=30$, the power performance 
is shown in Figure \ref{fig:SMI_I34_H1_full}.
When $m=10$, the results of the size accuracy and power performance are shown in 
Figure \ref{fig:SMI_I33_H0_full} and \ref{fig:SMI_I33_H1_full}, respectively.

We also investigate the performance of the tests when $m$ gradually increases from $2$ 
to $128$; see Remark \ref{rem:small_m} for a discussion about the role of $m$.
From Figure \ref{fig:SMI_U001} (a), 
the maximum relative size error of T4 approaches zero when $m$ increases.
Indeed, T4 is uniformly better than T1--T3 when $m>4$. 
Its performance stabilizes when $m\geq 16$. 
However, when $m=2$, that is, the minimum value for MI to be well-defined, 
T4 is slightly worse than T2.
Since T4 requires weaker assumptions, 
a larger imputation number ($m$) is needed in order to estimate more a complicated 
null distribution $\mathbb{D}$.
Fortunately, it is very unusual to use $m=2$ in real practice.

\begin{remark}\label{rem:small_m}
The value of $m$ plays a similar role as in the classical MI inference. 
According to Theorem 2 of \citet{wangRobins1998}, 
the MI estimators are in general inefficient unless $m\rightarrow\infty$
because MI is essentially \emph{``a size $m$ Monte Carlo simulation from the posterior''}; 
see Section 1.2 of \citet{XXMeng2017}.
More discussions can be found in \citet{rubin1996multiple}. 
Hence, for all MI tests, including our proposal and the classical Rubin's MI test, 
the power and size-accuracy increases as $m$ increases. 
Moreover, the nuisance parameters concerning the missing mechanism,  
for example, $B$, $r_1, \ldots, r_m$, 
can be more precisely estimated when $m$ increases. 
\end{remark}

\begin{figure}
\centering
\includegraphics[width=\linewidth]{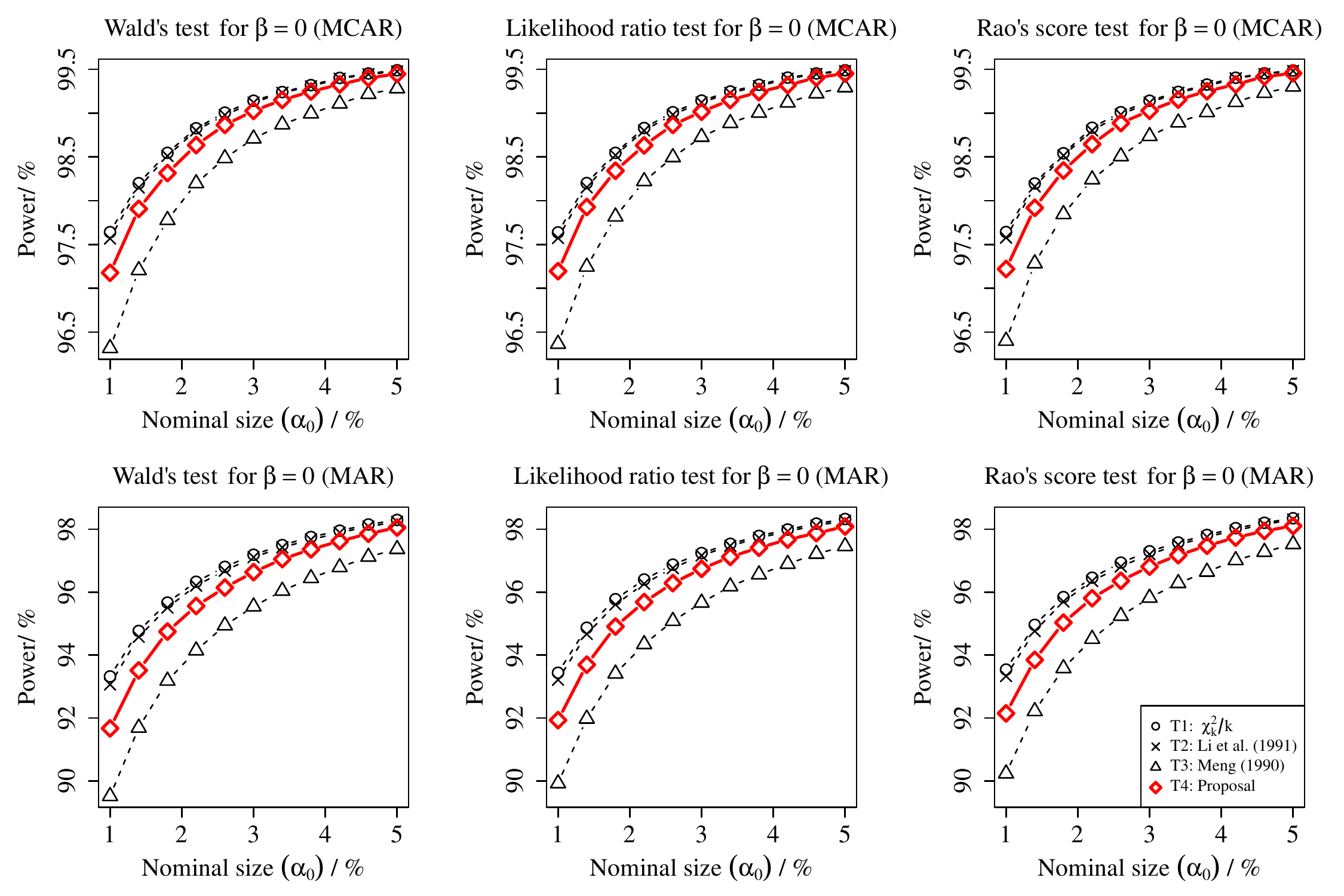}
\vspace{-0.5cm}
\caption{Powers of the tests in Section \ref{sec:eg_reg} when $m=30$.}
\label{fig:SMI_I34_H1_full}
\end{figure}

\begin{figure}
\centering
\includegraphics[width=\linewidth]{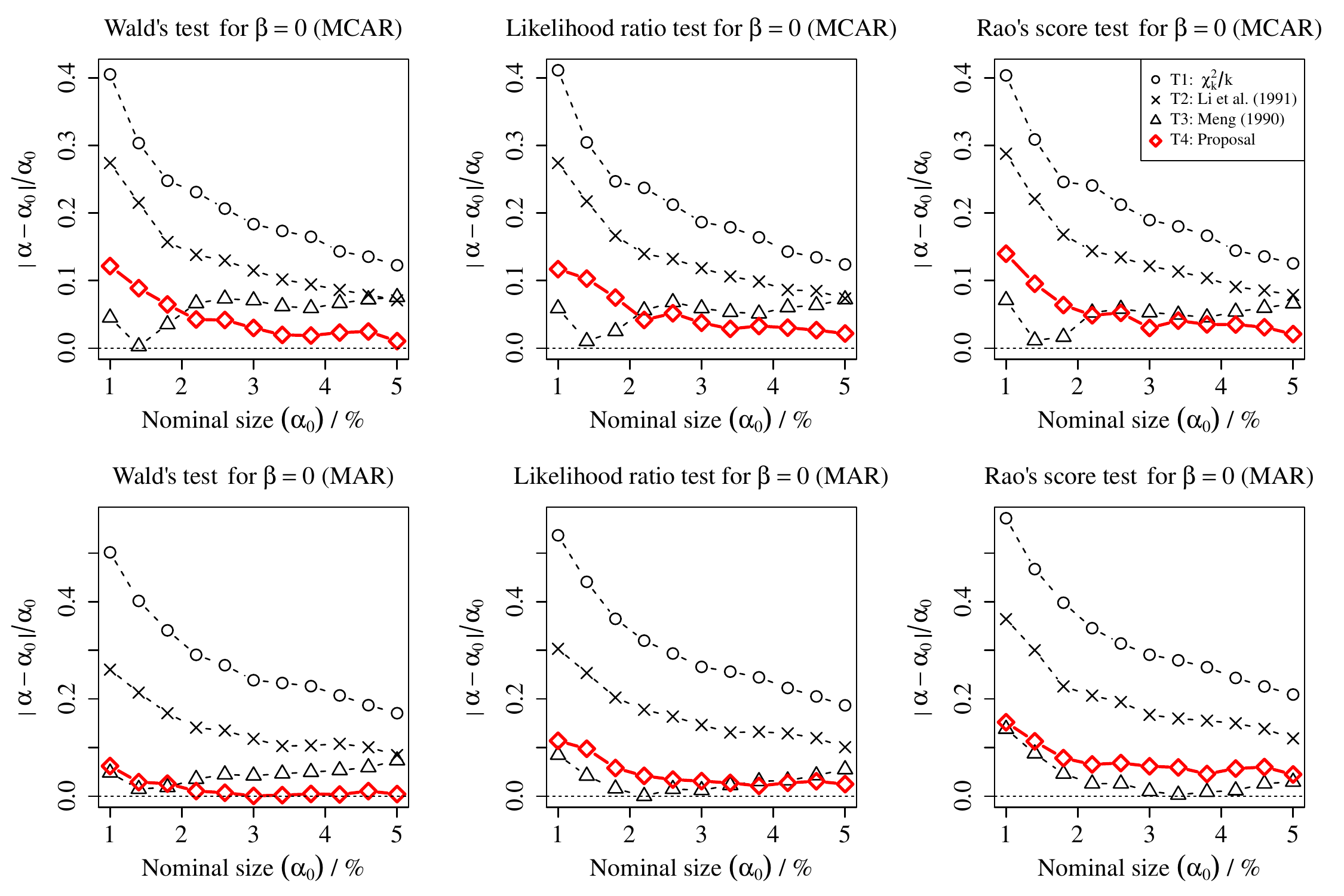}
\vspace{-0.5cm}
\caption{Size accuracy of the tests in Section \ref{sec:eg_reg} when $m=10$.}
\label{fig:SMI_I33_H0_full}
\end{figure}

\begin{figure}
\centering
\includegraphics[width=\linewidth]{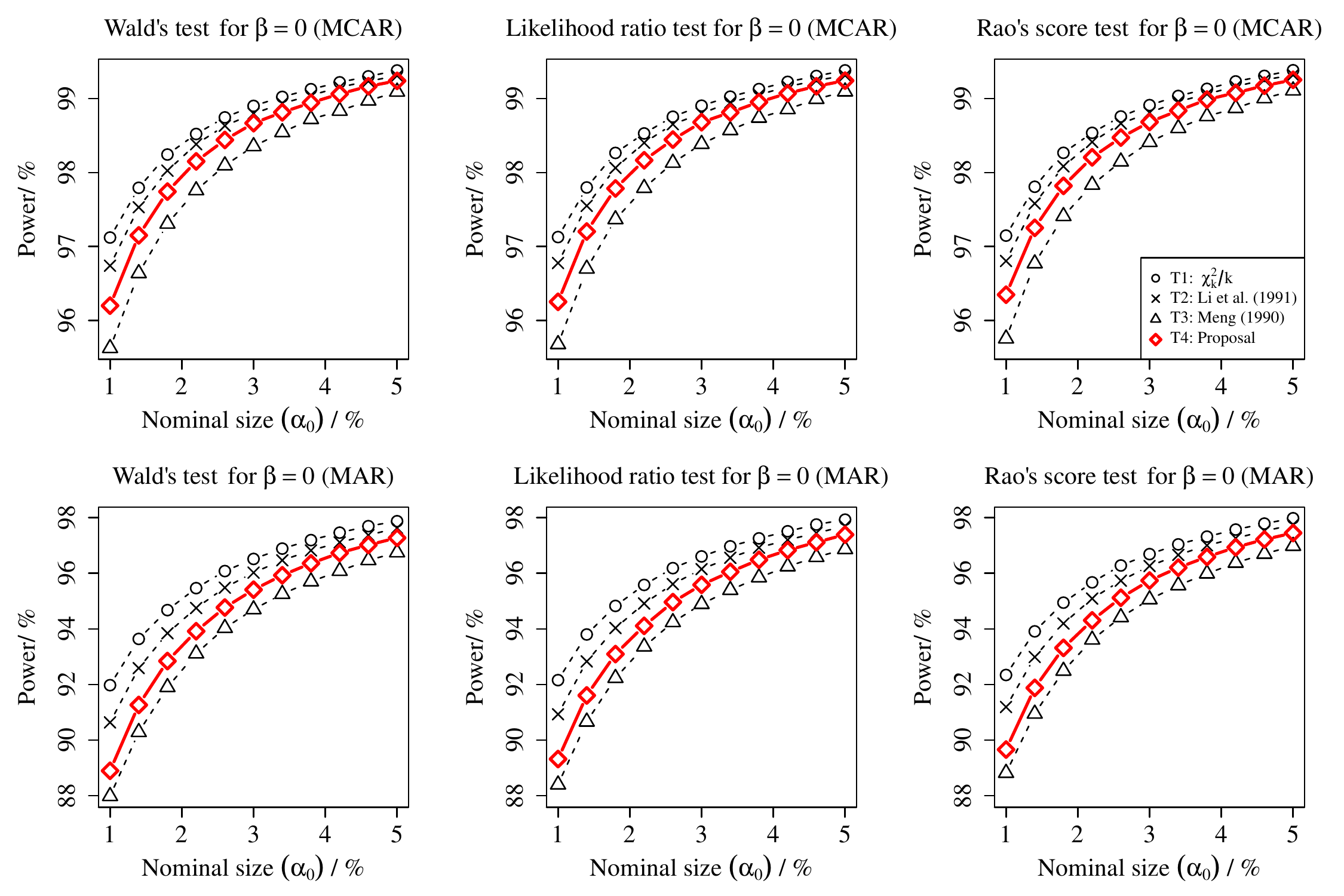}
\vspace{-0.5cm}
\caption{Powers of the tests in Section \ref{sec:eg_reg} when $m=10$.}
\label{fig:SMI_I33_H1_full}
\end{figure}

\begin{figure}
\centering
\includegraphics[width=1\linewidth]{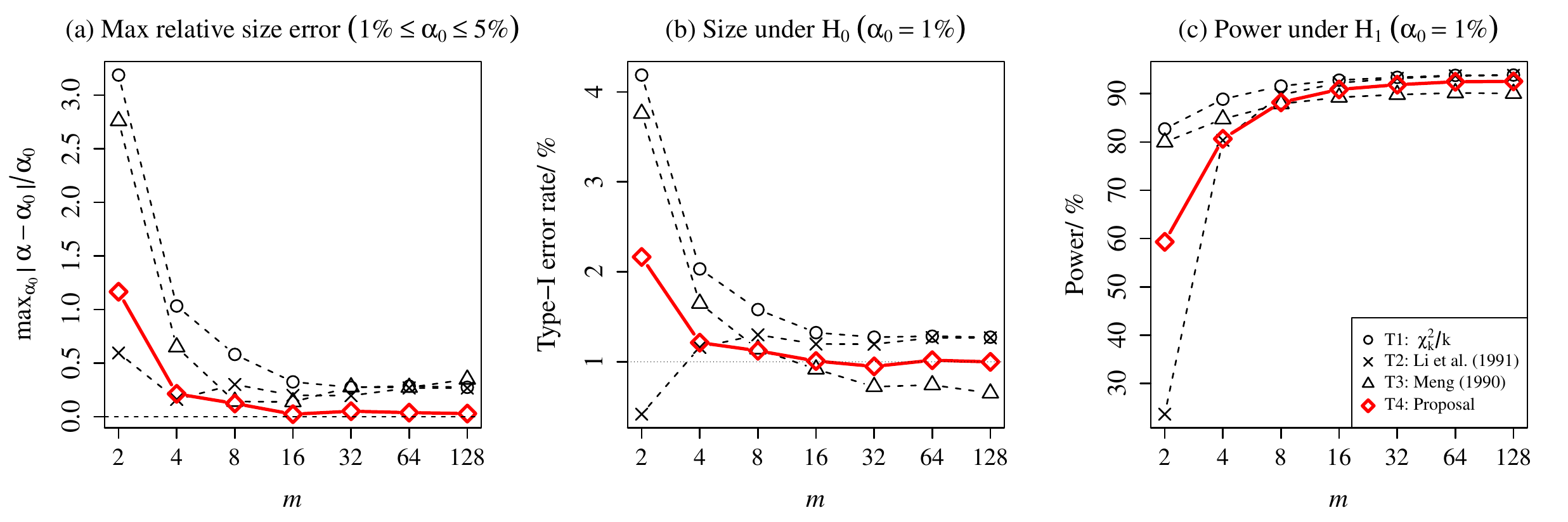}
\vspace{-0.5cm}
\caption{\footnotesize The size and power performance of the MI tests in Section \ref{sec:eg_reg}. 
The following performance measures are plotted against $m$:
(a) the maximum relative size error over $1\%\leq \alpha_0 \leq 5\%$; 
(b) the actual type-I error rate when the nominal size is $\alpha_0=1\%$; and 
(c) the power of the test under $H_1$ when $\alpha_0=1\%$.}
\label{fig:SMI_U001}
\end{figure}

\subsection{Inference of variance-covariance matrix}\label{sec:eg_mvn}
Let $x_1, \ldots, x_n \simIID \Normal_p(\mu, \Sigma)$, 
where $\mu$ and $\Sigma$ are unknown. 
We want to test $H_0:\Sigma = \Sigma_0$ against $H_1:\Sigma \neq \Sigma_0$, 
where $\Sigma_0$ is a fixed variance-covariance matrix. 
Let the complete dataset be $X_{\com}=(X_1^{\T},\ldots, X_{n}^{\T})^{\T}$, 
and $\widehat{\Sigma}:=\sum_{i=1}^n (x_i-\bar{x})(x_i-\bar{x})^{\T}/n$.
Then the device for computing the LR test statistic is 
\[
	\mathcal{d}_{\lrt}(X_{\com}) 
		= n \tr\left( \Sigma_0^{-1} \widehat{\Sigma} \right) - n \log \det\left( \Sigma_0^{-1} \widehat{\Sigma} \right) - np,
\]
where 
$\tr(\cdot)$ and $\det(\cdot)$ denote the trace and determinant, respectively.

There are two purposes of presenting this example. 
First, the parameter of interest is $\theta = \vech(\Sigma)$, 
where $\vech(M)$ denotes the half-vectorization of a matrix $M$.
So, the dimension of $\theta$ is $k = p(p+1)/2$, which increases very quickly when $p$ increases. 
In practice, 
$k$ can easily be larger than the number of imputation $m$ even for a small $p$.
In this section, we investigate the performance of various MI tests when 
$k<m$, $k=m$ and $k>m$.
Second, 
apart from our proposal, 
the MI test T3 \citep{mengPhDthesis} is the only existing procedure without 
employing Condition \ref{ass:EOMI}.
Computing this MI test requires an estimator of 
$\Var\{\vech(\widehat{\Sigma})\}$, that is, $\widehat{V}$ in (\ref{eqt:com_WT}). 
Unfortunately, in this example, a good estimator $\widehat{V}$ is not easy to construct 
because $\widehat{V}$ is a huge matrix, whose dimensions are $k \times k$, 
where $k = p(p+1)/2$. 
Hence, it is desirable to have a MI test that bypasses computing $\widehat{V}$.
Our proposed MI test T4 resolves this problem.

In the experiments, 
we use $n=400$, $m\in\{10,30\}$, $\mu=1_p$, and $\Sigma_0 = I_p$, 
where $I_p$ is a $p\times p$ identity matrix. 
We consider $p\in\{3,4,5\}$ so that $k\in\{6,10,15\}$, respectively.
The data are missing according to the mechanism (\ref{eqt:missingDataMech})
with $\gamma_0 = 1.25$ and $\gamma_1 = -0.5$.
The parameters are chosen such that, on average, 50\% of the 5th components of $x_i$'s are missing. 
The missing data are imputed as in Section \ref{sec:eg_reg}.

We compare the size accuracy of the tests in Figure \ref{fig:SMI_M29-30_H0}.
When $m=10$, the proposed T4 controls the size much better than T1 and T2
no matter $k<m$, $k=m$, or $k>m$. 
Similar to the example in Section \ref{sec:theoretical_test_size}, 
the performance of T3 is hard to predict when $m$ is small. 
It may perform well in some cases (e.g., $k>m$), but perform less satisfactorily 
in some cases (e.g., $k<m$ and $k=m$).
When $m=30$, the advantage of using T4 is obvious. 
Its size accuracy is the highest in all cases. 
For reference, we also present the power performance 
in Figure \ref{fig:SMI_M29-30_H1_full}. 
In a nutshell, the proposed T4 generally performs well.

\begin{figure}
\centering
\includegraphics[width=\linewidth]{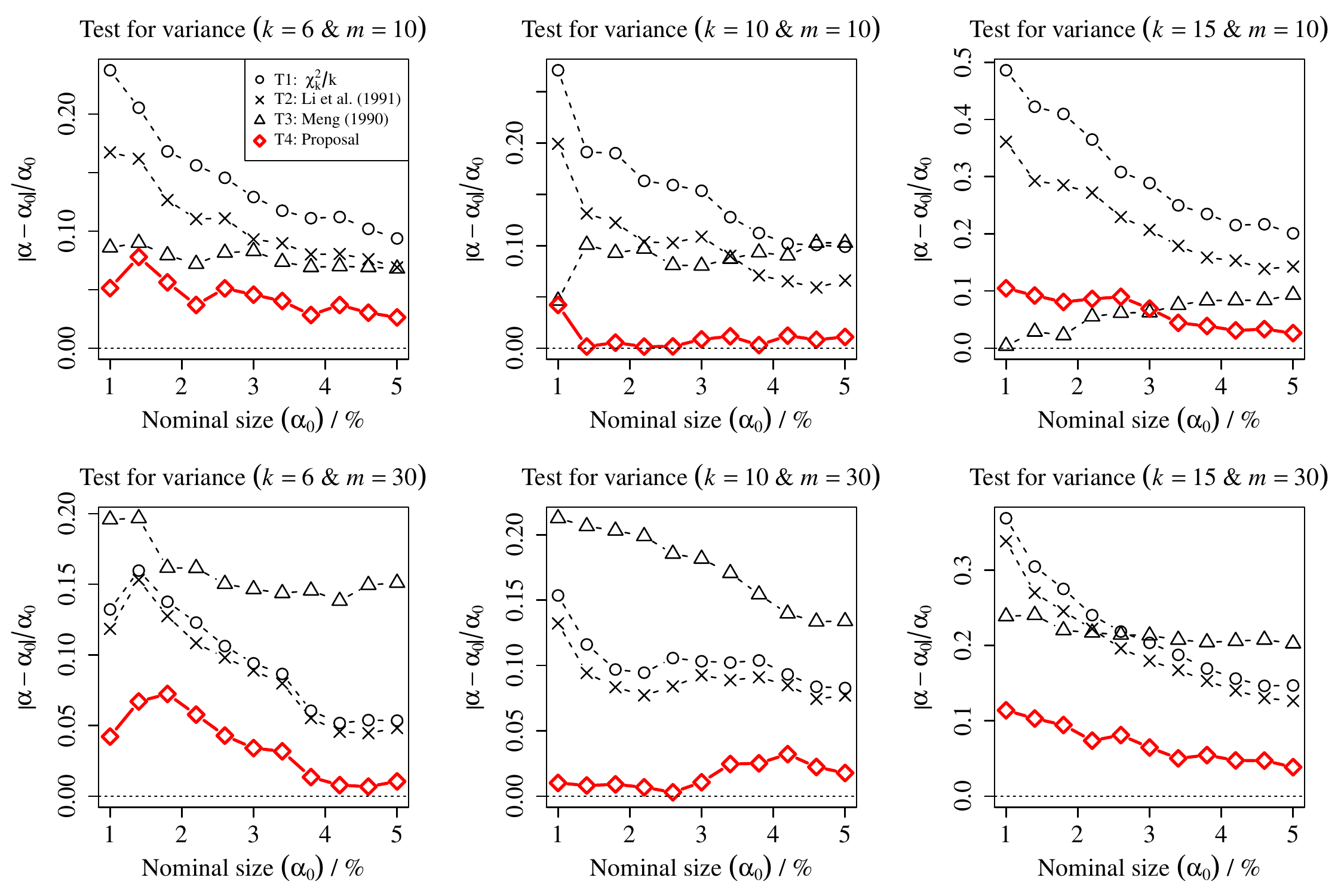}
\vspace{-0.5cm}
\caption{\footnotesize Size accuracy of the MI tests regarding variance-covariance matrix in Section \ref{sec:eg_mvn}.}
\label{fig:SMI_M29-30_H0}
\end{figure}

\begin{figure}
\centering
\includegraphics[width=\linewidth]{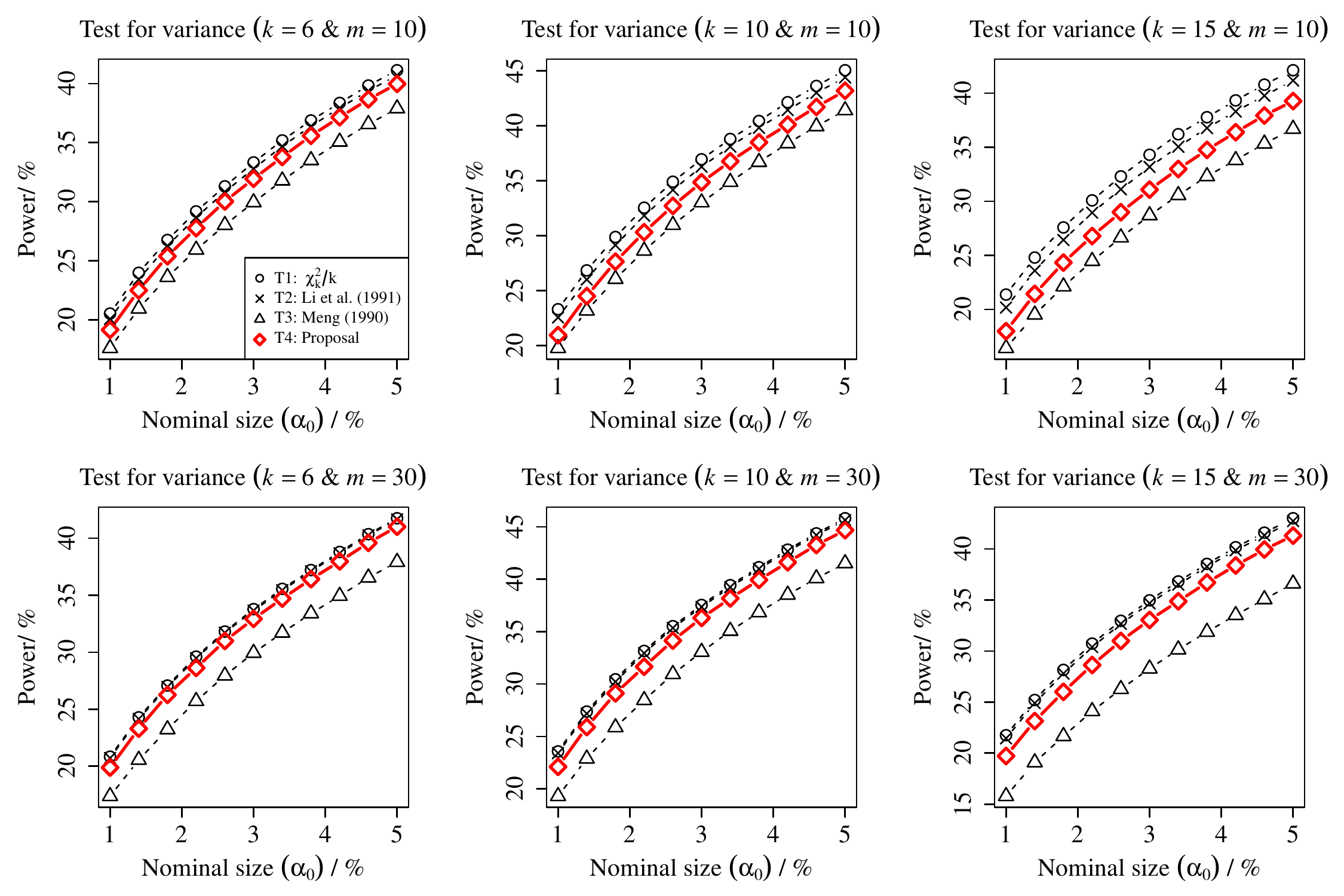}
\vspace{-0.5cm}
\caption{Powers of the tests in Section \ref{sec:eg_mvn}.}
\label{fig:SMI_M29-30_H1_full}
\end{figure}

\subsection{Generalized linear model}\label{variable_selection_glm}
As discussed in Section \ref{sec:applicationsDiscussion}, 
our proposed test can be used not only for merely null hypothesis testing, 
but also for other automatic procedures that use $p$-value as an assessment tool. 
This subsection presents how our proposed method be used in stepwise variable selection in 
a generalized linear model. 

Let $y_i$ be the count response and $x_{i,1}, \ldots, x_{i,14}$ be the covariates for the 
$i$th unit. They are generated as follow:
\begin{align*}
	\left[y_{i}\mid x_{i, 1:14}\right] 
		&\sim \Po( \exp( \beta_0+\beta_1x_{i,1}+\cdots+\beta_{14}x_{i,14}) ), \\
	x_{i1} &\sim \Bern(0.5), \\
	x_{i2} &\sim \Ga(2,1),  \\
	(x_{i3}, \ldots, x_{i,12})^{\intercal} &\sim \Normal_{10}(\mu, \Sigma), \\
	\left[x_{i,13}\mid x_{i, 1:3} \right] &\sim \Bern( \expit( \gamma_0+\gamma_1x_{i,1}+\gamma_{2}x_{i,2}+\gamma_{3}x_{i,3} )), \\
	\left[x_{i,14}\mid x_{i, 1:3} \right] &\sim \Po( \exp( \eta_0+\eta_1x_{i,1}+\eta_{2}x_{i,2}+\eta_{3}x_{i,3} ) ), 
\end{align*}
where 
$\beta_{0:14} = (5,5,10,5,10,10,0,0,0,0,10,5,0,0,0)^{\intercal}/100$, 
$\mu = 0$, $\Sigma_{ab} = 0.5^{|a-b|}$ for $1\leq a,b\leq 10$, and 
$\gamma_j = \eta_j = 1/10$ for $j=0,1,2,3$. 
So, $x_{i,1}$ and $x_{i,13}$ are binary; 
$x_{i,2}$ is positive;
$x_{i,3:12}$ are continuous, 
$x_{i,14}$ is integral.  
Note that the non-zero value of $\beta_{1:14}$ are small,
so selecting the predictive variables is, by definition, very hard. 
Suppose further that 
the following groups of covariates must exist or not exist together as a group: 
$\{x_{i,1}, x_{i,4}, x_{i,5}, x_{i,10}\}$,  
$\{x_{i,2}, x_{i,3}, x_{i,11}\}$,  
$\{x_{i,7}, x_{i,8}, x_{i,9}, x_{i,13}\}$,  
$\{x_{i,6}, x_{i,12}, x_{i,14}\}$.
It mimics the situation that we have some domain knowledge on the underlying mechanism. 
Note that the first two groups of variables should be selected as their corresponding 
regression coefficient $\beta_j$'s are not zero. 
We consider a forward stepwise procedure by using LR test 
with different nominal level of significance $10\%\leq \alpha_0\leq 15\%$. 
All four reference distributions (T1--T4) are compared.

Let $I_{i,j}$ be the non-missing indicator of $x_{i,j}$, that is, $I_{i,j}=1$ if $x_{i,j}$ is observed. 
We assume that $y_i$, $x_{i,1}$, $x_{i,2}$, $x_{i,3}$ are always observed, but the 
$x_{i, 4:14}$ may be missing according to the following mechanisms: 
\begin{align*}
	\pr\left( I_{i,j} = 1 \mid x_{i,j-1},I_{i,j-1} =a\right) 
		&= \expit\left(  1.5 + 0.5 x_{i,j-1} \right) \mathbb{1}(a=1), \; j=4, \ldots, 12, \\
	\pr\left( I_{i,13} = 1 \mid x_{i,1:2} \right) 
		&= \frac{1}{1+0.5(0.7+x_{i,1})x_{i,2}} , \\
	\pr\left( I_{i,14} = 1 \mid x_{i,1:3} \right) 
		&=	\Phi(0.18-0.23x_{i,1}-0.0013x_{i,2}-0.36x_{i,3}), 
\end{align*}
where $\Phi(\cdot)$ is the distribution function of $\Normal(0,1)$. 
Note that the parameters in the missing mechanisms are chosen so that 
$x_{i,12}$, $x_{i,13}$ and $x_{i,14}$ are missing with around 50\% probability. 

In the simulation study, the commonly used \texttt{R} package ``\texttt{mice}'' is employed 
to perform multiple imputation. 
This package uses chained equations to perform MI, thus Condition \ref{cond-properImp} is not satisfied. 
Although it is not a proper imputation model in general, 
it still performs satisfactorily in many problems.  
We investigate whether the correct groups of covariates can be selected. 

We record the sensitivity (SEN), specificity (SPE) and Matthews correlation coefficient (MCC)
for each method; see Figure \ref{fig:SMI_S4_606}.  
There are several important observations. 
First, there is a tradeoff between SEN and SPE.
For T1, its SEN is slightly better than all other methods, 
however its SPE is obviously lower.
Similarly, for T3, its SPE is slightly better, but its SPE is low. 
Consequently, their overall performance measured by MCC is less satisfactory compared with T2 and T4. 
Second, our proposed T4 has the best overall MCC performance no matter what cutoff $\alpha_0$ is used. 
Hence, in practice, T4 is more recommended.
Last but not least, this experiment demonstrates that our proposed T4 continues to perform well 
even Condition \ref{cond-properImp} is (slightly) violated. 
This robustness property give confidence for practitioners  
who may use (slightly) improper imputation methods, for example, the ``\texttt{mice}'' package.

\begin{figure}
\begin{center}
\includegraphics[width=\textwidth]{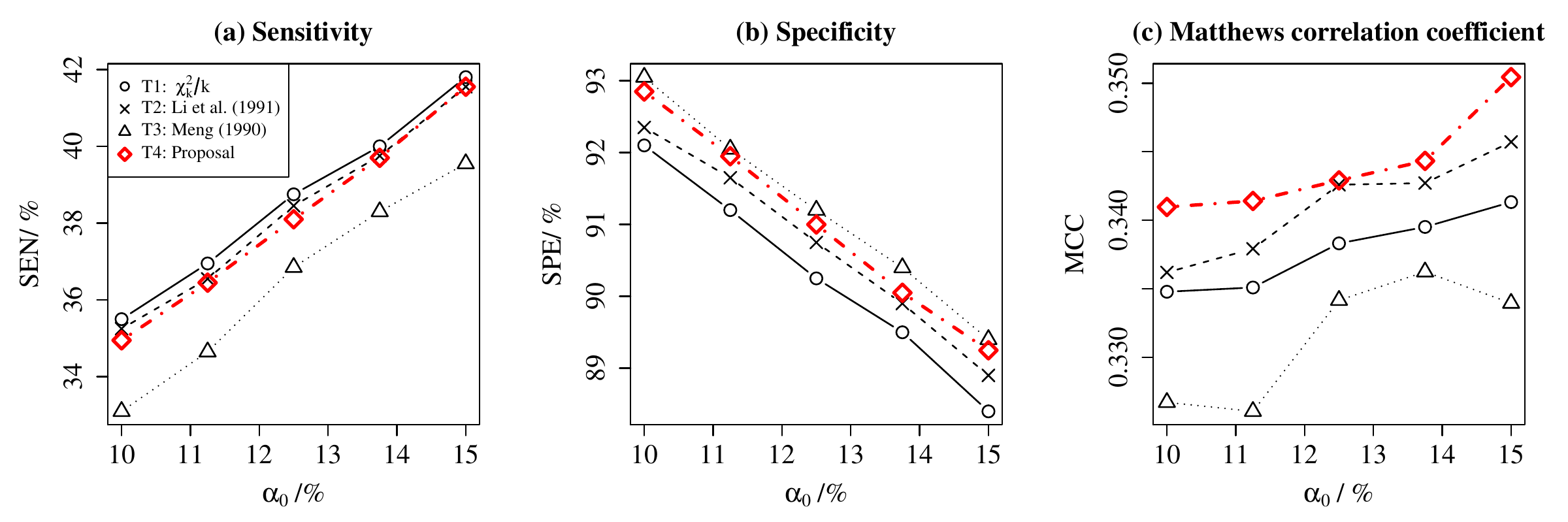}
\end{center}
\vspace{-0.5cm}
\caption{\footnotesize The variable selection performance 
of different methods in Section \ref{variable_selection_glm}.} 
\label{fig:SMI_S4_606}
\end{figure}

\subsection{Contingency table}\label{sec:eg_contingency_table}
\cite{mengRubin92} applied their MI LR test to a partially identified table. 
The dataset is given in Table \ref{table:realData_dataset}.
The dataset is arranged in a $2\times 2\times 2$ contingency table indexed by $i,j,k\in\{0,1\}$, 
where $i$, $j$ and $k$ index, respectively, clinic (A or B), amount of parental care (more or less) and survival status (died or survived). 
The clinic label $k$ is missing for some of the observations. 
The missing-data mechanism was assumed to be ignorable.   
They tested the null hypothesis $H_0$ that the clinic and parental care are conditionally independent given the survival status.
In this testing problem, $k=2$.
We estimate $r_1$ and $r_2$ by our proposed estimators $\widehat{r}_1$ and $\widehat{r}_2$, respectively, 
by using  $m\in\{5,10,20, \ldots, 2560\}$. 
Since MI is a random procedure,  
the realized values of $\widehat{r}_1$ and $\widehat{r}_2$ vary 
in different replications. 
We repeat the MI procedures $2^{10}$ times. 
The results are shown in Figure \ref{fig:SMI_N1}.

\begin{table}
\centering
\footnotesize
\begin{tabular}{cccc}
\toprule
&&\multicolumn{2}{c}{{Survival Status} ($j$)} \\
\cmidrule(r){3-4}
    {Clinic} ($k$) & {Parental care} ($i$) & Died & Survived  \\ 
   \cmidrule(r){1-4}
    A & Less & 3 & 176  \\ 
     & More & 4 & 293 \\ 
\cmidrule(r){1-4}
    B & Less & 17 & 197  \\ 
     & More & 2 & 23  \\ 
\cmidrule(r){1-4}
    ? & Less & 10 & 150  \\ 
     & More & 5 & 90  \\ 
\bottomrule
\end{tabular}
\caption{\footnotesize Data from \cite{mengRubin92}. The notation ``?'' indicates missing label.} 
\label{table:realData_dataset}
\end{table}

\begin{figure}
\begin{center}
\includegraphics[width=.9\textwidth]{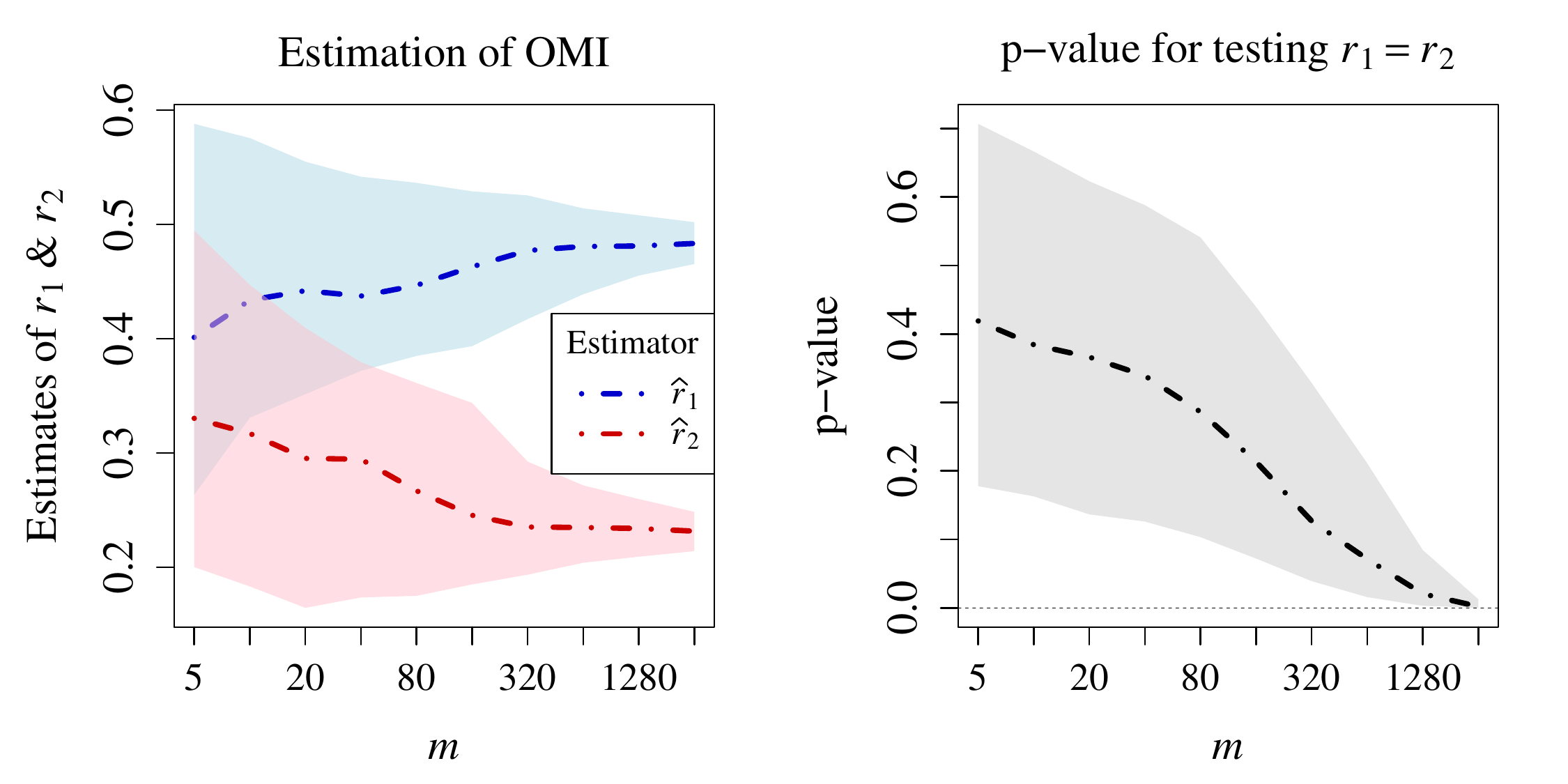}
\end{center}
\vspace{-0.5cm}
\caption{\footnotesize Left: The estimates of $r_1$ and $r_2$. 
Right: The $p$-values for testing $H_0':r_1=r_2$.  
The dotted lines denote the median of respective estimates over $2^{10}$ MI trials. 
the shaded regions designate the respective 25\% and 75\% quantiles.} 
\label{fig:SMI_N1}
\end{figure}

First, when $m$ is large enough, 
$r_1$ and $r_2$ are estimated to be 48\% and 23\%, respectively. 
When $m$ is as small as $10$, the estimates are 43\% and 32\%,
which are still reasonably close to the limiting values.
Besides, the estimates stabilize as $m$ increases. 
It confirms Corollary \ref{coro:rHat}.
Second, the $p$-value for testing  $H_0':r_1=r_2$  
is monotonically decreasing when $m$ increases.
So, the violation of $H_0'$ becomes increasingly more significant as $m$ increases. 
In practice, as we mentioned in Remark \ref{rem:EOMItest}, $H_0'$ is rarely true, 
so we can always reject $H_0'$ as long as $m$ is large enough. 
However, 
$H_0'$ cannot be rejected at 10\% level even if $m=640$ is used. 
Hence, for a reasonably small $m$ and a commonly used nominal size $\alpha_0$, 
we cannot find any significant statistical evidence to reject $H_0'$.
Although it provides a theoretical guard for assuming EOMI (i.e., Condition \ref{ass:EOMI}) in \cite{mengRubin92}, 
their MI test does not control the type-I error rate as accurately as our proposed MI test T4.
Hence, T4 is recommenced in practice for better insurance.

\subsection{Logistic regression}\label{sec:logisticEG}
In the 1994 Census income dataset \citep{Kohavi1996},  
the response $y\in\{0,1\}$ indicates whether the salary is larger than \$50,000. 
Six covariates are investigated: 
age ($x_1$), sampling weight ($x_2$), 
education level in a numeric scale ($x_3$), 
capital loss ($x_4$), number of working hour per week ($x_5$), and 
occupation class ($x_6$). 
The variables $x_1,\ldots,x_5$ are ordinal; 
whereas $x_6\in\{1,\ldots, 7\}$ is an unordered categorical variable of 7 levels. 
After removing 10 respondents who never work, 
the dataset contains $n=48832$ individuals.
There is no missing entry in $y,x_1,\ldots, x_5$, 
however, $x_6$ is missing in about 6\% of the individuals.

The missing data are imputed $m=10$ times
by a multinomial logit model:
\[
	\pr(x_6 = g) = \frac{\exp(\gamma_{0g} + \gamma_{1g}x_1 + \cdots + \gamma_{5g}x_5)}{ 1+ \sum_{g'=2}^7 \exp(\gamma_{0g'} + \gamma_{1g'}x_1 + \cdots + \gamma_{5g'}x_5)}	
\]
for $g = 2, \ldots, 7$, and  
$\pr(x_6 = 1) = 1-\sum_{g=2}^7 \pr(x_6 = g)$.
A flat prior on the $\gamma$'s is used.  
We study the relationship between $y$ and the covariates. 
The following logistic regression model is considered:
\begin{equation}\label{eqt:logisticModel}
	\pr(y = 1) 
		= \expit\left\{ \beta_0  
			+ \left(\sum_{j=1}^5 \beta_j x_j \right)  
			+ \left(\sum_{g=2}^7 \phi_g z_g \right)  
		\right\},
\end{equation}
where $z_g = \mathbb{1}\{x_{6} = g\}$ for $g=1,\ldots, 7$. 
We are interested in testing whether the occupation class label $x_6$ 
is predictive in the presence of effects of all other covariates, that is,  
to test $H_0: \phi_2=\cdots=\phi_7=0$.  
Note that the full model (\ref{eqt:logisticModel}) is also considered by \cite{WangZhuMa2018}.

Although one may perform either Wald's test, LR test or RS test,
performing RS test is substantially easier than the other two in this example. 
It is because the unrestricted MLEs of the regression coefficients 
do not have closed-form expressions, 
iterative root-finding algorithms are needed. 
On the other hand, the null-restricted MLEs are trivial as they are simply the null parameter values.
Hence, it is more appealing (in terms of computational cost and programming simplicity) 
to employ RS test.
In particular, Wald's test takes 59\% more computational time than RS test, 
whereas LR test requires twice as much computational time as RS test. 
We found that $H_0$ is rejected at any reasonable level of significance 
by the proposed MI test R4, no matter Wald's test, LR test or RS test is used.
In other words, the occupation class label $x_6$ is statistically significant despite its 6\% missingness.

\section{Proofs}\label{sec:proof}
\subsection{Proof of Theorem \ref{thm:asy_Distof_bar_d_S}}\label{sec:proof_thm:asy_Distof_bar_d_S}
Let $T_1 := n I^{-1}_{\obs}$, $V_1 := nI^{-1}_{\com}$ and $B_1 := T_1 - V_1$,
where $I_{\com}$ and $I_{\com}$ are 
the complete-data and missing-data expected Fisher's information of $\theta$
defined in (\ref{eqt:FI}).
Note that the subscript ``$1$'' in $T_1$, $V_1$, and $B_1$
emphasizes that they have been standardized by $n$ to represent the contribution of one observation.
By definition, $r_1,\ldots,r_k$ are the eigenvalues of $I_{\mis} I^{-1}_{\obs} = B_1 V_1^{-1}$, 
so they are also the eigenvalues of $V_1^{-1/2}B_1 V_1^{-1/2}$, 
where $V_1^{1/2}$ is the symmetric square root of $V_1$.
Let $R := \diag(r_1, \ldots, r_k)$ be a diagonal matrix with diagonal entries $r_1, \ldots, r_k$.

By spectral decomposition, 
we have $V_1^{-1/2}B_1 V_1^{-1/2} = QRQ^{\T}$
for some orthogonal matrix $Q$.  
Then, 
$B_1 = V_1^{1/2} QRQ^{\T} V_1^{1/2}$, 
$T_1 = V_1^{1/2} Q(R+I_k)Q^{\T} V_1^{1/2}$ and 
$V_1 = V_1^{1/2} Q^{\T} Q V_1^{1/2}$, 
where $I_k$ is an $k\times k$ identity matrix. 
So, we can define matrix square roots of $B_1$, $T_1$ and $V_1$ by 
\begin{equation}\label{eqt:sqrt_B_T}
	\tilde{B}_{1}^{1/2} := V_{1}^{1/2} QR^{1/2} , \quad
	\tilde{T}_{1}^{1/2} := V_1^{1/2} Q(R+I_k)^{1/2}, \quad 
	\tilde{V}_{1}^{1/2} := V_1^{1/2} Q
\end{equation}
so that $\tilde{B}_1^{1/2}(\tilde{B}_1^{1/2})^{\T} = B_1$,
$\tilde{T}_1^{1/2}(\tilde{T}_1^{1/2})^{\T} = T_1$ and 
$\tilde{V}_1^{1/2}(\tilde{V}_1^{1/2})^{\T} = V_1$.
It is remarked that 
$\tilde{B}_1^{1/2}$, $\tilde{T}_1^{1/2}$ and $\tilde{V}_1^{1/2}$ may not be symmetric matrices. 

Let $\widehat{\theta}^S := \widehat{\theta}(X^S)$ 
and $\widehat{V}^S := \widehat{V}(X^S)$ 
be the MLEs of $\theta$ and its covariance matrix estimator 
given the stacked dataset $X^S$. 
By Lemma 1 of \citet{wangRobins1998}, 
the estimator $\widehat{\theta}^S$ is asymptotically equivalent to $\bar{\theta}^S := \sum_{\ell\in S} \widehat{\theta}^{\ell}/|S|$ in the sense that $\sqrt{n}(\widehat{\theta}^S - \bar{\theta}^S) \inP 0_k$.
Under Conditions \ref{cond-Normality}--\ref{cond-properImp} and 
Definition \ref{cond-localH1} with $A = \tilde{V}_1^{-1/2}$,
we have the joint asymptotic ($n\rightarrow\infty$) representations: 
\begin{eqnarray}
	\sqrt{n} (\theta^{\star} - \theta_0) &\rightarrow& \tilde{V}_1^{1/2} \delta , \nonumber\\
	\sqrt{n}\left( \widehat{\theta}_{\obs} - \theta^{\star} \right) &\inD& \tilde{T}_1^{1/2} W, \nonumber\\
	\left\{ \sqrt{n}\left( \widehat{\theta}^{\ell} - \widehat{\theta}_{\obs}  \right) \mid X_{\obs} \right\} &\inD& \tilde{B}_1^{1/2} Z_{\ell}, \qquad \ell=1,\ldots, m, \label{eqt:representationthetaHatell}
\end{eqnarray}
where $W, Z_1, \ldots, Z_m\sim \Normal_k(0,I_k)$ independently.
Let $\bar{Z}_{(S)} := \sum_{\ell\in S}Z_{\ell}/|S|$ for some non-empty multiset 
$S$ with elements from $\{1,\ldots,m\}$.
Summing the above three representations and taking average over $\ell\in S$, 
we have, unconditionally, that  
\begin{eqnarray}
	\sqrt{n}\left( \widehat{\theta}^{S} - \theta_0 \right) 
		&\inD& \tilde{V}_1^{1/2} \delta + \tilde{T}_1^{1/2} W +  \tilde{B}_1^{1/2} \bar{Z}_{(S)}  \nonumber\\
		&=& V_1^{1/2}Q \left\{ \delta + (R+I_k)^{1/2} W +  R^{1/2} \bar{Z}_{(S)} \right\}, \label{eqt:asy_rep_wt_oneSide}
\end{eqnarray}
where (\ref{eqt:asy_rep_wt_oneSide}) is followed from (\ref{eqt:sqrt_B_T}). 
See Remark \ref{rem:alternativeProofRep} for an alternative proof of 
(\ref{eqt:asy_rep_wt_oneSide}).
Using the fact that $V = nV_1$ and Condition \ref{cond-properImp},  
we also know that
\begin{eqnarray}
	n|S| \widehat{V}^S \inP V_1.
	\label{eqt:VSconv}
\end{eqnarray} 
Consequently, applying the Wald's test device (\ref{eqt:com_WT}) on $X^S$, 
we have
\begin{align*}
	\mathcal{d}_{\wt}(X^S)
		&= \left( \widehat{\theta}^{S} - \theta_0 \right)^{\T}
			\left( \widehat{V}^S \right)^{-1} 
			\left( \widehat{\theta}^{S} - \theta_0 \right) \nonumber \\
		&= |S| \sqrt{n}\left( \widehat{\theta}^{S} - \theta_0 \right)^{\T}
			\left( n|S| \widehat{V}^S \right)^{-1} 
			\sqrt{n}\left( \widehat{\theta}^{S} - \theta_0 \right) .
\end{align*}
By (\ref{eqt:VSconv}), Slutsky's lemma and orthogonality of $Q$, we have  
\begin{align}
	\mathcal{d}_{\wt}(X^S)	
		&\inD |S| \left\{ \delta +  (R+I_k)^{1/2} W +  R^{1/2} \bar{Z}_{(S)} \right\}^{\T} \nonumber \\
			&\qquad\qquad \left\{ \delta + (R+I_k)^{1/2} W +  R^{1/2} \bar{Z}_{(S)} \right\} .\label{eqt:proof_rep_of_d}
\end{align}
Since $\mathcal{d}_{\wt}(X^S)$, $\mathcal{d}_{\lrt}(X^S)$ 
and $\mathcal{d}_{\rt}(X^S)$ are asymptotically equivalent under local alternative hypothesis, 
we obtain (\ref{eqt:representation_of_d}) after 
expanding (\ref{eqt:proof_rep_of_d}).
Finally, note that the asymptotic representation (\ref{eqt:representationthetaHatell})
is true jointly for all $\ell=1, \ldots,m$. 
Hence, (\ref{eqt:proof_rep_of_d}) is also true for all multiset $S$.

\begin{remark}\label{rem:alternativeProofRep}
As pointed out by an anonymous referee, 
the proof of Theorem \ref{thm:asy_Distof_bar_d_S} is clearer if the estimator $\widehat{\theta}$
is assumed to be asymptotically linear as in \cite{wangRobins1998}. 
For reference, we outline the proof below.

We need the following notations.
Denote  
the observed-data score,
the complete-data score, and the missing-data score 
contributed by the $i$th subject in the $\ell$th imputed dataset by 
$U_{i}^{\obs}$, $U_{i\ell}^{\com}$ and $U_{i\ell}^{\mis}=U_{i\ell}^{\com}-U_{i}^{\obs}$, respectively. 
Also let 
\[
	\mathcal{i}_{\com} = I_{\com}/n, \qquad
	\mathcal{i}_{\obs} = I_{\obs}/n, \qquad
	\mathcal{i}_{\mis} = I_{\mis}/n
\]
be the Fisher's information contributed by one observation only.

Assume, in addition, that 
all regularity conditions stated in Theorem 2 of \cite{wangRobins1998}
are satisfied. 
The equation (A3)
in their paper states an asymptotic representation of 
the classical MI estimator $\widehat{\theta}^{\{1:m\}}$
by using influence functions. 
With some trivial modifications, we can extend it to $\widehat{\theta}^{S}$ for any multiset $S$: 
\begin{align*}
	\sqrt{n}\left( \widehat{\theta}^S - \theta^{\star} \right)
		&= \Psi_0
				+ \frac{1}{|S|}\sum_{\ell\in S} \Psi_{\ell} +o_p(1),\label{eqt:representation_thetaHatS_linear}
\end{align*}
where 
\begin{eqnarray*}
	\Psi_0 = \frac{1}{\sqrt{n}}\sum_{i=1}^n \mathcal{i}_{\obs}^{-1}U_i^{\obs}
	\quad \text{and}\quad
	\Psi_{\ell} = \frac{1}{\sqrt{n}} \sum_{i=1}^n \mathcal{i}_{\com}^{-1}U_{i\ell}^{\mis} 
				+ \sqrt{n} \mathcal{i}_{\obs}^{-1} \mathcal{i}_{\mis}
					\left( \widehat{\theta}^{\ell} - \widehat{\theta}_{\obs} \right).
\end{eqnarray*}
As argued by \cite{wangRobins1998},  
we know that, as $n\rightarrow\infty$,  
$\Psi_0 \inD\Normal_k(0, \mathcal{i}_{\obs}^{-1})$
and, for each $\ell$,  	
\begin{align*}
	\left[ \Psi_{\ell} \mid X_{\obs} \right] 
		&\inD \Normal_k\left(0, 
					\mathcal{i}_{\com}^{-1}\mathcal{i}_{\mis}\mathcal{i}_{\com}^{-1}
					+ (\mathcal{i}_{\mis}\mathcal{i}_{\com}^{-1})^{\intercal} \mathcal{i}_{\obs}^{-1} (\mathcal{i}_{\mis}\mathcal{i}_{\com}^{-1}) \right) \nonumber \\
		&=\Normal_k\left(0, \mathcal{i}_{\obs}^{-1} \mathcal{i}_{\mis} \mathcal{i}_{\com}^{-1}\right) ,
\end{align*}
where the last line is 
followed from the identity $I_k + \mathcal{i}_{\mis}\mathcal{i}_{\obs}^{-1} = \mathcal{i}_{\com}\mathcal{i}_{\obs}^{-1}$.
Recall that $T_1 = \mathcal{i}^{-1}_{\obs}$, $V_1 = \mathcal{i}^{-1}_{\com}$ and $B_1 = T_1 - V_1$.
So, 
\[
	\Psi_0 \inD\Normal_k(0, T_1)\qquad\text{and}\qquad
	\left[ \Psi_{\ell} \mid X_{\obs} \right] \inD \Normal_k\left(0, B_1 \right)	.
\]
Since $\Psi_{1}, \ldots, \Psi_{m}$ are conditionally independent given $X_{\obs}$, 
we have the following asymptotic representation (jointly over all $S$):
\begin{eqnarray*}
	\sqrt{n}\left( \widehat{\theta}^{S} - \theta^{\star} \right) 
		\inD \tilde{T}_1^{1/2} W +  \tilde{B}_1^{1/2} \bar{Z}_{(S)} , 
\end{eqnarray*}
where the matrix square roots $\tilde{T}_1^{1/2}$ 
and $\tilde{B}_1^{1/2}$ are defined in (\ref{eqt:sqrt_B_T}).
Using the local alternative representation in Definition \ref{cond-localH1} with $A = \tilde{V}_1^{-1/2}$, 
we obtain (\ref{eqt:asy_rep_wt_oneSide}).
Following the remaining part of the proof above, we obtain the desired result. 
\end{remark}

\subsection{Proof of Proposition \ref{prop:limit_of_J}}\label{sec:proof_prop:limit_of_J}
For part 1, it is a direct application of Theorem \ref{thm:asy_Distof_bar_d_S}.
For part 2, note that
\begin{eqnarray}\label{eqt:diffZbar}
	\bar{Z}_{j(S_1)} - \bar{Z}_{j(S_2)} 
		\sim \Normal\left(0, \frac{|S_1|+|S_2|-2|S_1\cap S_2|}{|S_1|\times|S_2|} \right) 
\end{eqnarray}
independently over $j =1,\ldots,k$. Then summing the square of (\ref{eqt:diffZbar}) over $j=1, \ldots, k$, 
we obtain the desired result. 
For part 3, consider $\mathbb{T} = \sum_{j=1}^k r_j U_j$, 
where $U_1,\ldots, U_k\sim \chi^2_1$ independently.
Recall that the $\tau$th order cumulant of $U \sim \chi^2_1$ is
$\kappa_{\tau}(U) = 2^{\tau-1}(\tau-1)!$, 
where $\tau = 1, \ldots, k$. 
By additivity and homogeneity of cumulants, 
the $\tau$th cumulant of $\mathbb{T}$ is given by  
\[
	\kappa_{\tau}(\mathbb{T}) = \sum_{j=1}^k 2^{\tau-1}(\tau-1)! r_j^{\tau}.
\]
Using the recursion formula between moments and cumulants, 
we can derive the expression of $\E(\mathbb{T}^{\tau}) = t_{\tau}$ for each $\tau = 1, \ldots, k$.

\subsection{Proof of Theorem \ref{thm:limit_EJ}}\label{sec:proof_thm:limit_EJ}
For part (1), 
it follows from part 3 of Proposition \ref{prop:limit_of_J} and linearity of expectation. 
For part (2), 
By straightforward calculations, we know that
\[
	\Corr\left\{ \bar{Z}_{j(S_1)}-\bar{Z}_{j(S_2)}, \bar{Z}_{j'(S_3)}-\bar{Z}_{j'(S_4)} \right\}
		= \left\{ 
		\begin{array}{ll}
			\rho(S_1,S_2,S_3,S_4) & \text{if $j=j'$}; \\
			0 & \text{if $j\neq j'$}.
		\end{array}
		\right.
\] 
For simplicity, write $\rho = \rho(S_1,S_2,S_3,S_4)$.
We emphasize that $\rho$ depends on $S_1,S_2,S_3,S_4$.
From part 2 of Proposition \ref{prop:limit_of_J}, 
we can represent $\mathbb{T} = \sum_{j=1}^k r_i G_j^2$, 
where $G_1, \ldots, G_k\sim N(0,1)$ independently. 
Together with part 1 of Proposition \ref{prop:limit_of_J}, we have 
$\widehat{T}_{S_1, S_2}\inD \mathbb{T} = \sum_{j=1}^k r_i G_j^2$ marginally for each $(S_1,S_2)\in\Lambda$. 
By Condition \ref{cond-UI} and Cholesky decomposition, we have 
\begin{align}\label{eqt:proof_var_barJtau}
	\Var\left\{ \widehat{t}_{\tau}(\Lambda) \right\}
		&\rightarrow  \frac{1}{|\Lambda|^2} \mathop{\sum\sum}_{(S_1,S_2),(S_2,S_4)\in\Lambda} \nonumber \\
		&\qquad
			\Cov\left( \left( \sum_{j=1}^k r_j G_j^2 \right)^{\tau} , 
				\left[ \sum_{j=1}^k r_j \left\{\rho G_i + (1-\rho^2)^{1/2} W_i \right\} ^2 \right]^{\tau}
			\right) , 
\end{align}
where $W_1,G_1, \ldots, W_k,G_k\sim N(0,1)$ are all independent.
Denote the covariance term in (\ref{eqt:proof_var_barJtau}) by $\Xi=\Xi(S_1,S_2,S_3,S_4)$.
Define $Q_1 := \sum_{j=1}^k r_j G_j^2$, $Q_2 := \sum_{j=1}^k r_j W_j^2$
and $Q_{12} := \sum_{j=1}^k r_j G_j W_j$. 
By trinomial expansion,  
\begin{align}
	\Xi &= \Cov\left[ Q_1^{\tau}, \left\{ \rho^2Q_1 + (1-\rho^2) Q_2 + 2\rho(1-\rho^2)^{1/2}Q_{12} \right\}^{\tau}\right]  \nonumber \\
		&= \sum_{a=0}^{\tau} \sum_{b=0}^a {{\tau}\choose{a}}{{a}\choose{b}} 2^{a-b} \rho^{a+b} (1-\rho^2)^{\tau - (a+b)/2} 
			\Cov\left(Q_1^{\tau}, Q_1^b Q_2^{\tau-a} Q_{12}^{a-b}\right)  \nonumber\\
		&= \rho^2 C_{\tau}(S_1,S_2,S_3,S_4) , \label{eqt:Xi_final}
\end{align}
where 
\begin{align*}
	&C_{\tau}(S_1,S_2,S_3,S_4)\\
		&\qquad:= \sum_{a=1}^{\tau} \sum_{b=1}^a {{\tau}\choose{a}}{{a}\choose{b}} 2^{a-b} \rho^{a+b-2} (1-\rho^2)^{\tau - (a+b)/2} 
			\Cov\left(Q_1^{\tau}, Q_1^b Q_2^{\tau-a} Q_{12}^{a-b}\right) . 
\end{align*}
Note that (\ref{eqt:Xi_final}) is obtained by using 
the fact that $\Cov\left(Q_1^{\tau}, Q_1^b Q_2^{\tau-a} Q_{12}^{a-b}\right) = 0$ 
for $(a,b)=(0,0),(0,1),(1,0)$.
Putting (\ref{eqt:Xi_final}) into (\ref{eqt:proof_var_barJtau}), 
we obtain (\ref{eqt:limVarJbar}).
Clearly, if $\tau=1$, then 
\begin{eqnarray*}
	C_{\tau}(S_1,S_2,S_3,S_4)
	= \Var(Q_1)
	= 2 \sum_{j=1}^k r_j^2.
\end{eqnarray*}

\subsection{Proof of Corollary \ref{coro:example_var}}\label{sec:proof_coro:example_var}
We start with $\Lambda_{\Jack}$.
Let $S_1 =\{\ell\}$, 
$S_2 = \{1,\ldots,m\} \setminus \{\ell\}$
$S_3 = \{\ell'\}$, 
$S_4 = \{1,\ldots,m\} \setminus \{\ell'\}$,
where $\ell, \ell'\in\{1,\ldots,m\}$. 
Clearly,  $(S_1,S_2), (S_3,S_4)\in\Lambda_{\Jack}$.
Note that $|\Lambda_{\Jack}|=m$ and 
\begin{eqnarray}\label{eqt:rho_Lambda1}
	\rho(S_1,S_2,S_3,S_4)
	= \left\{
	\begin{array}{ll}
	1	& \text{if $\ell=\ell'$};\\
	-1/(m-1) & \text{if $\ell\neq\ell'$}.
	\end{array}
	\right.
\end{eqnarray}
By part 1 of Theorem \ref{thm:limit_EJ}, 
$\E\{\widehat{T}_{1}(\Lambda_{\Jack})\} \rightarrow t_1$.
By part 2 of Theorem \ref{thm:limit_EJ}, we have 
\begin{eqnarray*}
	\Var\{\widehat{T}_{1}(\Lambda_{\Jack})\} 
		&\rightarrow& \frac{1}{m^2} \left( 2R_2\right) \left\{ m(1)^2 + (m^2-m)\left(\frac{-1}{m-1}\right)^2\right\} \\
		&=& 2R_2/(m-1), 
\end{eqnarray*}
and 
\begin{eqnarray*}
	\Var\{\widehat{T}_{\tau}(\Lambda_{\Jack})\}
	&\rightarrow& \frac{1}{m^2} \left\{ m(1)^2 C_1  - (m^2-m)\left(\frac{-1}{m-1}\right)^2 C_2\right\} \\
	&=& O(1/m),
\end{eqnarray*}
where $C_1,C_2$ are some finite constants. 

Next we consider $\Lambda_{\Full}$.
Let $S_1 =\{\ell\}$, 
$S_2 = \{1,\ldots,m\}$
$S_3 = \{\ell'\}$, 
$S_4 = \{1,\ldots,m\}$,
where $\ell, \ell'\in\{1,\ldots,m\}$. 
Clearly,  
$(S_1,S_2), (S_3,S_4)\in\Lambda_{\Full}$.
Note that $|\Lambda_{\Full}|=m$, and 
$\rho(S_1,S_2,S_3,S_4)$ admits the same form as (\ref{eqt:rho_Lambda1}).
Thus, the result follows. 

Finally, we consider $\Lambda_{\Pair}$.
Let $S_i =\{\ell_i\}$ for $i=1,2,3\Lambda_{\Pair},4$, 
where 
$\ell_1,\ell_2,\ell_3,\ell_4\in\{1,\ldots,m\}$
such that $\ell_1\neq \ell_2$ and 
$\ell_3\neq \ell_4$.
So, $(S_1,S_2), (S_3,S_4)\in\Lambda_{\Full}$.
We have
$|\Lambda_{\Pair}|={m\choose 2}$ and 
\begin{eqnarray*}
	\rho(S_1,S_2,S_3,S_4)
	= \left\{
	\begin{array}{ll}
	1	& \text{if $|\{\ell_1,\ell_2,\ell_3,\ell_4\}|=2$};\\
	1/2 & \text{if $|\{\ell_1,\ell_2,\ell_3,\ell_4\}|=3$};\\
	0 & \text{if $|\{\ell_1,\ell_2,\ell_3,\ell_4\}|=4$}.
	\end{array}
	\right.
\end{eqnarray*}
Again, by part 1 of Theorem \ref{thm:limit_EJ}, 
$\E\{\widehat{T}_{1}(\Lambda_{\Pair})\} \rightarrow t_1$.
Among all ${m\choose 2}^2$ combinations of 
$(S_1,S_2), (S_3,S_4)\in\Lambda_{\Full}$, 
there are ${m\choose 2}$ of them satisfying $|\{\ell_1,\ell_2,\ell_3,\ell_4\}|=2$, and  
${m\choose 2}{m-2\choose 2}$ of them satisfying 
$|\{\ell_1,\ell_2,\ell_3,\ell_4\}|=4$.
Hence, 
there are ${m\choose 2}^2-{m\choose 2}-{m\choose 2}{{m-2}\choose 2}$ of them satisfying $|\{\ell_1,\ell_2,\ell_3,\ell_4\}|=3$.
By part 2 of Theorem \ref{thm:limit_EJ}, 
\begin{eqnarray*}
	\Var\{\widehat{T}_{1}(\Lambda_{\Pair})\} 
	&\rightarrow& {m\choose 2}^{-2} \left( 2R_2\right) \left[ {m\choose 2} + \frac{1}{2}\left\{{m\choose 2}^2-{m\choose 2}-{m\choose 2}{{m-2}\choose 2} \right\} \right]\\
	&=& 2R_2/(m-1), 
\end{eqnarray*}
and 
\begin{eqnarray*}
	\Var\{\widehat{T}_{\tau}(\Lambda_{\Pair})\} 
	&\rightarrow& {m\choose 2}^{-2} 
		\left[ {m\choose 2} C_3 + \left\{{m\choose 2}^2-{m\choose 2}-{m\choose 2}{{m-2}\choose 2} \right\}C_4 \right]\\
	&=& O(1/m),
\end{eqnarray*}
where $C_3,C_4$ are some finite constants.

\subsection{Proof of Proposition \ref{prop:inverse}}\label{sec:proof_prop:inverse}
For part 1, we can solve for 
		$r_1,\ldots, r_k$ by solving the Newton's identities for the power sums $R_{\tau} = \sum_{j=1}^k r_j^{\tau}$, $\tau=1,\ldots,k$; 
		see, for example, \citet{Kalman2000}. 
		
For part 2, fix any $\tau \in\{ 1, \ldots, k\}$. By part 3 of Proposition \ref{prop:limit_of_J},
		we know that $R_1, \ldots, R_{\tau}$ are a linear transformation of $t_1, \ldots, t_{\tau}$.
		By recursively solving a system of linear equations, 
		we can express $R_{\tau}$ 
		in terms of $t_1,\ldots, t_{\tau}$ for each $\tau = 1, \ldots, k$. 
		Therefore, the desired result (\ref{eqt:estimator_Rstar}) can be derived.

\subsection{Proof of Corollary \ref{coro:rHat}}\label{sec:proof_coro:rHat}
Let $\| \xi \| := \sqrt{\E (\xi^{\T}\xi)}$ for any random vector $\xi$. 
By Corollary \ref{coro:example_var}, we know, for each $j=1, \ldots,k$, that
$\| \widehat{t}_j - t_j \| \rightarrow \sqrt{V_j(m)}$ as $n\rightarrow\infty$, 
where $V_j(m)$ is some function satisfying $V_j(m) = O(1/m)$.
Let $M(\cdot):=M_1^{-1}(M_2^{-1}(\cdot))$, which is differentiable at $t_{1:k}$. 
Denote its grandaunt function by $\dot{M}(\cdot)$. 
Applying Taylor's expansion on $M(\cdot)$, we have 
\begin{eqnarray}\label{eqt:MMhat}
	M(\widehat{t}_{1:k}) = M(t_{1:k}) + \dot{M}(t_{1:k})  (\widehat{t}_{1:k}-t_{1:k}) + \Delta,
\end{eqnarray}
where 
$M(\widehat{t}_{1:k}) = \widehat{r}_{1:k}$ by (\ref{eqt:initialEst_rR}),  
$M(t_{1:k}) = r_{1:k}$ by Proposition \ref{prop:inverse}, and 
$\|\Delta\| = o( \|\widehat{t}_{1:k}-t_{1:k}\| )$ by Condition \ref{cond-UI}.
Fix any $j\in\{1, \ldots, k\}$, and 
consider the $j$th component of (\ref{eqt:MMhat}). 
By Minkowski's inequality, we have, as $n\rightarrow\infty$, that 
\begin{eqnarray*}
	\sqrt{\MSE(\widehat{r}_j)} 
		\;=\; \| \widehat{r}_j - r_j \|
		&\leq& \left\| \sum_{j'=1}^k \dot{M}_{jj'}(t_{1:k}) (\widehat{t}_{j'}-t_{j'})  \right\|
				+ \| \Delta_j \| \nonumber \\
		&\leq& \sum_{j'=1}^k \dot{M}_{jj'}(t_{1:k}) \left\| \widehat{t}_{j'}-t_{j'} \right\|
				+ \| \Delta \| \nonumber \\
		&\rightarrow&  \sum_{j'=1}^k \dot{M}_{jj'}(t_{1:k}) \sqrt{V_{j'}(m)} + o(1/\sqrt{m}) \nonumber \\
		&=:& \sqrt{V(m)} , \label{eqt:normrHat}
\end{eqnarray*}
where $\dot{M}_{jj'}(t_{1:k})$ is the $(j,j')$th entry of $\dot{M}(t_{1:k})$, and 
$\Delta_j$ is the $j$th entry of $\Delta$. 
Since $\sqrt{V_{j'}(m)}=O(1/\sqrt{m})$ for each $j'$, we know that $\sqrt{V(m)} = O(1/\sqrt{m}) = o(1)$.
Hence, we obtain $\MSE(\widehat{r}_j) \rightarrow V(m)$ as $n\rightarrow\infty$, where 
$V(m)$ satisfies $V(m)=o(1)$ as $m\rightarrow\infty$.

\subsection{Proof of Proposition \ref{prop:limit_of_Dhat}}\label{ref:proof_prop:limit_of_Dhat}
To prove part 1, we first recall the following facts.
\begin{itemize}
	\item The testing devices $\mathcal{d}_{\wt}(\cdot), \mathcal{d}_{\lrt}(\cdot), \mathcal{d}_{\rt}(\cdot)$
			in (\ref{eqt:com_WT})--(\ref{eqt:com_RT})
			are asymptotically equivalent under $H_0$. 
			Hence, they can be asymptotically represented as $\mathcal{d}(\cdot)$. 
	\item According to Section \ref{sec:motivation},
			we may asymptotically represent $\widetilde{d}'_{\wt}$ and $\widetilde{d}'_{\lrt}$ 
			as $\widetilde{d}'$, 
			and represent $\widetilde{d}''_{\wt}$ and $\widetilde{d}''_{\lrt}$ as $\widetilde{d}''$, that is, 
			$\widetilde{d}' \sim \widetilde{d}'_{\wt} \sim \widetilde{d}'_{\lrt}$
			and $\widetilde{d}'' \sim \widetilde{d}''_{\wt} \sim \widetilde{d}''_{\lrt}$.
			By the same token, we can represent $\widetilde{D} \sim \widetilde{D}_{\wt} \sim \widetilde{D}_{\lrt}$; 
			see (\ref{eqt:asy_version_D}).
	\item From (\ref{eqt:dp_dpp}) and (\ref{eqt:barmurDef}), we have  
		$\widetilde{d}' \sim \widehat{d}' =  \sum_{\ell=1}^m \widehat{d}^{\{\ell\}}/m$ 
			and 
			$\widetilde{d}'' \sim \widehat{d}'' = \widehat{d}^{\{1:m\}}$.
			It implies that $\overline{\mu}_r - \widetilde{\mu}_r \inP 0$, 
			where $\overline{\mu}_r$ in defined in (\ref{eqt:barmurDef}); 
			and $\widetilde{\mu}_r$ can be either $\widetilde{\mu}_{r,\wt}$ or $\widetilde{\mu}_{r,\lrt}$.
	\item From Remark \ref{rem:asy_equiv_of_murHat},
			we have $\overline{\mu}_r - \widehat{\mu}_r \inP 0$. 
\end{itemize}
Hence, 
\begin{eqnarray*}
	\widetilde{D}
	= \frac{\widetilde{d}''}{k\left\{ 1 + (1+\frac{1}{m})\widetilde{\mu}_r \right\} }
	\sim \frac{\widehat{d}^{\{1:m\}}}{k\left\{ 1 + (1+\frac{1}{m})\overline{\mu}_r \right\} }
	\sim \frac{\widehat{d}^{\{1:m\}}}{k\left\{ 1 + (1+\frac{1}{m})\widehat{\mu}_r \right\} } 
	= \widehat{D},
\end{eqnarray*}
where $\widehat{D}$ can be either $\widehat{D}_{\wt}$, $\widehat{D}_{\lrt}$ or $\widehat{D}_{\rt}$; and 
$\widetilde{D}$ can be either $\widetilde{D}_{\wt}$ or $\widetilde{D}_{\lrt}$.
Hence,  $\widehat{D}- \widetilde{D} \inP 0$. 

For part 2, note that $\widehat{D}$ is a function of $\widehat{d}^{\{1\}},\ldots, \widehat{d}^{\{m\}},\widehat{d}^{\{-1\}}, \ldots, \widehat{d}^{\{-m\}}, \widehat{d}^{\{1:m\}}$.
Applying Theorem \ref{thm:asy_Distof_bar_d_S} 
we know that $\widehat{D} \inD \mathbb{D}$
such that the weak limit $\mathbb{D}$ admits the following asymptotic ($n\rightarrow\infty$) representation: 
\begin{eqnarray}\label{eqt:mathbbD_ver1}
	\mathbb{D}
		= \frac{m\sum_{j=1}^k \left\{ (1+r_j)^{1/2}W_j + r_j^{1/2}\overline{Z}_{j\bullet} \right\}^2}{mk + (m+1)\sum_{j=1}^k r_j s_{Z_j}^2 },
\end{eqnarray}
where $s_{Z_j}^2 := (m-1)^{-1} \sum_{\ell=1}^m (Z_{j\ell} - \overline{Z}_{j\bullet})^2$ is the 
sample variance of $\{Z_{j\ell}\}_{\ell=1}^m$.
Note that 
\begin{itemize}
	\item the set of random variables $\{ W_j, \overline{Z}_{j\bullet}, s_{Z_j}^2\}_{j=1}^k$ 
			are mutually independent; and 
	\item $(1+r_j)^{1/2}W_j + r_j^{1/2}\overline{Z}_{j\bullet} \sim \Normal(0, 1+(1+\frac{1}{m})r_j)$
			and 
			$s_{Z_j}^2 \sim \chi^2_{m-1}/(m-1)$.
\end{itemize}
Using these two facts, 
we can simplify (\ref{eqt:mathbbD_ver1}) to 
\begin{eqnarray*}\label{chp2-eqt:simpleRepDm}
	\mathbb{D}
		\;=\; \frac{m\sum_{j=1}^k \left\{ 1+(1+\frac{1}{m})r_j \right\} G_j^2}{mk + (m+1)\sum_{j=1}^k r_j H_j^2 } 
		\;=\; \frac{\frac{1}{k}\sum_{j=1}^k \left\{ 1+ (1+\frac{1}{m})r_{j} \right\} G_j}{ 1+\frac{1}{k}\sum_{j=1}^k (1+\frac{1}{m})r_{j} H_j}. 
\end{eqnarray*}
where $G^2_{1}, \ldots, G^2_k \simIID \chi^2_1$ and $H^2_1, \ldots, H_k^2 \simIID \chi^2_{m-1}/(m-1)$ are mutually independent.

\subsection{Proof of Corollary \ref{coro:unbiasedEstofVarlambda}}\label{sec:proof_unbiasedEstofVarlambda}
For part 1,  
recall that $\sigma^2_{\lambda} = R_2/k - (R_1/k)^2$.
In order to derive an asymptotically unbiased estimators for $\sigma^2_r$, 
it suffices to find asymptotically unbiased estimators for $R_2$ and $R_1^2$.
By Corollary \ref{coro:example_var}, 
we have, as $n\rightarrow\infty$, that 
\begin{eqnarray*}
	\begin{bmatrix}
		\E ( \widehat{t}_1^2  )\\
		\E (  \widehat{t}_2   )
	\end{bmatrix}
	\rightarrow
	\left[ 
	\begin{array}{cc}
		2/(m-1) 	\quad	& 1 \\
		2 \quad & 1
	\end{array}
	\right]
	\begin{bmatrix} 
		R_2\\
		R_1^2 
	\end{bmatrix}.
\end{eqnarray*}
Then solving for $R_1^2$ and $R_2$, we have 
\begin{eqnarray}\label{eqt:unbiased_comp}
	\frac{1}{2(m-2)}
	\begin{bmatrix}
		(m-1)\E( \widehat{t}_2 - \widehat{t}_1^2 ) \\
		2 \E\{ (m-1)\widehat{t}_1^2  - \widehat{t}_2 \} 
	\end{bmatrix}
	\rightarrow
	\left[ 
	\begin{array}{c}
		R_2\\
		R_1^2 
	\end{array}
	\right].
\end{eqnarray}
Since $\widehat{\sigma}^2_r$ is a linear combination of the left hand side of (\ref{eqt:unbiased_comp}), 
we can easily show that $\E(\widehat{\sigma}^2_r) \rightarrow R_2/k - (R_1/k)^2 = \sigma^2_{r}$
as $n\rightarrow\infty$. 

For part 2, note that $\widehat{\sigma}^2_r = C_1\widehat{t}_2 + C_2\widehat{t}_1^2$, 
		where $C_1 = \{k(m-1)+2\}/\{2k^2(m-2)\}$ and $C_2 = -(m-1)(k+2)/\{2k^2(m-2)\}$.
		Since $\Var(\widehat{\sigma}^2_r) \rightarrow C_1 \Var(\widehat{t}_2) + C_2 \Var(\widehat{t}_1^2) + 2C_1 C_2 \Cov(\widehat{t}_2,\widehat{t}_1^2)$, 
		it suffices to find the two variances and the covariance. 
		We find them one by one below. 
		Let $m_0 = m-1$ for simplicity.
		\begin{itemize}
			\item We start with deriving the limit of 
				$\Var(\widehat{t}_2) = \E(\widehat{t}_2^2) - \E^2(\widehat{t}_2)$.
				By Corollary \ref{coro:example_var}, $\E(\widehat{t}_2) \rightarrow 2R_2+R_1^2$. 
				By Proposition \ref{prop:limit_of_J}, 
				$\mathbb{T}_{\ell} := \mathbb{T}_{\{\ell\},\{-\ell\}}$, $\ell=1, \ldots, m$, 
				are identically distributed, 
				and $(T_{1}, \ldots, T_m) \inD (\mathbb{T}_1, \ldots, \mathbb{T}_{m})$. 
				Moreover, the values of $\E(T_{\ell}T_{\ell'})$'s are identical whenever $\ell\neq \ell'$. 
				Hence, we have 
				\begin{eqnarray}
					\E(\widehat{t}_2^2)
						&=& \left\{\frac{1}{m^2}\sum_{\ell=1}^m \E\left( T_{\ell}^4 \right)\right\}
							+ \left\{\frac{1}{m^2} \mathop{\sum\sum}_{\substack{\ell,\ell'\in\{1,\ldots,m\}\\\ell\neq \ell'}} \E\left( T_{\ell}^2 T_{\ell'}^2\right)\right\}\nonumber \\
						&=& \frac{1}{m} \E\left( T_1^4 \right) + \frac{m-1}{m} \E\left( T_{1}^2 T_{2}^2\right) . \label{chp2-eqt:mathcalA12} 
				\end{eqnarray}
				The limiting values of the expectations in (\ref{chp2-eqt:mathcalA12})
				can be found by Lemma \ref{chp2-lemma:E_combination_of_L}, 
				which will be presented in Section \ref{chp2-sec:techLemma}.
				Applying Lemma \ref{chp2-lemma:E_combination_of_L}
				and performing some simple algebras,  
				we have,  
				as $n\rightarrow\infty$, that 
				\begin{eqnarray*}\label{chp2-eqt:fianl_EsqL2}
					\E(\widehat{t}_2^2)
						&\rightarrow& \frac{1}{m} \E\left( \mathbb{T}_1^4 \right) + \frac{m-1}{m} \E\left( \mathbb{T}_{1}^2 \mathbb{T}_{2}^2\right)\nonumber \\
						&=& \frac{1}{m^3m_0}
							\bigg\{ \left( 48m_0^3+32m_0^2 + 16      \right) R_4 
									+ \left( 32m_0^3+32m_0^2           \right) R_3 R_1 \nonumber \\
						&&\qquad + \left(  8m_0^3 + 8      \right) R_2^2 
									+ \left(  8m_0^4 + 8m_0^2 \right) R_2R_1^2 
									\bigg\} .
				\end{eqnarray*}
			\item The derivation of $\Var(\widehat{t}_1^2) = \E(\widehat{t}_1^4) - \E^2(\widehat{t}_1^2)$ 
					is similar.	By doing the combinatoric counting and using Lemma \ref{chp2-lemma:E_combination_of_L}, we have, as $n\rightarrow\infty$, that  
					\begin{eqnarray*}
						\E(\widehat{t}_1^4)
							&=& \frac{1}{m^4} 
								\sum_{\ell_1=1}^m\sum_{\ell_2=1}^m\sum_{\ell_3=1}^m\sum_{\ell_4=1}^m
								\E\left( T_{\ell_1} T_{\ell_2} T_{\ell_3} T_{\ell_4}\right) \nonumber\\
							&\rightarrow& \frac{1}{m^4}\bigg\{ 
									m\E(\mathbb{T}_1^4) + 4(m^2-m)\E(\mathbb{T}_1^3\mathbb{T}_2) + 3(m^2-m) \E(\mathbb{T}_1^2 \mathbb{T}_2^2) \nonumber\\
							&&\qquad			+ 6(m^3-3m^2+2m) \E(\mathbb{T}_1^2 \mathbb{T}_2 \mathbb{T}_3) \nonumber \\
							&&\qquad			+ (m^4-6m^3+11m^2-6m) \E(\mathbb{T}_1\mathbb{T}_2\mathbb{T}_3\mathbb{T}_4)
								\bigg\} \nonumber\\
						&=& \frac{1}{m m_0^7}
							\bigg\{ 24\left( 2m_0^5+2m_0^4 + 3m_0-6      \right) \Lambda_4 
									+ 32m_0^5\left( m_0+1   \right) R_3 R_1 \nonumber \\
						&&\qquad + 12m_0^5\left( m_0+1   \right) R_2^2 
									+ 12m_0^6\left( m_0+1   \right) R_2R_1^2 
									\nonumber \\
						&& \qquad+ m_0^7 (m_0+1) R_1^4
							\bigg\} . \label{chp2-eqt:4ndMoment_L1Bar}
					\end{eqnarray*}
			\item The derivation of $\Cov(\widehat{t}_2,\widehat{t}_1^2) 
					= \E(\widehat{t}_2\widehat{t}_1^2) -\E(\widehat{t}_2)\E(\widehat{t}_1^2)$ 
					is also similar. We have, as $n\rightarrow\infty$, that   								
					\begin{eqnarray*}
						\E(\widehat{t}_2\widehat{t}_1^2)
							&\rightarrow& \frac{1}{m^3}\bigg\{ 
									m\E(\mathbb{T}_1^4) + (m^2-m)\E(\mathbb{T}_1^2\mathbb{T}_2^2) + 2(m^2-m) \E(\mathbb{T}_1^3 \mathbb{T}_2) \\
							&&\qquad	+ (m^3-3m^2+2m) \E(\mathbb{T}_1^2\mathbb{T}_2\mathbb{T}_3)
								\bigg\} \nonumber \\
							&=& \frac{1}{mm_0^3}\bigg\{ 
										  48m_0^2 R_4 
										+ 16m_0^2\left( m_0+1   \right) R_3 R_1 \nonumber \\
							&&\qquad + 4m_0^2\left( m_0+2   \right) R_2^2 
										+ 2m_0^3\left( m_0+5   \right) R_2R_1^2 
										+ m_0^4 R_1^4
								\bigg\} .\label{chp2-eqt:E_L1BarSq_L2Bar}
					\end{eqnarray*}
		\end{itemize}
		Finally, the limit of $\Var(\widehat{\sigma}^2_r)$ can be found by combining the above results.

\subsection{Proof of Proposition \ref{prop:wHat}}\label{sec:proof_prop:wHat}
Note that $\widehat{Q}$ is a function of $\widehat{T}_{1}, \ldots, \widehat{T}_k$. 
By Proposition \ref{prop:limit_of_J}, we have $\widehat{Q}\inD \mathbb{Q}$
such that the weak limit $\mathbb{Q}$ admits the representation defined in (\ref{eqt:conv_of_u}).

\section{Technical Lemmas}\label{chp2-sec:techLemma}
\begin{lemma}\label{chp2-lem:splitFourFoldedSum}
Let $k\in\mathbb{N}$ and $r_1,\ldots,r_k\in\mathbb{R}$.
Define 
\begin{eqnarray*}
	S_{\phi} \equiv \sum_{(i,i',j,j')\in \Xi_{\phi}} r_ir_{i'}r_{j}r_{j'}, 
	\qquad \phi\in\{4,31,22,211,111\},
\end{eqnarray*}
where
{\small
\begin{align}
	\Xi_{4}    &= \left\{ (i,i',j,j')\in\{1,\ldots,k\}^4 : i=i'=j=j' \right\},
	\label{chp2-eqt:partition_of_fourFoldSuq_1}\\
	\Xi_{31}   &= \left\{ (i,i',j,j')\in\{1,\ldots,k\}^4 : i=i'=j\neq j' \right\},
	\label{chp2-eqt:partition_of_fourFoldSuq_2}\\  
	\Xi_{22}   &= \left\{ (i,i',j,j')\in\{1,\ldots,k\}^4 : i=i'\neq j=j' \right\},
	\label{chp2-eqt:partition_of_fourFoldSum_3}\\  
	\Xi_{211}  &= \left\{ (i,i',j,j')\in\{1,\ldots,k\}^4 :  \text{$i=i'$ but none of $i,j,j'$ are equal}  \right\},
	\label{chp2-eqt:partition_of_fourFoldSum_4}\\  
	\Xi_{1111} &= \left\{ (i,i',j,j')\in\{1,\ldots,k\}^4 : \text{None of $i,i',j,j'$ are equal} \right\}.	\label{chp2-eqt:partition_of_fourFoldSum_5}
\end{align}
}
Then the following identity holds:
\begin{align}\label{chp2-eqt:lem_split_fourFlodSum_toSumOfS}
	\sum_{i=1}^k\sum_{i'=1}^k\sum_{j=1}^k\sum_{j'=1}^k r_ir_{i'}r_{j}r_{j'}
		\equiv S_{4} + 4S_{31} + 3S_{22} + 6S_{211} + S_{1111},
\end{align}
where 
\begin{align}
	S_4     &= R_4, 
	\label{chp2-eqt:def_S4}\\
	S_{31}  &= R_3R_1 - R_4,
	\label{chp2-eqt:def_S31}\\
	S_{22}  &= R_2^2 - R_4  ,
	\label{chp2-eqt:def_S22}\\
	S_{211} &=  R_2R_1^2 - 2R_3R_1 - R_2^2 + 2 R_4,
	\label{chp2-eqt:def_S211}\\
	S_{1111} &= R_1^4 - 6 R_2R_1^2 + 8 R_3R_1 + 3R_2^2 - 6R_4.
	\label{chp2-eqt:def_S1111}
\end{align}
\end{lemma}

\begin{proof}[Proof of Lemma~\ref{chp2-lem:splitFourFoldedSum}]
For any $h\in\mathbb{N}$, and $a,b\in\{1,\ldots,h\}$, let $\Omega^h = \{1,\ldots,k\}^h$, and $\Omega^h_{ab}\subset \Omega^h$
such that every $x\in \Omega_{ab}^h$ is a $h$-dimensional vector in $\Omega^h$
such that $x$ has the same $a$th and $b$th components, e.g.,
$\Omega_{12}^3 = \left\{ (i,j,j')\in \Omega^3 :  i=j \right\}.$
Note that 
\begin{align*}
	&\sum_{i=1}^k\sum_{i'=1}^k\sum_{j=1}^k\sum_{j'=1}^k r_ir_{i'}r_{j}r_{j'} \\
	&\qquad= 	\left( \sum_{i=1}^k r_i^4 \right)
		+\left\{ {4\choose 1}\mathop{\sum\sum}_{i\neq j} r_i^3 r_j  \right\}
		+ \left\{ \frac{1}{2}{4\choose 2} \mathop{\sum\sum}_{i\neq j} r_i^2 r_j^2 \right\} \\
	&\qquad\qquad	+ \left\{ {4\choose 2} \mathop{\sum\sum\sum}_{(i,j,j')\in\left(\Omega_{12}^3\right)^{\C}\cap\left(\Omega_{13}^3\right)^{\C}\cap\left(\Omega_{23}^3\right)^{\C}} r_i^2r_jr_{j'} \right\} 
	+\left\{ \mathop{\sum\sum\sum}_{i<i'<j<j'} r_ir_{i'}r_jr_{j'}  \right\} . 
\end{align*}		
Hence, it remains to verify (\ref{chp2-eqt:def_S4})--(\ref{chp2-eqt:def_S1111}).
First, (\ref{chp2-eqt:def_S4}) is trivial. 
For (\ref{chp2-eqt:def_S31}) and (\ref{chp2-eqt:def_S22}), we know 
\begin{eqnarray*}
	S_{31} 
		&=& \mathop{\sum\sum}_{i\neq j} r_i^3 r_j 
		= \left(\sum_{i=1}^m \sum_{j=1}^m r_i^3 r_j \right)
			-\left( \sum_{i=1}^m r_i^4 \right) 	
		=R_3R_1 - R_4 ,\\
	S_{22} 
		&=& \mathop{\sum\sum}_{i\neq j} r_i^2 r_j^2
		= \left(\sum_{i=1}^m \sum_{j=1}^m r_i^2 r_j^2 \right)
			-\left( \sum_{i=1}^m r_i^4 \right) 	
		=R_2R_2 - R_4 .
\end{eqnarray*}
For (\ref{chp2-eqt:def_S211}) and (\ref{chp2-eqt:def_S1111}),
they follow form the inclusion-exclusion formula. 
Define $\oplus$, $\ominus$ and $\otimes$ be 
the addition, subtraction and multiplication in the sense of multisets, 
for example, $\{1,2\}\oplus\{1\}=\{1,1,2\}$, $\{1,1,2\}\ominus\{1,2\}=\{1\}$
and $2\otimes \{1,2\}=\{1,2\}\oplus\{1,2\}$.
Consider (\ref{chp2-eqt:def_S211}) first.
We have
\begin{align}
	&\left(\Omega_{12}^3\right)^{\C}\cap\left(\Omega_{13}^3\right)^{\C}\cap\left(\Omega_{23}^3\right)^{\C}\nonumber\\
		&\qquad= \Omega^3 \ominus \left( \Omega_{12}^3 \cap \Omega_{13}^3 \cap \Omega_{23}^3\right) \nonumber\\
		&\qquad= \Omega^3 \ominus \Omega_{12}^3 \ominus \Omega_{13}^3 \ominus \Omega_{23}^3
				\oplus \left(\Omega_{12}^3 \cap \Omega_{13}^3 \right) \nonumber \\
		& \qquad\qquad		\oplus \left(\Omega_{12}^3 \cap \Omega_{23}^3 \right)
				\oplus \left(\Omega_{13}^3 \cap \Omega_{23}^3 \right)
				\ominus \left(\Omega_{12}^3 \cap \Omega_{13}^3 \cap \Omega_{23}^3 \right) \nonumber\\
		&\qquad= \Omega^3 \ominus \Omega_{12}^3 \ominus \Omega_{13}^3 \ominus \Omega_{23}^3
				\oplus 2 \otimes \left(\Omega_{12}^3 \cap \Omega_{13}^3 \right), \label{chp2-eqt:ExInFormula3}
\end{align}
where $(\cdot)^{\C}$ denotes set complement in the space $\Omega^3$.
Using (\ref{chp2-eqt:ExInFormula3}), we can simplify 
\begin{eqnarray*}
	S_{211} 
		&=& \mathop{\sum\sum\sum}_{(i,j,j')\in\left(\Omega_{12}^3\right)^{\C}\cap\left(\Omega_{13}^3\right)^{\C}\cap\left(\Omega_{23}^3\right)^{\C}} r_i^2r_jr_{j'} \\
		&=& \left(\sum_{(i,j,j')\in \Omega^3} - \sum_{(i,j,j')\in \Omega^3_{12}} - \sum_{(i,j,j')\in \Omega_{13}^3} - \sum_{(i,j,j')\in \Omega_{23}^3}
				+ 2\sum_{(i,j,j')\in\left(\Omega_{12}^3 \cap \Omega_{13}^3 \right)}\right)r_i^2r_jr_{j'} \\
		&=&  R_2R_1^2 - 2R_3R_1 - R_2^2 + 2 R_4.
\end{eqnarray*}
Similarly, for (\ref{chp2-eqt:def_S1111}), 
We have
\begin{align*}
	&		 \left(\Omega_{12}^4\right)^{\C}
		\cap\left(\Omega_{13}^4\right)^{\C}
		\cap\left(\Omega_{14}^4\right)^{\C}
		\cap\left(\Omega_{23}^4\right)^{\C}
		\cap\left(\Omega_{24}^4\right)^{\C}
		\cap\left(\Omega_{34}^4\right)^{\C}
	\nonumber\\
		&\qquad= \Omega^4 
			- \left(\sumPrime_{a_1b_1} \Omega^4_{a_1b_1}\right)
			+ \left(\sumPrime_{a_1b_1,a_2b_2} \Omega^4_{a_1b_1}\cap \Omega^4_{a_2b_2}\right)\nonumber \\
	& \qquad\qquad
			- \cdots+ \left(\sumPrime_{a_1b_1,\ldots,a_6b_6} \Omega^4_{a_1b_1}\cap \cdots \cap \Omega^4_{a_6b_6}\right),
\end{align*}
where 
$(\cdot)^{\C}$ denotes set complement in the space $\Omega^4$; and
the summation 
$
	\sumPrime_{a_1b_1,\ldots,a_jb_j}
$
is summing over all possible combinations of $a_1b_1,\ldots,a_jb_j\in\{12,13,14,23,24,34\}$.
Then, we have 
\begin{eqnarray*}
	S_{1111} 
		&=&  R_1^4 - 6 R_2R_1^2 + 8 R_3 R_1
				+ 3 R_2^2 - 6 R_4.
\end{eqnarray*}
Thus, all results follow. 
\end{proof}

\begin{lemma}\label{chp2-lemma:E_combination_of_L}
Suppose that Conditions \ref{cond-Normality}--\ref{cond-UI} hold. 
Let $m_0=m-1\geq 2$.
Under $H_0$, we have, as $n\rightarrow\infty$, that 
\begin{eqnarray*}
	\E\left( T_1^4 \right) 
		&\rightarrow& 48R_4 + 32R_3R_1 + 12 R_2^2 + 12 R_1^2 R_2 + R_1^4 , \\
	\E\left( T_1^2T_2T_3 \right)
		&\rightarrow& \frac{16(1-2m_0)}{m_0^4}R_4 
						+ \frac{16(m_0-1)}{m_0^3}R_3R_1 
						+ \frac{4(m_0^2+2)}{m_0^4}R_2^2 \\
						&&\qquad+ \frac{2(m_0^2+5)}{m_0^2}R_1^2R_2 
						+ R_1^4, \\
	\E\left( T_1^3T_2 \right) 
		&\rightarrow& \frac{48}{m_0^2}R_4 
						+ \frac{8(m_0^2+3)}{m_0^2}R_3R_1
						+ \frac{12}{m_0^2}R_2^2
						+ \frac{6(m_0^2+1)}{m_0^2}R_1^2R_2
						+ R_1^4, \\
	\E\left( T_1^2T_2^2 \right) 
		&\rightarrow& \frac{16(m_0^2+1)}{m_0^4}R_4 
						+ \frac{32}{m_0^2}R_3R_1
						+ \frac{4(m_0^4+2)}{m_0^4}R_2^2\\
						&&\qquad + \frac{4(m_0^2+2)}{m_0^2}R_1^2 R_2
						+ R_1^4, \\ 
	\E\left( T_1T_2T_3T_4 \right) 
		&\rightarrow& \frac{48m_0^5-48m_0^4+72m_0^2+144m_0+72}{m_0^8(m_0-1)}R_4 
						+ \frac{32(1-m_0)}{m_0^3(m_0-1)}R_3 R_1\\
		&&\qquad				+ \frac{12(m_0-1)}{m_0^4(m_0-1)}R_2^2 
						+ \frac{12(m_0-1)}{m_0^2(m_0-1)}R_1^2 R_2
						+ R_1^4.
\end{eqnarray*}
\end{lemma}

\begin{proof}[Proof of Lemma~\ref{chp2-lemma:E_combination_of_L}]
Recall that, from Proposition \ref{prop:limit_of_J}, we have 
\begin{eqnarray}\label{eqt:weakLimitTell}
	T_{\ell} \inD
	\mathbb{T}_{\ell}
		&=&\frac{m-1}{m}
			\sum_{j=1}^k r_j \left\{ \bar{Z}_{j(\{\ell\})} - \bar{Z}_{j(\{-\ell\})} \right\}^2  \nonumber \\
		&=& \sum_{j=1}^k r_j \left\{ Z_{j\ell} - \bar{Z}_{j\bullet} \right\}^2. 
\end{eqnarray}
For $i,i',j,j'\in\{1,\ldots,k\}$ and $\ell,\ell',h,h'\in\{1,\ldots,m\}$,
define
\begin{eqnarray*}
	\mathcal{E}_{ii'jj'}^{\ell\ell'hh'} := \E\left\{  
					\left( Z_{i\ell}-\overline{Z}_{i\bullet}\right)^2
					\left( Z_{i'\ell'}-\overline{Z}_{i'\bullet}\right)^2
					\left( Z_{jh}-\overline{Z}_{j\bullet}\right)^2
					\left( Z_{j'h'}-\overline{Z}_{j'\bullet}\right)^2
				\right\}.
\end{eqnarray*}
Within this proof, let $G,G_1,G_2,\ldots\simIID\Normal(0,1)$.
Note that 
\begin{itemize}
	\item when $i\neq j$, 
			$Z_{i\ell}-\overline{Z}_{i\bullet}$ and $Z_{jh}-\overline{Z}_{j\bullet}$
			are independent $\Normal(0,m_0/m)$ 
			for any $\ell,h$; and 
	\item when $i=j$, $\Cov\left(Z_{i\ell}-\overline{Z}_{i\bullet},Z_{jh}-\overline{Z}_{j\bullet}\right) = -1/m$ 
			for any $\ell\neq h$.
\end{itemize}
Using these two facts, we can apply Cholesky decomposition to jointly represent 
\begin{align}\label{chp2-eqt:generalChol}
	\left[\begin{array}{c}
		Z_{11}-\overline{Z}_{1\bullet}\\
		Z_{12}-\overline{Z}_{1\bullet}\\
		Z_{13}-\overline{Z}_{1\bullet}\\
		Z_{14}-\overline{Z}_{1\bullet}
	\end{array}
	\right]
	= 
	\left(\frac{m_0}{m}\right)^{1/2}
	\left[\begin{array}{cccc}
		C_{11} & 0 & 0 & 0 \\
		C_{21} & C_{22} & 0 & 0 \\
		C_{31} & C_{32} & C_{33} & 0 \\
		C_{41} & C_{42} & C_{43} & C_{44} 
	\end{array}
	\right]
	\left[\begin{array}{c}
		G_1\\
		G_2\\
		G_3\\
		G_4\\
	\end{array}
	\right],
\end{align}
where $C_{11}=1$, 
\begin{align*}
	C_{22} &= \left\{ \frac{(m_0+1)(m_0-1)}{m_0^2} \right\}^{1/2},&
	C_{33} &= \left\{ \frac{(m_0+1)(m_0-2)}{m_0(m_0-1)} \right\}^{1/2},\\
	C_{44} &= \left\{ \frac{(m_0+1)(m_0-3)}{m_0(m_0-2)} \right\}^{1/2},&
	C_{21} &= C_{31} = C_{41} = - \frac{1}{m_0} ,\\
	C_{32} &= C_{42} = -\frac{m_0+1}{m_0^2(m_0-1)}, &
	C_{43} &= - \frac{m_0+1}{m_0(m_0-1)(m_0-2)}.
\end{align*}
We prove the lemma for $m\geq 5$ first
because the decomposition (\ref{chp2-eqt:generalChol}) is valid when $m\geq 5$.

\underline{\bf(Step 1)} We begin with $\E ( T_1^4  )$. 
Using (\ref{eqt:weakLimitTell}) and Condition \ref{cond-UI}, 
we have 
\begin{eqnarray*}
	 \E\left( T_1^4 \right) 
		\rightarrow \left( \frac{m}{m_0} \right)^4\sum_{i=1}^k\sum_{i'=1}^k\sum_{j=1}^k\sum_{j'=1}^k
			r_i r_{i'} r_j r_{j'}
				\mathcal{E}_{ii'jj'}^{1111}.
\end{eqnarray*}
To find $\mathcal{E}_{ii'jj'}^{1111}$, 
we consider five different cases of values of $(i,i',j,j')$ 
according to the partition rules 
(\ref{chp2-eqt:partition_of_fourFoldSuq_1})--(\ref{chp2-eqt:partition_of_fourFoldSum_5}).
We have 
\[
	\mathcal{E}_{ii'jj'}^{1111}
	= \left\{ 
	\begin{array}{ll}
	\E( G^8) = 105 & \text{if $(i,i',j,j')\in\Xi_4$}; \\
	\E( G^6) \E( G^2) = 15 & \text{if $(i,i',j,j')\in\Xi_{31}$}; \\
	\{ \E( G^4) \}^2= 9& \text{if $(i,i',j,j')\in\Xi_{22}$}; \\
	\E( G^4) \{ \E( G^2)\}^2= 3& \text{if $(i,i',j,j')\in\Xi_{211}$}; \\
	\{ \E( G^2) \}^4= 1& \text{if $(i,i',j,j')\in\Xi_{1111}$}; \\
	\end{array}
	\right.
\]
Applying Lemma \ref{chp2-lem:splitFourFoldedSum}, we have 
\begin{eqnarray*}\label{chp2-eqt:fianl_mathcalA1}
	\E\left( T_1^4 \right)
		&\rightarrow& 
			 (1)(105)S_4 + (4)(15)S_{31} + (3)(9)S_{22} + (6)(3)S_{211} + (1)(1)S_{1111} \nonumber\\
		&=& 48 R_4 + 32 R_3 R_1 + 12 R_2^2 + 12R_2R_1^2 + R_1^4,
\end{eqnarray*}
where $S_4,S_{31},S_{22},S_{211},S_{1111}$ are defined in 
(\ref{chp2-eqt:def_S4})--(\ref{chp2-eqt:def_S1111}), respectively.

\underline{\bf(Step 2)}
Similarly, we can find the limit of $\E ( T_{1}^2 T_{2}^2 )$.
	We begin with deriving $\mathcal{E}_{ii'jj'}^{1122}$. 
	Consider the following cases 
	according to different values of $(i,i',j,j')$.
\begin{itemize}
	\item Case 1: $i=i'=j=j'$.
By 
the  decomposition (\ref{chp2-eqt:generalChol}), we have
\begin{eqnarray*}
	\mathcal{E}_{ii'jj'}^{1122} 
		= \E\left\{ G_1^4 \left( C_{21}G_1 + C_{22}G_2 \right)^4 \right\} 
		=  \frac{3\left( 3m_0^4+24m_0^2+8\right)}{m_0^4}.
\end{eqnarray*}
\item Case 2: $i=i'=j\neq j'$. %or $i=i'=j\neq j'$ or $j=j'=i\neq i'$ or $j=j'=i'\neq i$.
By independence of $Z_{i1} - \overline{Z}_{i\bullet}$ and $Z_{j'2} - \overline{Z}_{j'\bullet}$, 
we have 
\begin{eqnarray*}
	\mathcal{E}_{ii'jj'}^{1122} 
		= \E\left\{ G_1^4 \left( C_{21}G_1 + C_{22}G_2 \right)^2 \right\} \E\left( G_1^2 \right) 
		= \frac{3\left( m_0^2+4 \right)}{m_0^2}.
\end{eqnarray*}
\item Case 3.1: $i=i'\neq j= j'$.
Then $\mathcal{E}_{ii'jj'}^{1122} 
		=  \E^2 ( G_1^4  ) 
		= 9$.

\item Case 3.2: $i=j\neq i'= j'$.
Then 
\[
	\mathcal{E}_{ii'jj'}^{1122} 
		= \E^2\left\{ G_1^2 \left( C_{21}G_1 + C_{22}G_2 \right)^2 \right\} 
		= (m_0^4+4m_0^2+4 )/ m_0^4.
	\]

\item Case 4.1: $i=i'$ but none of $i,j,j'$ are equal.
Then \[\mathcal{E}_{ii'jj'}^{1122} 
		= \E ( G^4  ) \E^2 ( G^2 )
		= 3.\]
\item Case 4.2: $i=j$ but none of $i,i',j'$ are equal.
Then 
\[\mathcal{E}_{ii'jj'}^{1122} 
		= \E\left\{ G_1^2 \left( C_{21}G_1 + C_{22}G_2 \right)^2 \right\} 
			\E^2(  G^2)
		= (m_0^2+2)/m_0^2.\]
\item Case 5: none of $i,i',j,j'$ are equal.
Then $\mathcal{E}_{ii'jj'}^{1122} 
		= \E^4(  G^2 )
		=1$.
\end{itemize}
Applying Lemma \ref{chp2-lem:splitFourFoldedSum} again, we have 
\begin{eqnarray}
	\E\left( T_{1}^2 T_{2}^2\right) 
		&\rightarrow& \sum_{i=1}^k\sum_{i'=1}^k\sum_{j=1}^k\sum_{j'=1}^k
			r_i r_{i'} r_j r_{j'}
				\mathcal{E}_{ii'jj'}^{1122} \nonumber\\
		&&\quad= \frac{1}{m_0^4}
			\bigg[ (9m_0^4+72m_0^2+24)S_4 
					+ 4(3m_0^4+12m_0^2)S_{31}  \nonumber\\
		&&\qquad\qquad			
					+ \left\{ (9m_0^4) + 2(m_0^4+4m_0^2+4)\right\} S_{22} \nonumber\\
		&&\qquad\qquad	+ \left\{ 2(3m_0^4) + 4(m_0^4+2m_0^2)\right\} S_{211} 
					+ (m_0^4)S_{1111}\bigg] \nonumber\\
		&&\quad= \frac{1}{m_0^4}
			\bigg\{ \left( 32m_0^2 + 16      \right) R_4 
					+        32m_0^2           R_3 R_1 
					+ \left(  4m_0^4 + 8      \right) R_2^2 \nonumber \\
		&& \qquad\qquad
					+ \left(  4m_0^4 + 8m_0^2 \right) R_2R_1^2 
					+          m_0^4           R_1^4\bigg\} . 
\end{eqnarray}

%-------------------------------------------------------------------------
\underline{\bf(Step 3)}
Next, we find the limit of $\E\left(T_1^3T_2 \right)$. 
We begin with deriving $\mathcal{E}_{ii'jj'}^{1112}$. 
Consider 
\begin{itemize}
\item Case 1: $i=i'=j=j'$. Then
$$
	\mathcal{E}_{ii'jj'}^{1112} 
		 =\E\left\{ G_1^6 \left( C_{21}G_1 + C_{22}G_2 \right)^2 \right\} 
		= (90+15m_0^2)/m_0^2.
$$
\item
Case 2.1: $i=i'=j\neq j'$.  Then
$
	\mathcal{E}_{ii'jj'}^{1112} 
		= \E(G^6) \E(G^2) 
		= 15.
$
\item
Case 2.2: $i=i'=j'\neq j$.  Then
$$
	\mathcal{E}_{ii'jj'}^{1112} 
		= \E\left\{ G_1^4 \left( C_{21}G_1 + C_{22}G_2 \right)^2 \right\}
		= (3m_0^2+12)/m_0^2.
$$
\item
Case 3: $i=i'\neq j= j'$. Then
$$
	\mathcal{E}_{ii'jj'}^{1112} 
		= \E(G^4)
			\E\left\{ G_1^2 \left( C_{21}G_1 + C_{22}G_2 \right)^2 \right\} 
		= (3m_0^2+6)/m_0^2.	
$$
\item
Case 4.1: $i=i'$ but none of $i,j,j'$ are equal. Then
$$
	\mathcal{E}_{ii'jj'}^{1112} 
		=\E\left( G^4 \right)  \E^2( G^2)
		= 3.
$$
\item
Case 4.2: $i=j'$ but none of $i,i',j$ are equal. Then
$$
	\mathcal{E}_{ii'jj'}^{1112} 
		= \E\left\{ G_1^2 \left( C_{21}G_1 + C_{22}G_2 \right)^2 \right\}\E^2(  G^2) 
		= (m_0^2+2 )/m_0^2.
$$
\item
Case 5: none of $i,i',j,j'$ are equal. Then
$
	\mathcal{E}_{ii'jj'}^{1112} 
		=  \E^4(  G^2)
		= 1.
$
\end{itemize}

Applying Lemma \ref{chp2-lem:splitFourFoldedSum}, we have 
\begin{eqnarray*}
	\E\left( T_{1}^3 T_{2}\right)
		&\rightarrow& \sum_{i=1}^k\sum_{i'=1}^k\sum_{j=1}^k\sum_{j'=1}^k
			r_i r_{i'} r_j r_{j'}
				\mathcal{E}_{ii'jj'}^{1112} \\
		&=&\frac{1}{m_0^2}
			\bigg\{ (90+15m_0^2)S_4 
					+ \left\{15m_0^2+3(12+3m_0^2)\right\} S_{31} 
					+ 3(6+3m_0^2)S_{22} \nonumber\\
		&&\qquad	+ \left\{ 3(3m_0^4) + 3(2+m_0^2)\right\} S_{211} 
					+ (m_0^2)S_{1111}\bigg\} \nonumber\\
		&=&\frac{1}{m_0^2}
			\left\{ 48 R_4 
					+ \left( 8m_0^2+24\right) R_3 R_1 
					+ 12 R_2^2 
					+ \left(  6m_0^2 + 6 \right) R_2R_1^2 
					+          m_0^2           R_1^4\right\} .
\end{eqnarray*}

%-------------------------------------------------------------------------
\underline{\bf(Step 4)}
To find $\E\left( T_1^2T_2T_3 \right)$, we consider 
\begin{itemize}
\item 
Case 1: $i=i'=j=j'$.
Then 
\begin{eqnarray*}
	\mathcal{E}_{ii'jj'}^{1233} 
		&=&\E\left\{ G_1^4 \left( C_{21}G_1 + C_{22}G_2 \right)^2
			\left( C_{31}G_1 + C_{32}G_2 + C_{33}G_3 \right)^2 \right\} \\
		&=& \frac{3(m_0^4 + 10 m_0^2 - 16 m_0 + 8)}{m_0^4}.
\end{eqnarray*}
\item 
Case 2.1: $i=i'=j\neq j'$.  Then
\begin{eqnarray*}
	\mathcal{E}_{ii'jj'}^{1233} 
		&=&\E(G^2)
			\E\left\{ G_1^2 \left( C_{21}G_1 + C_{22}G_2 \right)^2 \left( C_{31}G_1 + C_{32}G_2 + C_{33}G_3 \right)^2 \right\} \\
		&=&\frac{m_0^3 + 6 m_0 - 8}{m_0^3}.
\end{eqnarray*}
\item 
Case 2.2: $i'=j=j'\neq i$.  Then
$$
	\mathcal{E}_{ii'jj'}^{1233} 
		= \E\left( G^2\right)
			\E\left\{ G_1^4 \left( C_{21}G_1 + C_{22}G_2 \right)^2 \right\}
		= (3m_0^2+12)/m_0^2.
$$
\item 
Case 3.1: $i=i'\neq j= j'$. Then
$$
	\mathcal{E}_{ii'jj'}^{1233} 
		= \E(G^4)
			\E\left( G_1^2 \left( C_{21}G_1 + C_{22}G_2 \right)^2 \right) 
		= (3m_0^2+6)/m_0^2.	
$$
\item 
Case 3.2: $i=j\neq i'= j'$. Then
$$
	\mathcal{E}_{ii'jj'}^{1233} 
		= \E^2\left\{ G_1^2 \left( C_{21}G_1 + C_{22}G_2 \right)^2 \right\} 
		= \left(m_0^2+2\right)^2/m_0^4.	
$$
\item 
Case 4.1: $i=i'$ but none of $i,j,j'$ are equal. Then
$$
	\mathcal{E}_{ii'jj'}^{1233} 
		= \E^2(G^2)
				\E\left\{ G_1^2 \left( C_{21}G_1 + C_{22}G_2 \right)^2 \right\}
		= (m_0^2+2)/m_0^2.
$$
\item 
Case 4.2: $j=j'$ but none of $i,i',j$ are equal.  Then
$$
	\mathcal{E}_{ii'jj'}^{1112} 
		= \E(G^4)\E^2(  G^2)
		= 3.
$$
\item 
Case 5: none of $i,i',j,j'$ are equal. Then
$
	\mathcal{E}_{ii'jj'}^{1112} 
		= \E^4(  G^2 ) 
		= 1.
$
\end{itemize}
Applying Lemma \ref{chp2-lem:splitFourFoldedSum}, we have 
\begin{eqnarray*}
	\E\left( T_{1}^2 T_{2} T_3\right)
		&\rightarrow& \sum_{i=1}^k\sum_{i'=1}^k\sum_{j=1}^k\sum_{j'=1}^k
			r_i r_{i'} r_j r_{j'}
				\mathcal{E}_{ii'jj'}^{1233}\\
		&=&\frac{1}{m_0^4}
			\bigg\{ 3(m_0^4+10m_0^2-16m_0+8)S_4 
					+ 2(4m_0^4+18m_0^2-8m_0) S_{31} \nonumber\\
		&&\qquad	+ (5m_0^4+14m_0^2+8) S_{22} 
					+ (8m_0^4 + 10 m_0^2) S_{211} 
					+ m_0^4S_{1111}\bigg\} \nonumber\\
		&=& \frac{1}{m_0^4}
			\bigg\{ (-32m_0+16) R_4 
					+ (16m_0^2-16m_0) R_3 R_1 
					+ (4m_0^2+8) R_2^2 	\nonumber \\
		&&\qquad + (2m_0^4+10m_0^2) R_2R_1^2 
					+   m_0^4    R_1^4\bigg\} .
\end{eqnarray*}

%-------------------------------------------------------------------------
\underline{\bf(Step 5)}
To find $\E\left( T_1T_2T_3T_4 \right)$, we consider
\begin{itemize}
\item 
Case 1: $i=i'=j=j'$. Then
\begin{eqnarray*}
	\mathcal{E}_{ii'jj'}^{1234} 
		&=& \E\left\{ 
			G_1^2 
			\left( \sum_{\ell=1}^2 C_{2\ell}G_{\ell} \right)^2
			\left( \sum_{\ell=1}^3 C_{3\ell}G_{\ell} \right)^2
			\left( \sum_{\ell=1}^4 C_{4\ell}G_{\ell} \right)^2
			 \right\} \\
		&=& \frac{m_0^9 - m_0^8 + 12m_0^7 - 44m_0^6 + 92 m_0^5 - 60 m_0^4 + 72 m_0^2 + 144 m_0 + 72}{m_0^8(m_0-1)}.
\end{eqnarray*}
\item 
Case 2: $i=i'=j\neq j'$.  Then
\begin{eqnarray*}
	\mathcal{E}_{ii'jj'}^{1234} 
		= \E(G^2)
			\E\left\{ G_1^2 
			\left( \sum_{\ell=1}^2 C_{2\ell}G_{\ell} \right)^2
			\left( \sum_{\ell=1}^3 C_{3\ell}G_{\ell} \right)^2 \right\} 
		= \frac{m_0^3+ 6m_0-8}{m_0^3}.
\end{eqnarray*}
\item 
Case 3: $i=i'\neq j= j'$. Then
$$
	\mathcal{E}_{ii'jj'}^{1234} 
		= \E^2 \left\{ G_1^2 \left( C_{21}G_1 + C_{22}G_2 \right)^2 \right\} 
		=  {\left(m_0^2+2\right)^2}/{m_0^4}.
$$
\item 
Case 4: $i=i'$ but none of $i,j,j'$ are equal. Then
$$
	\mathcal{E}_{ii'jj'}^{1233} 
		= \E^2( G^2 )
				\E\left\{ G_1^2 \left( C_{21}G_1 + C_{22}G_2 \right)^2 \right\}
		= (m_0^2+2)/{m_0^2}.
$$
\item 
Case 5: none of $i,i',j,j'$ are equal. Then
$
	\mathcal{E}_{ii'jj'}^{1112} 
		=  \E^4(  G^2)
		= 1.
$
\end{itemize}
Applying Lemma \ref{chp2-lem:splitFourFoldedSum}, we have 
\begin{align*}
	&\E\left( T_{1} T_{2} T_3 T_4\right)\\\
		&\quad\rightarrow \sum_{i=1}^k\sum_{i'=1}^k\sum_{j=1}^k\sum_{j'=1}^k
			r_i r_{i'} r_j r_{j'}
				\mathcal{E}_{ii'jj'}^{1234} \\
		&\quad= \frac{1}{m_0^8(m_0-1)}
			\bigg\{  4(m_0^9-m_0^8+6m_0^7-14m_0^6+8m_0^5) S_{31} \nonumber\\
		&\qquad\qquad+ 3(m_0^9-m_0^8+4m_0^7-4m_0^6+4m_0^5-4m_0^4) S_{22} \nonumber \\
		&\qquad\qquad	+ 6(m_0^9 - m_0^8 + 2m_0^7 - 2m_0^6) S_{211} \nonumber \\
		&\qquad\qquad	+ (m_0^9-m_0^8)S_{1111} \nonumber \\
		&\qquad\qquad	+ (m_0^9 - m_0^8 + 12m_0^7 - 44m_0^6 + 92 m_0^5 - 60 m_0^4 + 72 m_0^2 + 144 m_0 + 72)S_4 		
					\bigg\} \nonumber\\
		&\quad= \frac{1}{m_0^8(m_0-1)}
			\bigg\{ (48m_0^5-48m_0^4+72m_0^2+144m_0+72) R_4 \nonumber \\
		&\qquad\qquad + (-32m_0^6+32m_0^5) R_3 R_1 \nonumber \\
		&\qquad\qquad + (12m_0^5-12m_0^4) R_2^2 
					+ (12m_0^7-12m_0^6) R_2R_1^2 
					+ (m_0^9-m_0^8)    R_1^4\bigg\} .
\end{align*}

Finally, we check that the above results are also valid for $3\leq m\leq 4$. Thus, we complete the proof of the lemma. 
\end{proof}

\end{document}